%% file: discrep.tex
\documentclass[11pt]{article}
\input{paper_macros.tex}

\newcommand{\myline}{\par
  \kern3pt 
  \hrule height 1.5pt
  \kern2pt 
  \hrule height 0.8pt
  \kern3pt 
}

\DeclareMathOperator{\pr}{pr}

\DeclareMathOperator{\disc}{disc}
\DeclareMathOperator{\csub}{\mathbf{C}}

\newcommand{\change}[1] {\textcolor{red}{#1}}

\begin{document}

\title{ Investigating the discrepancy property of de Bruijn sequences}

\author{
Daniel Gabric \ \ \   \
Joe Sawada\footnote{Research supported by the \emph{Natural Sciences and Engineering Research Council of Canada} (NSERC) grant RGPIN-2018-04211. 
}  \\ \\
}

\maketitle

\begin{abstract}
The discrepancy of a binary string refers to the maximum (absolute) difference between the number of ones and the number of zeroes over all possible substrings of the given binary string.  
We provide an investigation of the discrepancy of over a dozen simple constructions of de Bruijn sequences as well as de Bruijn sequences based on linear feedback shift registers whose feedback polynomials are primitive.  Furthermore,  we demonstrate constructions that attain the lower bound of $\Theta(n)$ and a new construction that attains the previously known upper bound of $\Theta(\frac{2^n}{\sqrt{n}})$. 
This extends the work of Cooper and Heitsch~[\emph{Discrete Mathematics}, 310 (2010)].

\end{abstract}

\section{Introduction}

Let $\bB(n)$ denote the set of binary strings of length $n$. 
A \defo{de Bruijn sequence} is a circular string of length $2^n$ that contains every string in $\bB(n)$ as a substring.  Thus, each length-$n$ substring must occur exactly once.  As an example, 
\begin{equation} \label{eq:db6}
000000\underline{1111110111100111010111}000110110100110010110000101010001001
\end{equation}
is a de Bruijn sequence of order $n=6$; it contains each length-$6$ binary string as a substring when viewed circularly.  
There is an extensive literature on de Bruijn sequences motivated in part by their random-like properties.  As articulated by Golomb~\cite{something}, de Bruijn sequences have the following properties: 
\begin{itemize}
\item they are \defo{balanced}, as they contain the same number of $0$s and $1$s;
\item they satisfy a \defo{run property}, as there are an equal number of contiguous runs of $0$s and $1$s of the same length in the sequence;
\item they satisfy a \defo{span-$n$ property}, as they contain every distinct length $n$ binary string as a substring.
\end{itemize}  
From the example in (\ref{eq:db6}) for $n=6$, note that there are exactly $2^{n-1}$ $0$s and $1$s respectively; there are $2^{n-2}$ contiguous runs of $0$s and $1$s respectively; and by definition, it contains every distinct length $n$ binary string as a substring.  

Despite these properties, many de Bruijn sequences display other properties that are far from random.   For instance, consider the greedy prefer-$1$ construction~\cite{martin}.  After starting with an initial seed, successive bits are appended by always trying a $1$ first.  Only if adding a $1$ results in repeating a length-$n$ substring will a $0$ be appended instead. The resulting de Bruijn sequence for $n=6$ is the one obtained in (\ref{eq:db6}) by rotating the initial prefix of $0$s to the suffix.
As one would expect, it has more 1s than 0s at the start of the sequence.    
One measure that accounts for this is the \defo{discrepancy}, which is informally defined to be the maximum absolute difference between the number of $0$s and $1$s in any substring of a given sequence.

To formally define discrepancy, we first introduce some notation. Let $w$ be a binary string. Let $|w|_a$ denote the number of occurrences of the symbol $a$ in $w$. Let $\csub(w)$ denote the set of all substrings of $w$ where we interpret $w$ as a circular string. For example, taking $w=110$ we have that $|w|_1 = 2$ and $\csub(w) = \{\epsilon, 0, 1, 11, 10, 01, 110,101,011\}$. Then the \defo{discrepancy} $\disc(w)$ of $w$  is defined as \[\disc(w) = \max_{u \in \csub(w)}\Big\lvert ~ |u|_1-|u|_0 ~\Big\rvert.\]

The discrepancy of the sequence in (\ref{eq:db6}) is $|17-5|=12$ as witnessed by the underlined substring.  The sequences generated by the prefer-$1$ construction are known to have discrepancy $\Theta(\frac{2^n\log n}{n})$~\cite{cooper} with an exact formulation based on the Fibonacci and Lucas numbers~\cite{cooper1}. In contrast, the expected discrepancy of a random sequence of length $2^n$ is $\Theta(2^{n/2} \sqrt{\log n})$~\cite{cooper}. Letting $\mathcal{D}$ be a de Bruijn sequence of order $n$, $\disc(\mathcal{D})$ is bounded below by $n$ since $\mathcal{D}$ must contain $1^n$ as a substring. In other words, $\disc(\mathcal{D}) \in \Omega(n)$. By putting upper bounds
on character sums of non-linear recurrence sequences, an upper bound
of $\disc(\mathcal{D}) \in O(2^n/\sqrt{n})$ can be obtained; see page 1131 in~\cite{discrep1} for an explicit calculation. One of the main results of this paper is to build on the preliminary work by the authors~\cite{Gabric2017} to demonstrate de Bruijn sequence constructions with discrepancies that obtain these asymptotic lower and upper bounds.


Some applications in pseudo-random bit generation require de Bruijn sequences that do not have large discrepancy.  For example, when used as a carrier signal, a de Bruijn sequence with a large discrepancy causes spectral peaks that could interfere with devices operating at these frequencies~\cite{pseudo1}. 
Similar measures described as ``balance'' and ``uniformity'' are discussed in~\cite{hsieh}.  However, they focus only on $n=2$ and instead vary the size of the alphabet. They explain that de Bruijn sequences with good balance and uniformity are useful in the planning of reaction time experiments~\cite{emerson, sohn}.  De Bruijn sequences with high discrepancy necessarily have poor balance and uniformity. 

In this paper, we extend the work initiated by Cooper and Heitsch~\cite{cooper} providing a more complete analysis of discrepancy for a wide variety of de Bruijn sequence constructions.
 
\begin{enumerate}
\item We evaluate the discrepancies of an additional $12$ simple/interesting de Bruijn sequence constructions up to $n=30$.
\item We evaluate the discrepancies of all de Bruijn sequences obtained from  linear feedback shift registers (LFSRs) based on primitive polynomials up to $n=25$.
\item We demonstrate de Bruijn sequences constructions with discrepancy that attain the asymptotic lower bound of $\Theta(n)$. 
\item We present a new de Bruijn sequence construction with discrepancy that attains the asymptotic upper bound of $\Theta\big(\frac{2^n}{\sqrt{n}}\big)$. 

\end{enumerate}

The remainder of this paper is presented as follows.
In Section~\ref{sec:back} we present background definitions and notation.   In Section~\ref{sec:overview} we present an overview of our experimental results.  They are partitioned into five groups which are further analyzed in Sections~\ref{sec:CCR}, \ref{sec:same}, \ref{sec:pcr},  \ref{sec:weight}, and \ref{sec:lfsr}.
We conclude in Section~\ref{sec:fut} with open problems and future avenues of research.

\subsection{Background and notation} \label{sec:back}

Let $\alpha = a_1a_2\cdots a_n$ be a binary string in $\mathbf{B}(n)$.  
A \defo{feedback function} is a function $f$ that maps $\mathbf{B}(n)$ to $\{0,1\}$.  A \defo{feedback shift register} (FSR) is a function  on $\mathbf{B}(n)$ that maps a string $\alpha$ to $a_2a_3\cdots a_nf(\alpha)$, where $f(\alpha)$ is feedback function.  If $f$ is linear, then an FSR is called a \defo{linear feedback shift register} (LFSR). The following four ``simple'' LFSRs are presented on page 171 in the third edition of the classic work by Golomb~\cite{golomb2017}. We follow their notation letting the operator $\oplus$ represent addition modulo $2$.

\begin{itemize}
    \item The \defo{pure cycling register} (PCR) has feedback function $f(\alpha) = a_1$. 
    \item  The \defo{complemented cycling register} (CCR) has feedback function $f(\alpha) = 1 \oplus a_1$. 
    \item  The \defo{pure summing register} (PSR) has feedback function $f(\alpha) = a_1 \oplus a_2 \oplus \cdots \oplus a_n$.
    \item  The \defo{complemented summing register} (CSR) has feedback function $f(\alpha) = 1 \oplus a_1 \oplus a_2 \oplus \cdots \oplus a_n $.
\end{itemize} 
In addition to these four LFSRs, there is one other LFSR that relates to the de Bruijn sequences under investigation~\cite{simple_same,sala}.
\begin{itemize}
    \item The \defo{pure run-length register} (PRR) has feedback function $f(\alpha) = a_1 \oplus a_2 \oplus a_n$. 
\end{itemize}

Let $0^n$ denote $n$ copies of $0$ concatenated together.
When the feedback function of an LFSR is based on a primitive polynomial (discussed further in Section~\ref{sec:lfsr}), then its corresponding LFSR produces a \defo{maximal-length sequence} or \defo{m-sequence}, which is a de Bruijn sequence without the $0^n$ substring. Adding an additional $0$ before the substring $0^{n-1}1$ in an  m-sequence yields a regular de Bruijn sequence.

For the remainder of the paper, when discussing a specific algorithm for constructing a de Bruin sequence we will put it in bold, e.g., the greedy {\bf Prefer-1} construction.

\subsection{The discrepancy of de Bruijn sequence constructions up to $n=25$} \label{sec:overview}

In Table~\ref{tab:discreps} we present exact discrepancies for $13$ de Bruijn sequence constructions for values of $n$ between $10$ and $25$.   
The results are partitioned into the following four groups based on increasing discrepancy.   A larger table up to $n=30$ is provided in the appendix.
\begin{itemize}
\item[] {\bf   Group 1}: Constructions based on the CCR.
\item[] {\bf   Group 2}: Constructions based on the PRR, including the greedy prefer-same ({\bf Prefer-same}) and prefer-opposite ({\bf Prefer-opposite}) constructions.  
\item[] {\bf   Group 3}: Constructions based on the PCR, including the {\bf Prefer-1} construction.
Table~\ref{tab:discreps} also shows a random entry based on taking the average discrepancy of 10000 randomly generated\footnote{The sequences were generated in \texttt{C} using the \texttt{srand} and \texttt{rand} functions.} sequences of length $2^n$.  
\item[] {\bf   Group 4}: Constructions based on joining smaller weight-range cycles, including one based on the PSR/CSR.
 \end{itemize}
Details about the constructions from each group are presented in their respective upcoming sections.  Implementations for each of these constructions can be found at \footnotesize \url{http://debruijnsequence.org}. \normalsize
Each de Bruijn sequence  can be generated in $O(n)$ time or better per bit using only $O(n)$ space.

\begin{table}[t]
\begin{center}
\scriptsize
\begin{tabular} {c |  cccc | c c c   }
  & \multicolumn{4}{c}{\bf ( Group 1 )}  &  \multicolumn{3}{c}{\bf ( Group 2 ) } \\
$n$ 	& {\bf Huang} & {\bf CCR2} & {\bf CCR3} & {\bf CCR1} & {\bf Pref-same} & {\bf Lex-comp} & {\bf Pref-opposite}   \\ \hline
10 &      12 &       13 &   13   &     16 &    ~  24 &        ~ 24    &    ~ 27  \\
11 &      13 &       14 &  15  &      18 &    ~  29 &        ~ 29    &    ~ 34  	\\
12 &      15 &      16 &  16  &      22 &     ~ 35 &        ~ 35    &    ~ 43  	\\
13 &      16 &      17 &  18  &      23 &    ~ 43 &        ~ 43    &    ~ 52  	\\
14 &      18 &      19 &  20  &      30 &     ~ 48 &        ~ 48    &    ~ 63  	\\
15 &      19 &      21 &  21  &      29 &    ~ 59 &        ~ 59    &    ~ 74  	\\
16 &      21 &      22 &  23  &      36 &    ~  68 &        ~ 68    &    ~ 87 \\
17 &      22 &      24 &  25  &      37 &     ~ 79 &        ~ 79    &   100	\\
18 &      24 &      26 &  26  &      43 &     ~ 88 &         ~ 88   &   115  	\\
19 &      25 &      27 &  28  &      43 &     103 &       103  &    130	\\
20 &      27 &      29 &  30  &      52 &     114 &       114   &    147	\\
21 &      28 &      31 &  31  &      50 &     127 &       127  &    164	\\
22 &      30 &      32 &  33  &      59 &     142 &       142  &    183	\\
23 &      31 &      34 &  35  &      59 &     155 &       155  &    202	\\
24 &      33 &      36 &  36  &      67 &     172 &       172  &    223	\\
25 &      35 &      37 &  38  &      66 &     187 &       187  &    244	 \\
\end{tabular}
\end{center}

\medskip

\begin{center}
\scriptsize
\begin{tabular} {c |  r r r r r | r r}
  & \multicolumn{5}{c}{\bf ( Group  3 )}  &  \multicolumn{2}{c}{\bf ( Group 4 ) } \\
$n$ 	&  {\bf PCR4} &  \bblue{Random}  &  {\bf PCR3}  & {\bf PCR2}   & {\bf  Prefer-1/PCR1}  & {\bf Cool-lex}  & {\bf Weight-range}  \\ \hline
10 &       29     &  \blue{50} 	&      75           	&     101   &     120   &     131	    &     131	\\
11 &      41      &   \blue{71} 	&     141    	 &     180 &     222  &     257	&     257	\\
12 &     51       &  \blue{101} 	&     248   		&     321 &     416  &     468 	&     468	\\
13 &     70       &  \blue{143}	&     468   		&     587 &     784  &     801 	&     930	\\
14 &     85       &  \blue{203} 	&     850   		&    1065 &    1488  &    1723 	&    1723	\\
15 &    110      &   \blue{288}  	&    1604   	&    1974 &    2824  &    3439 	&    3439	\\
16 &    175      &  \blue{407}  	&    2965   	&    3632 &    5376  &    6443 	&    6443	\\
17 &    246      &  \blue{575} 	&    5594   	&    6785 &   10229  &   11452 	&   12878	\\
18 &    326      &   \blue{815} 	&   10461    	&   12635 &   19484  &   24319 	&   24319	\\
19 &     462     &  \blue{1157} 	&   19765   	&   23746 &   37107  &   48629	&   48629	 \\
20 &       730   &  \blue{1634}   	&   37243  	&   44585 &   71250  &   92388 		&   92388	\\
21 &     954     &  \blue{2311} 	&   70575   	&   84270 &  138332  &  167975 	&  184766		\\
22 &    1327    &   \blue{3264} 	&  133737  	&  159281 &  268582  &  352727 	&  352727		\\
23 &    1820    &    \blue{4565}	&  254322 	&  302449 &  521553  &  705443	&  705443		 \\
24 &    2684    &    \blue{6252}	&  484172 	&  574819 & 1012795  & 1352090	& 1352090 		\\
25 &    3183    &   \blue{9192} 	&  924071 	& 1096009 & 1966813  & 2496163	& 2704168		 \\
\end{tabular}
\end{center}
\caption{Discrepancies of de Bruijn sequence constructions of order $n$ ordered by increasing discrepancy and partitioned into four groups.  }
\label{tab:discreps}
\end{table}

\newpage
We also consider a fifth group of de Bruijn sequences, where each sequence corresponds to a unique primitive polynomial.
\begin{itemize}
\item[] {\bf   Group 5}: Constructions based on primitive polynomials and their corresponding LFSRs.
 \end{itemize}
By generating all primitive polynomials of degree $n$ and their corresponding LFSRs, we compute the minimum, the maximum, and average discrepancies for their corresponding de Bruijn sequences in 
 Table~\ref{tab:lfsr}.  Any related m-sequence known to be generated by such an LFSR can be completely determined after reading only $2n$ bits~\cite{massey}.  Thus, their application for generating pseudo-random numbers is limited.  Further discussion is given in Section~\ref{sec:lfsr}.

\begin{table}[t]
\begin{center}
\scriptsize
\begin{tabular} {c |  r r r r  | r}
  & \multicolumn{4}{c}{\bf ( Group  5 )}  \\
$n$ 	&  {\bf min} &  {\bf avg} &  \bblue{Random}  &   {\bf max} &  {\bf LFSRs}  \\ \hline
10 &       36       & 41  &  \blue{50} 	            &    46     & 60\\
11 &      51      & 58 &   \blue{71}     	        &    68     & 176 \\
12 &     72       & 84 &  \blue{101}    	        &    106    & 144\\
13 &     97       &118 &  \blue{143}	      	    &     144   & 630\\
14 &     141      &167 &  \blue{203}    	        &    206    & 756\\ 
15 &    200      & 236 &   \blue{288}     	        &    294    &  1800\\ 
16 &    280      & 335 &  \blue{407}     	        &    432    & 2048\\ 
17 &    391      & 473 &  \blue{575} 	      	    &    625    & 7710\\ 
18 &    544      & 669 &   \blue{815}     	        &   860     & 7776 \\   
19 &     775     & 947 &  \blue{1157}    	        &   1262    & 27594 \\   
20 &      1066   & 1341  &  \blue{1634}     	    &   1842    & 24000\\ 
21 &     1500    & 1896  &  \blue{2311}    	        &   2619    & 84672\\  
22 &    2128     & 2681 &   \blue{3264}   	        &  3634     & 120032\\ 
23 &    3009     & 3793 &     \blue{4565} 	        &  5326     & 356960 \\ 
24 &    4236    &  5362 & \blue{6252}   	        &  7545     & 276480\\ 
25 &    5905    &  7586 & \blue{9192}  	                & 11291         & 1296000\\  
\end{tabular}
\end{center}

\caption{The minimum, average, and maximum discrepancies of de Bruijn sequence constructions of order $n$ based on primitive polynomials and their corresponding LFSRs.  The number of LFSRs based on primitive polynomials is given in the final column.}
\label{tab:lfsr}
\end{table}

\subsection{Computing the discrepancy of a de Bruijn sequence}

Given a de Bruijn sequence of order $n$, how quickly can the discrepancy be calculated? De Bruijn sequences are exponentially long with respect to their order. Thus, it is natural to try to find a fast algorithm to compute the discrepancy of de Bruijn sequences.

A na\"{i}ve way to calculate the discrepancy of a de Bruijn sequence $\mathcal{D}$ is to consider every substring $u$ of $\mathcal{D}$ and compute $\Big\lvert~|u|_1-|u|_0 ~\Big\rvert$, keeping track of the maximum value. In a circular string of length $m$, there are $m$ substrings of length $1$, $m$ substrings of length $2$, and more generally $m$ substrings of length $j$. Therefore, there are $\Theta(m^2)$ substrings in a length-$m$ circular string. For every substring $u$ that the algorithm visits, the absolute difference between the number of $0$s and the number of $1$s in $u$ is computed, which takes $\Theta(|u|)$ time. Thus, when given a de Bruijn sequence of order $n$ as input, this algorithm runs in $\Theta(2^{3n})$ time.
A slight upgrade from this na\"{i}ve approach is obtained by observing that every substring of a circular string $w$ is  a prefix of a rotation of $w$. For every rotation of $w$, we scan from left to right one bit at a time while keeping track of the number of $0$s, the number of $1$s, and the maximum absolute difference between them. This algorithm runs in $\Theta(2^{2n})$ time when given a de Bruijn sequence of order $n$ as input.

We show that the discrepancy of a de Bruijn sequence of order $n$ can be calculated in $\Theta(2^n)$ time. First, we define some notation. Let $w$ be a binary string and let
\[d_0(w) = \max_{w=uv} \big( |u|_0-|u|_1 \big)\] and \[d_1(w) = \max_{w=uv}\big( |u|_1-|u|_0 \big).\] In other words, $d_0(w)$ (resp., $d_1(w)$) denotes the maximum difference between the numbers of $0$s and $1$s (resp., $1$s and $0$s) in any prefix of $w$.

\begin{lemma}\label{lemma:substring}
Let $\mathcal{D}$ be a de Bruijn sequence of order $n$. There exists a string $y$ such that $\mathcal{D}=xyz$ and $\disc(\mathcal{D})=\Big\lvert~ |y|_1-|y|_0~\Big\rvert$.
\end{lemma}
\begin{proof}
Suppose we cannot write $\mathcal{D}=xyz$ where $\disc(\mathcal{D})=\Big\lvert~ |y|_1-|y|_0~\Big\rvert$ for any strings $x,y,z$. Then we must have $\disc(\mathcal{D})= \Big\lvert~ |zx|_1-|zx|_0~\Big\rvert$ for some choice of $x,z$. However, since $\mathcal{D}$ is a de Bruijn sequence, it must contain the same number of $1$s as $0$s. Thus, $\disc(\mathcal{D}) = \Big\lvert~ |zx|_1-|zx|_0~\Big\rvert = \Big\lvert~ |y|_1-|y|_0~\Big\rvert$, which contradicts our initial assumption.
\end{proof}

\begin{lemma}\label{lemma:linearAlgorithm}
Let $\mathcal{D}$ be a de Bruijn sequence of order $n$. Then $\disc(\mathcal{D})= d_0(\mathcal{D}) + d_1(\mathcal{D})$.
\end{lemma}
\begin{proof}
Let $\gamma_i$ denote the length-$i$ prefix of $\mathcal{D}$. Let $j_0,j_1,\ldots, j_{2^n}$ denote a sequence of integers such that $j_0=0$ and $j_i =  |\gamma_i|_1-|\gamma_i|_0$ for all $i\in \{1,2,\ldots, 2^n\}$. Since $\mathcal{D}$ is a de Bruijn sequence, it has an equal number of $1$s and $0$s. Thus, $j_{2^n} = 0$.
 By Lemma~\ref{lemma:substring}, we can write $\mathcal{D}=xyz$ for strings $x,y,z$ such that $\disc(\mathcal{D})=\Big\lvert~ |y|_1-|y|_0~\Big\rvert$. The number of $1$s in $y$ is equal to the number of $1$s in $xy$ minus the number of $1$s in $x$. The same is true for the number of $0$s in $y$. 
 Therefore,  $|y|_1-|y|_0 = j_{|xy|}- j_{|x|}$. 
 The value  $\Big\lvert~  j_{|xy|}-j_{|x|}~\Big\rvert$  is maximized when either $j_{|x|}$ is as large as possible and $j_{|xy|}$ is as small as possible, or when $j_{|x|}$ is as small as possible and $j_{|xy|}$ is as large as possible.
 In the former case, the value corresponds to $d_1(\mathcal{D}) - (- d_0(\mathcal{D}))$, and in the latter case the
 value corresponds to $|-d_0(\mathcal{D}) - d_1(\mathcal{D})|$. In both cases the value simplifies to
  $d_1(\mathcal{D}) + d_0(\mathcal{D})$.
 %
\end{proof}
\begin{corollary}
Let $\mathcal{D}$ be a de Bruijn sequence of order $n$. The discrepancy of $\mathcal{D}$ can be calculated in $\Theta(|\mathcal{D}|)$ time.
\end{corollary}


\section{Group 1: CCR-based constructions}  \label{sec:CCR}

In this section we consider the four de Bruijn sequence constructions in Group 1 based on the CCR.     The sequences generated by the constructions {\bf CCR1}, {\bf CCR2}, and {\bf CCR3} are based on shift-rules presented in~\cite{binframework}.   The sequences generated by the {\bf CCR2} and {\bf CCR3} constructions can also be constructed by  concatenation approaches~\cite{GS18} described later in this section; the equivalence of the shift-rules to their respective concatenation constructions  has been confirmed up to $n=30$, though no formal proof has been given.  The {\bf Huang} construction is a shift-rule based construction in~\cite{huang}.    Since every de Bruijn sequence of order $n$ contains the substring $0^n$, a lower bound on discrepancy is clearly $n$.  In this section we prove that two aforementioned concatenation based constructions have discrepancy at most $2n$, and thus attain the smallest possible asymptotic discrepancy of $\Theta(n)$.   

To get a better feel for these four de Bruijn sequence constructions, the following graphs illustrate the running difference between the number of 1s and the number of 0s in each prefix of the given de Bruijn sequence.  The examples are for $n=10$, so the de Bruijn sequences have length $2^{10} = 1024$.

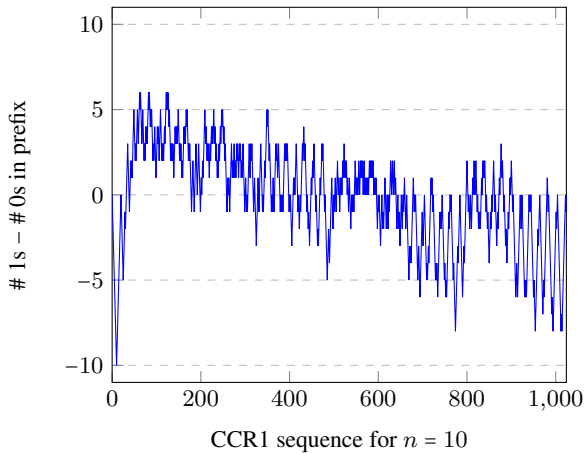
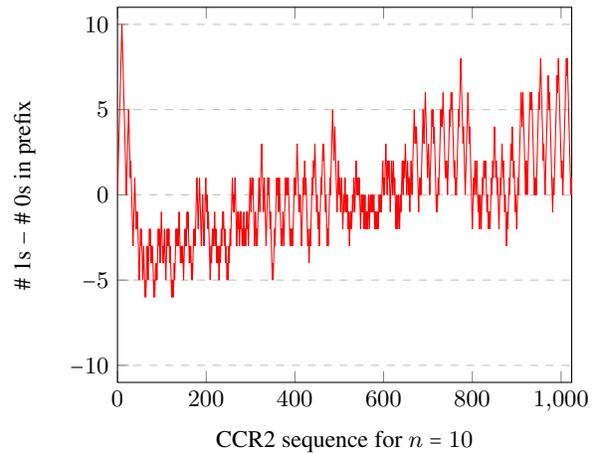
\input{tables-CCR.tex}

 Recall that the CCR is a feedback shift register with feedback function $f(a_1a_2\cdots a_n) = a_1 + 1 \pmod{2}$. 
{$\bB(n)$  is partitioned into equivalence classes of strings, called \defo{co-necklaces}, by the orbits of $f$.
For example, the following four columns are the co-necklace equivalence classes for $n=5$:
\begin{center}
\begin{tabular}{c@{\hspace{4em}}c@{\hspace{4em}}c@{\hspace{4em}}c}
\bblue{00000} & \bblue{00010}  & \bblue{00100} & \bblue{01010} \\
00001 & 00101 &01001  &\underline{10101} \\
00011 & 01011 & \underline{10011}  \\
00111 & \underline{10111} & 00110 \\
01111 & 01110 & 01101 \\
\underline{11111} & 11101 &11011 \\
11110 & 11010 & 10110 \\
11100 & 10100  &01100 \\
11000 & 01000  &11001 \\
10000 & 10001 & 10010 \\
\end{tabular}
\end{center}

 The \defo{periodic reduction} of a string $\alpha$, denoted by $\pr(\alpha)$, is the shortest prefix $\beta$ of $\alpha$ such that $\alpha = \beta^j$ for some $j\geq 1$. 
In~\cite{GS18}, the following two de Bruijn sequence constructions {\bf CCR2} and {\bf CCR3} concatenate the periodic reductions of $\alpha\overline{\alpha}$ for  given representatives $\alpha$ of each co-necklace equivalence class.
\begin{boxed2} \noindent
{\bf CCR2} 
\begin{enumerate}
\item Let the representative for each co-necklace equivalence class of order $n$ be its lexicographically smallest string.
\item Let  $\alpha_1, \alpha_2, \ldots , \alpha_m$ denote these representatives in colex order.
\item {\bf Output}: $\pr(\alpha_1\overline{\alpha_1}) \cdot  \pr(\alpha_2\overline{\alpha_2}) \   \cdots  \  \pr(\alpha_m\overline{\alpha_m})$.
\end{enumerate}
\end{boxed2} 
\noindent
For $n=5$,  the representatives for this algorithm are the bolded strings in the equivalence classes above and {\bf CCR2} produces:
\[ 0000011111 \cdot  0010011011 \cdot  0001011101 \cdot  01. \]
%

\begin{boxed2} \noindent
{\bf CCR3} 
\begin{enumerate}
\item Let the representative for each co-necklace equivalence class of order $n$ be the string obtained by taking the lexicographically smallest string, removing its largest prefix of the form $0^j$, and then appending $1^j$ to the end.
\item Let  $\alpha_1, \alpha_2, \ldots , \alpha_m$ denote these representatives in lexicographic order.
\item {\bf Output}: $\pr(\alpha_1\overline{\alpha_1}) \cdot  \pr(\alpha_2\overline{\alpha_2}) \    \cdots \   \pr(\alpha_m\overline{\alpha_m})$.
\end{enumerate}
\end{boxed2} 
\noindent
For $n=5$, the representatives for this algorithm are the underlined strings in the equivalence classes  above and {\bf CCR3} produces:

\[ 1001101100  \cdot 10 \cdot 1011101000 \cdot  1111100000. \]

We now prove that the discrepancy resulting from these two de Bruijn sequence constructions is at most $2n$.

\begin{lemma} \label{lem:dis}
Consider a sequence of binary strings $\alpha_1, \alpha_2, \ldots , \alpha_m$ where each $\alpha_i$ has the same number of $0$s as $1$s and has discrepancy at most $n$.
Then $\disc( \alpha_1\alpha_2 \cdots \alpha_m)\leq 2n$.
\end{lemma}
\begin{proof}
Let $\mathcal{S} = \alpha_1\alpha_2 \cdots \alpha_m$. By Lemma~\ref{lemma:substring}  there exists a shortest substring $y= u\alpha_{i+1} \alpha_{i+2}\cdots \alpha_{j-1} v$ of $\mathcal{S}$ such that $\disc(\mathcal{S}) = \Big\lvert~ |y|_1-|y|_0~\Big\rvert$ where $u$ is a suffix of $\alpha_i$ and $v$ is a prefix of $\alpha_j$. Since the number of $0$s and $1$s is the same in each $\alpha_k$ and $\disc(\alpha_k) \leq n$,  $\Big\lvert~ |y|_1-|y|_0~\Big\rvert = \Big\lvert~ |uv|_1-|uv|_0~\Big\rvert \leq 2n$.
\end{proof}

\begin{theorem}
The de Bruijn sequences constructed by {\bf CCR2}  and  {\bf CCR3} have discrepancy at most $2n$.
\end{theorem}
\begin{proof}
Given a length $n$ binary string $\alpha$,  $\alpha\overline{\alpha}$ has the same number of $0$s and $1$s and has discrepancy at most $n$.  These properties also hold for $\pr(\alpha\overline{\alpha})$ by definition of the periodic reduction. Thus, by Lemma~\ref{lem:dis},  the sequences constructed by {\bf CCR2}  and  {\bf CCR3} have discrepancy at most $2n$.
\end{proof}

Interestingly, from Table~\ref{tab:discreps},  these two concatenation-based constructions  do not demonstrate the smallest discrepancy for $n \leq 30$.  The construction by Huang~\cite{huang}, which is based on a cycle-joining approach,  demonstrates slightly smaller discrepancy.  In particular the author states:
\begin{quote} \emph{``It seems clear that the sequences produced by
our algorithm have a relatively good characteristic of local 0-1 balance in
comparison with the ones produced by the `prefer one' algorithm.''}
\end{quote}
So the author indicates that their construction may have small discrepancy, however no analysis is provided.  

\normalsize

\section{Group 2: PRR-based constructions}  \label{sec:same}

In this section we consider the three de Bruijn sequence constructions in Group 2 based on the PRR.  The {\bf Pref-same}~\cite{eldert,fred-nfsr,simple_same} and the {\bf Pref-opposite}~\cite{pref-opposite} are greedy constructions based on the last bit of the sequence as it is constructed.  They have the downside of requiring an exponential amount of memory.  The {\bf Lex-comp} construction is obtained by concatenating lexicographic compositions. It was an attempt to generate the sequence generated by the {\bf Pref-same} construction without using exponential space and it the resulting sequences were conjectured to be the same for a very long prefix~\cite{lex-comp}.  In fact, it attains the same discrepancy as the {\bf Pref-same} for all values of $n$ tested. 
Recently, it was demonstrated that the {\bf Pref-same} and the {\bf Pref-opposite} sequences can be generated in $O(n)$ time per bit using only $O(n)$ space by applying the PRR~\cite{sala}.  There is also a PRR-based construction that produces an equivalent sequence as {\bf Lex-comp} for large $n$, but there is no formal proof showing they are equivalent.  

To get a better feel for the two greedy de Bruijn sequence constructions, the following graphs illustrate the running difference between the number of 1s and the number of 0s in each prefix of the given de Bruijn sequence.  The examples are for $n=10$, so the de Bruijn sequences have length $2^{10} = 1024$.   

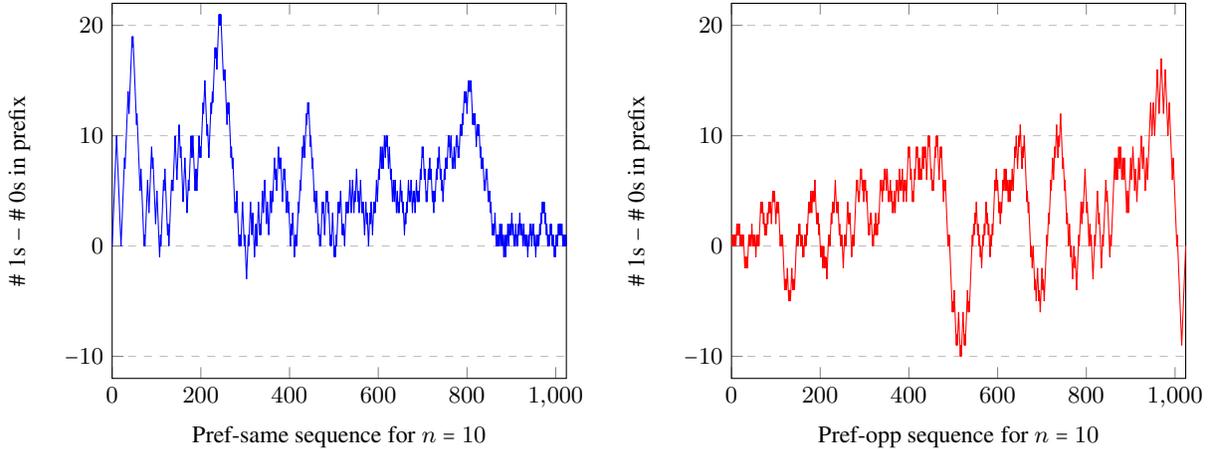
\input{tables-same.tex}


In the following table we study some experimental results for the {\bf Pref-same} construction.   In particular,  for $10 \leq n \leq 25$ we compute the maximum difference between the number of $1$s and the number of $0$s along with the maximum difference between the number $0$s and the number of $1$s,  over all prefixes of each {\bf Pref-same} de Bruijn sequence of order $n$.  Adding these two values together, we get the discrepancies shown in Table~\ref{tab:discreps}.

\medskip

\noindent \small
\begin{tabular}{c| rrrrrrrrrrrrrrrr} 
  $n$ 			&  10 & 11 & 12 & 13 & 14 & 15 & 16 & 17 & 18 & 19 & 20 & 21 & 22 & 23 & 24 & 25 \\ \hline 
  $max$(\#1s $-$ \#0s) 	& 21 &  26 & 31  & 36  & 43  & 50  & 57  & 64  & 73  & 82  & 91  & 100  & 111  & 122  & 133  & 144   \\  
  $max$(\#0s $-$ \#1s) 	& 3  & 3   &  4   & 7   & 5   & 9   & 11   & 15   & 15   & 21   & 23   & 27   & 31   & 33   & 39   & 43  \\ \hline
  discrepancy    			& 24 & 29 & 35 & 43 & 48 & 59 & 68 & 79 & 88 & 103 & 114 & 127 & 142 & 155 & 172 & 187 
\end{tabular} \normalsize

\medskip

\noindent
Interestingly, the values in the row $max$(\#$1$s $-$ \#$0$s) are equivalent to the known sequence A008811 in the Online Encyclopedia of Integer Sequences (OEIS)~\cite{OEIS} offset by four positions.   The sequence  enumerates the ``Expansion of $x(1+x^4)/((1-x)^2(1-x^4))$'' and the provided formula demonstrates that each value is $\Theta(n^2)$.
More specifically the values  match  the sequence for $6 \leq n \leq 30$, though we have no intuition as to why this is the case. This leads to the following conjecture.

\begin{conjecture} \label{con:same}
The de Bruijn sequences constructed by the {\bf Pref-same} and {\bf Lex-comp} algorithms have discrepancy $\Theta(n^2)$.
\end{conjecture}

A similar analysis was performed for sequences generated by the {\bf Pref-opposite} construction.

\medskip

\noindent  \small
\begin{tabular}{c| rrrrrrrrrrrrrrrr} 
  $n$ 			&  10 & 11 & 12 & 13 & 14 & 15 & 16 & 17 & 18 & 19 & 20 & 21 & 22 & 23 & 24 & 25 \\ \hline 
  $max$(\#1s $-$ \#0s) 	&  10  & 13   & 17   & 21   & 26   & 31   & 37   & 43   & 50   & 57   &  65   & 73   &  82   &  91   & 101   & 111    \\ 
   $max$(\#0s $-$ \#1s) 	&  17   & 21   & 26   & 31   & 37   & 43   & 50   & 57   &  65   & 73   &  82   &  91   & 101   & 111 & 122 & 133     \\ \hline
   discrepancy           &  27 & 34 & 43 & 52 & 63 & 74 & 87 & 100 & 115 & 130 & 147 & 164& 183 & 202& 223 & 244
\end{tabular} \normalsize

  \medskip

\noindent
Remarkably, observe that the two middle rows are a shift from each other by two positions.  Just as interesting, the sequences also correspond to a known sequence in OEIS~\cite{OEIS}, namely A033638.  Specifically, the row $max$(\#$1$s$-$ \#$0$s)  corresponds to this sequence shifted by four positions.  The sequence does not match for $n<10$, but we have verified it matches for $10 \leq n \leq 30$.  The sequence corresponds to ``quarter squares plus $1$'', and by applying the appropriate shifts, the discrepancy for the {\bf Prefer-opposite} sequence of order $n$, for $10 \leq n \leq 30$ is given by:
\[ \bigg\lfloor \frac{(n-4)^2}{4} \bigg\rfloor + \bigg\lfloor \frac{(n-2)^2}{4} \bigg\rfloor + 2.\]
This leads to the following conjecture. 
 
\begin{conjecture} \label{con:opp}
The de Bruijn sequence constructed by the {\bf Pref-opposite} algorithm has discrepancy $\Theta(n^2)$.
\end{conjecture}

We conclude this section with an observation regarding the {\bf Pref-opposite} de Bruijn sequence:  For $2 \leq n \leq 25$, each sequence has the following suffix where $j=\lceil n/3 \rceil$:
\[ 0^j1^{n-j} \cdot 0^{j-1}1^{n-j+1} \ \cdots \ 01^{n-1} \cdot 10^{n-1}. \] 
For example, when $n=10$, the {\bf Pref-opposite} de Bruijn sequence has suffix 
\[ 0000001111 \cdot 00000\underline{11111 \cdot 0000111111 \cdot 0001111111 \cdot 0011111111 \cdot 0111111111 \cdot 1}000000000, \]
and the underlined substring has $5+6+7+8+10$ ones and $4+3+2+1$ zeros.  A slight rearrangement gives a lower bound of  $(5-1)+ (6-2) + (7-3) + (8-4) + 10 = 4\cdot 4 + 10=  26$ for the discrepancy of the sequence.  The actual discrepancy is 27.
More generally, if this suffix is indeed a suffix for each {\bf Pref-opposite} de Bruijn sequence, then a lower bound on its discrepancy will be
\[  (\lceil n/2 \rceil -1)  (\lfloor n/2 \rfloor -1) + n \  = \  \Omega(n^2).\]

\section{Group 3: PCR-based constructions} \label{sec:pcr}

In this section we consider the four de Bruijn sequence constructions in Group 3 based on the PCR.   The constructions {\bf PCR1}, {\bf PCR2}, {\bf PCR3}, and {\bf PCR4} are based on shift-rules presented in~\cite{binframework}. Like the other shift-rule constructions, these four rules result from joining smaller cycles based on the underlying feedback function; depending on how the ``bridge states'' are defined leads to the different shift-rules. 
The sequences generated by {\bf PCR1}  are the same as the ones generated by the prefer-$0$ greedy construction; they are the complements of the sequences generated by {\bf PCR1}, and so they have the same. The sequences generated by {\bf PCR1} can also be generated by two necklace concatenation constructions, one based on lexicographic order~\cite{fkm2}, and another taking a recursive approach~\cite{Ralston:1981}.

The sequences generated by {\bf PCR2}  are the same as the ones generated by a necklace concatenation construction based on colex order~\cite{grandma,grandma2}.  The {\bf PCR3} is based on a general approach in~\cite{jansen} and revisited in~\cite{wong}.  

To get a better feel for these four de Bruijn sequence constructions, the following graphs illustrate the running difference between the number of $1$s and the number of $0$s in each prefix of the given de Bruijn sequence.  The examples are for $n=10$, so the de Bruijn sequences have length $2^{10} = 1024$.

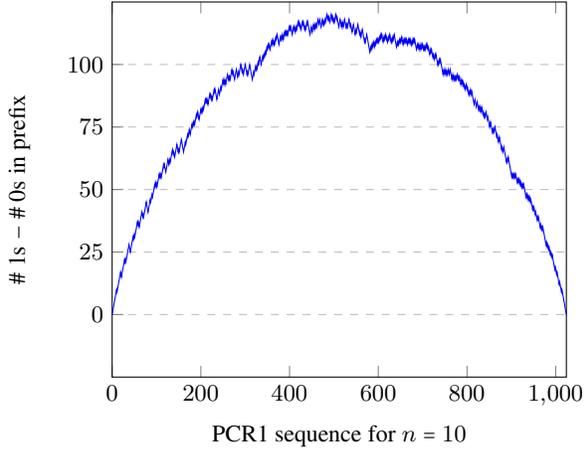
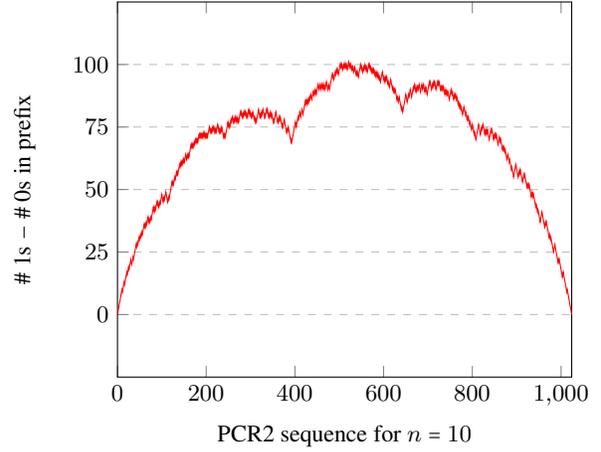
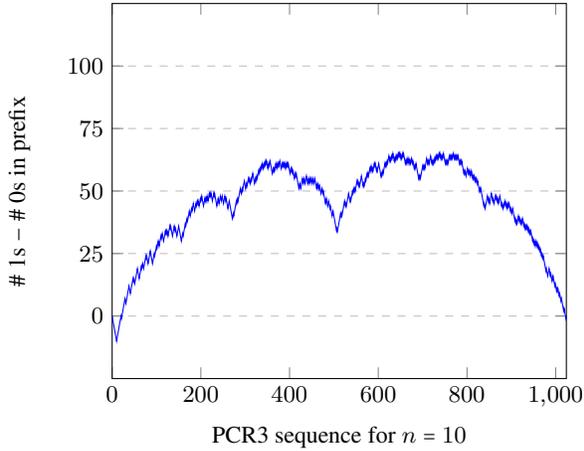
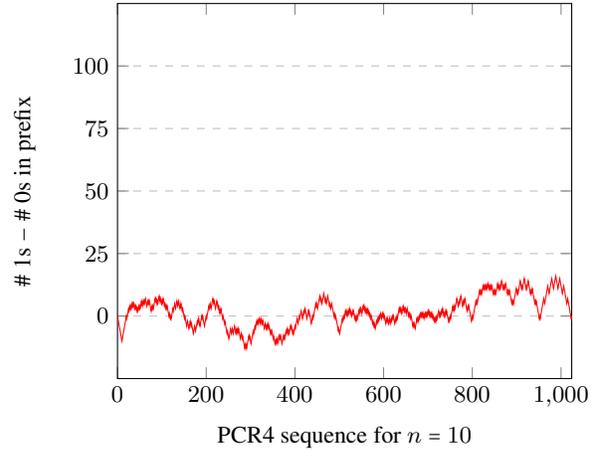
\input{tables-PCR.tex}

The discrepancy for the sequence generated by the {\bf PCR1} construction has already been studied in~\cite{cooper} where they show that the discrepancy is $\Theta(\frac{2^n\log n}{n})$.  The sequences generated by the {\bf PCR2} and {\bf PCR3}  constructions appear to have a similar growth trajectories.   More interesting are the sequences generated by the {\bf PCR4} construction that, from Table~\ref{tab:discreps}, appear to have discrepancy that is closest to that of a random string.  It would be interesting to do a more detailed investigation of this construction, which is based on a very simple successor rule.

\section{Group 4: Weight range constructions and the PSR/CSR} \label{sec:weight}


In this section we consider two de Bruijn sequence constructions that join smaller cycles based on weight (number of $1$s).  In some related works the term density is also used to mean weight, so will use the variable $d$ to indicate a weight.
The {\bf Cool-lex} construction~\cite{bubble3}, is a concatenation approach which is based on creating underlying cycles which contain all strings with weights $d$ and $d+1$ given $0 \leq d < n$.  Then, appropriate such cycles can be joined together to obtain a de Bruijn sequence~\cite{dbrange}.  By the nature of how the cycles are joined, most length-$n$ substrings in the first half of the resulting de Bruijn sequence have weight less than or equal to $n/2$.  Similarly, most length-$n$ substrings in the latter half of the sequence have weight greater than or equal to $n/2$.  Thus, as one would expect, the resulting de Bruijn sequence has a very large discrepancy.  The  {\bf Weight-range} construction is a new construction presented in this section.  Its resulting de Bruijn sequence
has discrepancy that attains the asymptotic upper bound of $\Theta(2^n/\sqrt{n})$.

To get a better feel for these two de Bruijn sequence constructions, the following graphs illustrate the running difference between the number of $1$s and the number of $0$s in each prefix of the given de Bruijn sequence.  The examples are for $n=10$, so the de Bruijn sequences have length $2^{10} = 1024$.

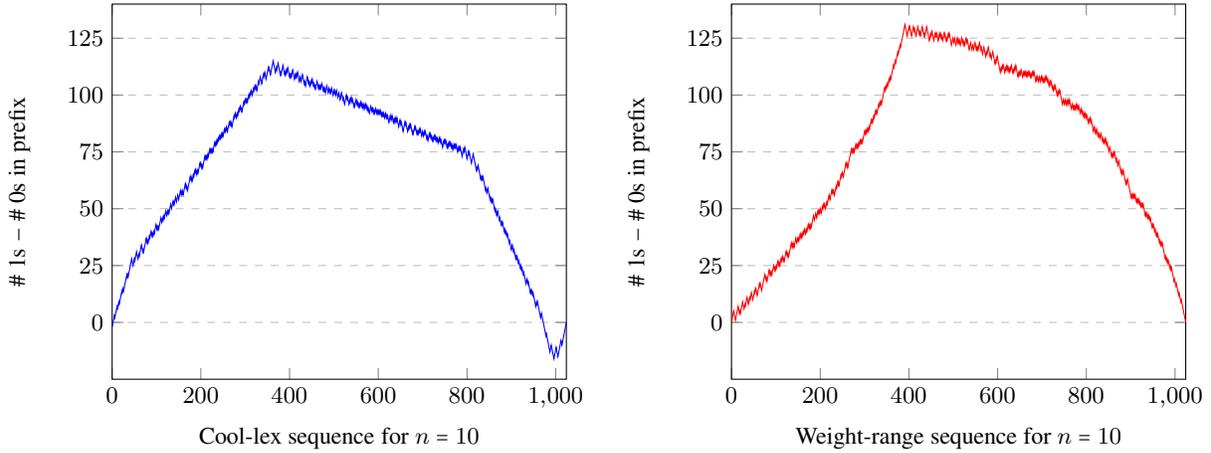
\input{tables-range.tex}

Notice that if we had shifted the starting position of the {\bf Cool-lex} sequence the profile of the graph would be very similar to that the {\bf Weight-range} sequence.  In fact, the discrepancies of the two sequences are the same except when $n \bmod 4 \equiv 1$ (see Table~\ref{tab:discreps}).  This will be discussed more after we present the {\bf Weight-range} construction.

A \defo{minimum weight de Bruijn sequence}  is a cyclic sequence that contains each binary string of length $n$ with weight at least $d$ exactly once.   A \defo{maximum weight de Bruijn sequence} is defined similarly where the weight of each string is at most $d$.  A construction for the former sequence is given in~\cite{min-weight}; it is constructed by concatenating the periodic reduction of each necklace of weight $\geq d$ when the necklaces are listed in lexicographic order.   Let the resulting sequence be denoted by $\mathcal{D}_d(n)$. 
\begin{remark}
For any $d < n$, $\mathcal{D}_d(n)$ begins with $0^{n-d}1^d$ and ends with $1^n$.   
\end{remark}
By complementing the bits in $\mathcal{D}_d(n)$, we obtain a maximum weight de Bruijn sequence with  weight at most $n-d$.  Denote this sequence by $\mathcal{\overline{D}}_{d}(n)$.   From the previous remark, it begins with $1^{n-d}0^d$ and ends with $0^n$. 
 
\begin{exam} 
The necklaces of length 6 with weight $d \geq 3$ in lexicographic order are:
\[0 0 0 1 1 1,
0 0 1 0 1 1,
0 0 1 1 0 1,
0 0 1 1 1 1,
0 1 0 1 0 1,
0 1 0 1 1 1,
0 1 1 0 1 1,
0 1 1 1 1 1,
1 1 1 1 1 1.\]
Concatenating together their periodic reductions we obtain the minimum weight de Bruijn sequence~$\mathcal{D}_3(6)$.
\[0 0 0 1 1 1 \cdot 
0 0 1 0 1 1 \cdot 
0 0 1 1 0 1 \cdot 
0 0 1 1 1 1 \cdot 
0 1 \cdot 
0 1 0 1 1 1 \cdot 
0 1 1 \cdot 
0 1 1 1 1 1 \cdot 
1\]
As further examples, 
\[ \mathcal{D}_4(6) =   0 0 1 1 1 1 \cdot  0 1 0 1 1 1 \cdot  0 1 1 \cdot  0 1 1 1 1 1 \cdot  1\] and
\[ \mathcal{\overline{D}}_4(6) =   110000 \cdot  101000 \cdot  100 \cdot  100000 \cdot  0.\]
\end{exam}

From the above example observe that:
\begin{itemize}
\item   $\mathcal{D}_3(6)$ contains all binary strings of length $6$ with weight greater than or equal to $3$, 
\item   $\mathcal{\overline{D}}_4(6)$ contains all binary strings of length $6$ with weight less than or equal to $2$,
\item   The length $n{-}1$ prefix of $\mathcal{\overline{D}}_4(6)$, namely $11000$, appears in the wraparound of $\mathcal{D}_3(6)$.
\end{itemize}
Let $\mathcal{D}^r_d(n)$ denote the sequence $\mathcal{D}_d(n)$ with the suffix $1^{d-1}$ rotated to the front.  Then by applying the Gluing Lemma~\cite{dbrange}, the following is a de Bruijn sequence of order $6$: 
\[\underbrace{\underline{11000}01010001001000000}_{\text{$\mathcal{\overline{D}}_4(6)$}} ~ \cdot
~\underbrace{\underline{\red{11}000}1110010110011010011110101011101101111}_{\text{$\mathcal{D}^r_3(6)$}}.\]
Applying this strategy more generally, let $\mathcal{DB}_{max}(n)$ denote the de Bruijn sequence obtained by joining two such smaller cycles.
%
\begin{boxed2}
\noindent
{\bf Weight-range construction}
\[ \mathcal{DB}_{max}(n) = \mathcal{\overline{D}}_{d}(n) \cdot \mathcal{D}^r_{d'}(n), \]
where   $d = \lfloor n/2\rfloor +1$ and $d'=\lceil n/2\rceil$.  
\end{boxed2}
\noindent
A complete C implementation to construct $\mathcal{DB}_{max}(n)$ is given in the Appendix\footnote{It is also available at \url{http://debruijnsequence.org}.}.

The following technical lemma leads to a lower bound for the discrepancy of $\mathcal{DB}_{max}(n)$.

\begin{lemma} \label{lem:binomial}
A maximum weight de Bruijn sequence of order $n$ and maximum weight $d$ has ${n-1 \choose d}$ more 0s than 1s.

\end{lemma}
\begin{proof}
By definition, a maximum weight de Bruijn sequence of order $n$ and maximum weight $d$ contains every binary string of length $n$ with weight at most $d$ as a substring exactly once. Since each bit in this sequence belongs to $n$ different strings the total number of $1$s in the sequence is
\begin{eqnarray*}
\mathit{ones}  
	 & = & \frac{1}{n} \sum_{j=0}^d  j {n \choose j} \\
	& = &   \frac{0}{n} {n \choose 0}  + \frac{1}{n}   {n \choose 1} + \frac{2}{n}   {n \choose 2} +  \cdots +   \frac{d}{n} {n \choose d}     \\ 
          & =  &  0  +  \ {n-1 \choose 0} + {n-1 \choose 1} + \cdots +  {n-1 \choose d-1},
\end{eqnarray*} 
and the total number of $0$s is
\begin{eqnarray*}
\mathit{zeros}  
	 & = & \frac{1}{n} \sum_{j=0}^d  (n-j) {n \choose j} \\
& = &    \frac{n}{n}  {n \choose 0} +  \frac{n-1}{n} {n \choose 1} +  \frac{n-2}{n}  {n \choose 2} +  \cdots +   \frac{n-d}{n}  {n \choose d}  \\
          	      & =  &   {n-1 \choose 0} + {n-1 \choose 1} + {n-1 \choose 2} + \cdots  + {n-1 \choose d}.
\end{eqnarray*}
Thus $\mathit{zeros} - \mathit{ones} = {n-1 \choose d}$.  
\end{proof}

\begin{theorem}
The de Bruijn sequence $\mathcal{DB}_{max}(n)$ has discrepancy at least  ${n-1 \choose \lfloor n/2 \rfloor} + \lfloor \frac{n}{2} \rfloor$.
\end{theorem}
\begin{proof}
Let $d = \lfloor n/2\rfloor +1$ and $d'=\lceil n/2\rceil$.  Recall that $\mathcal{\overline{D}}_{d}(n)$ is a maximum weight de Bruijn sequence with maximum weight $n-d$.
Thus, by  Lemma~\ref{lem:binomial},  it has ${n-1 \choose n- d}   = {n-1 \choose n- (\lfloor n/2\rfloor +1)}  = {n-1 \choose \lfloor n/2 \rfloor}$ more $0$s than $1$s.
Consider  $\mathcal{\overline{D}}_{d}(n)$ with its prefix of $1^{n-d}$ removed.  The resulting string, which is a substring of  $\mathcal{DB}_{max}(n)$, has  ${n-1 \choose \lfloor n/2 \rfloor} + (n-d)$ more $0$s than $1$s. 
When $n$ is odd we have $n-d = n- \lfloor n/2\rfloor -1 =  \lfloor \frac{n}{2} \rfloor$ and thus $\mathcal{DB}_{max}(n)$  has discrepancy at least ${n-1 \choose \lfloor n/2 \rfloor} + \lfloor \frac{n}{2} \rfloor$.  
When $n$ is even, we additionally add the length $n-1$ prefix of $\mathcal{D}^r_{d'}(n)$ which has more $0$s than $1$s (exactly one more).  Since $n-d+1 = n - (\lfloor n/2\rfloor -1) +1 =  \lfloor \frac{n}{2} \rfloor$ (when $n$ is even) this again means that  $\mathcal{DB}_{max}(n)$  has discrepancy at least ${n-1 \choose \lfloor n/2 \rfloor} + \lfloor \frac{n}{2} \rfloor$.
\end{proof}
By applying Stirling's approximation to  ${n-1 \choose \lfloor n/2 \rfloor}$  we obtain the following corollary.  %
\begin{corollary}
The discrepancy of the de Bruijn sequence $\mathcal{DB}_{max}(n)$ attains the asymptotic upper bound of  $\Theta(\frac{2^n}{\sqrt{n}})$.
\end{corollary}

\noindent
Observe from Table~\ref{tab:discreps} that the discrepancy of $\mathcal{DB}_{max}(n)$ is exactly ${n-1 \choose \lfloor n/2 \rfloor} + \lfloor \frac{n}{2} \rfloor$ for $10 \leq n \leq 25$.  This leads to the following conjecture.

\begin{conjecture} \label{con:max}
The de Bruijn sequence $\mathcal{DB}_{max}(n)$ has discrepancy equal to ${n-1 \choose \lfloor n/2 \rfloor} + \lfloor \frac{n}{2} \rfloor$, and moreover, it is the maximum possible discrepancy over all de Bruijn sequences of order $n$.
\end{conjecture}

As noted earlier, the discrepancy of the {\bf Cool-lex} construction matches the discrepancy for the {\bf Weight-range} construction for $10 \leq n \leq 25$, except for when $n \bmod 4 \equiv 1$ (see Table~\ref{tab:discreps}).   As illustration, the {\bf Cool-lex} construction first constructs cycles of the following weights for $n=6,7,8,9$:
\begin{itemize}
\item $n=6$: (0,1,2), (3,4), (5,6)    
\item $n=7$: (0,1), (2,3), (4,5), (6,7) 
\item $n=8$: (0,1,2), (3,4), (5,6), (7,8)
\item $n=9$: (0,1), (2,3), (4,5), (6,7), (8,9)
\end{itemize}
before joining them together one at a time.  Note when $n=9$, strings with weights $4$ and $5$ are grouped together before the smaller cycles are joined together.  This causes a reduction in the discrepancy compared to the {\bf Weight-range} construction.  It is possible, however, to tweak the {\bf Cool-lex} implementation so the discrepancies are equivalent. For instance for $n=9$, the smaller cycles with weights $(0,1,2),(3,4),(5,6),(7,8,9)$ could be joined together instead.

Recently, a shift-rule construction based on the PSR and CSR has been discovered to generate the same sequence as {\bf Cool-lex}~\cite{PSR-manuscript}.  A discussion on generating de Bruijn sequences applying the PSR and CSR is also given in~\cite{Etzion1987}; it describes joining small cycles together in the same manner as {\bf Cool-lex}. Thus, we anticipate the resulting sequences would obtain a similar discrepancy profile.

\section{Group 5: LFSR constructions based on primitive polynomials} \label{sec:lfsr}

In this section we consider de Bruijn sequences that can be generated for a specific $n$ by a primitive polynomial of degree $n$.  As discussed by Golomb~\cite{golomb2017}, a primitive polynomial of the form $g(x) = c_0 + c_1x + c_2x^2 + \cdots + c_nx^n$ over GF(2) corresponds to a feedback function of the form $f(a_1a_2\cdots a_n) = c_na_1 \oplus c_{n-1}a_2 \oplus \cdots \oplus c_1a_n$. 
\begin{exam}
The primitive polynomial 
$1 + x^2 + x^5$ of degree $5$ over GF(2) corresponds to the feedback function
$f(a_1a_2a_3a_4a_5) = a_1 + a_4$.  If $a_1a_2a_3a_4a_5$ is initialized to $00001$, then
the LFSR with this feedback function produces the m-sequence
$$0000101011101100011111001101001$$
of length $2^5-1=31$ when outputting the value $a_1$ before each application of the LFSR. By prepending a 0 to the beginning of this m-sequence we obtain a de Bruijn sequence.
\end{exam}

 To obtain the data in Table~\ref{tab:lfsr}, we generated all primitive polynomials of degree $n$ for $n=10,11, \ldots, 25$ along with their corresponding LFSRs. The algorithm used to exhaustively list the primitive polynomials is based on the work in~\cite{Cattell&etal:2000} and is available at \url{http://debruijnsequence.org/lfsr}.   We seeded the LFSRs with $0^{n-1}1$ as described in the above example to obtain a de Bruijn sequence.   We then computed the discrepancy of all such sequences.  The number of primitive polynomials (and hence LFSRs) of degree $n$ is given by sequence A011260 in the On-Line Encyclopedia of Integer Sequences~\cite{OEIS}.   This number is listed in the final column of Table~\ref{tab:lfsr}.

Below is a list of feedback functions for $n=10,11, \ldots, 25$ that generated de Bruin sequences with discrepancy closest to the corresponding entry for a random sequence.

\begin{center}
\footnotesize
\begin{tabular} {c |  l r r }
$n$ 	&   {\bf Feedback function} &  \bblue{Random}  &   {\bf Discrepancy}  \\ \hline
10 &     $a_1 \oplus a_2 \oplus a_6 \oplus a_9$  &  \blue{50} 	            &    46    \\
11 &     $a_1 \oplus a_{6} \oplus a_{7} \oplus a_{10}$ &   \blue{71}     	        &    68      \\
12 &     $a_1 \oplus a_{4} \oplus a_{7} \oplus a_{8} \oplus a_{9} \oplus a_{12}$  &  \blue{101}    	        &    99   \\
13 &     $a_1 \oplus a_{2} \oplus a_{4} \oplus a_{5} \oplus a_{6} \oplus a_{7} \oplus a_{8} \oplus a_{10} \oplus a_{11} \oplus a_{12}$ &  \blue{143}	      	    &     143   \\
14 &     $a_1 \oplus a_{2} \oplus a_{4} \oplus a_{5} \oplus a_{7} \oplus a_{8} \oplus a_{9} \oplus a_{13}$&  \blue{203}    	        &    203    \\ 
15 &     $a_1 \oplus a_{2} \oplus a_{6} \oplus a_{11}$ &   \blue{288}     	        &    287   \\ 
16 &     $a_1 \oplus a_{2} \oplus a_{3} \oplus a_{4} \oplus a_{13} \oplus a_{15}$ &  \blue{407}     	        &    406   \\ 
17 &     $a_1 \oplus a_{3} \oplus a_{5} \oplus a_{8} \oplus a_{9} \oplus a_{10} \oplus a_{11} \oplus a_{12} \oplus a_{13} \oplus a_{15} \oplus a_{16} \oplus a_{17}$ &  \blue{575} 	      	    &    575    \\ 
18 &     $a_1 \oplus a_{4} \oplus a_{5} \oplus a_{7} \oplus a_{10} \oplus a_{11} \oplus a_{12} \oplus a_{15} \oplus a_{16} \oplus a_{17}$
 &   \blue{815}     	        &   814      \\   
19 &     $a_1 \oplus a_{2} \oplus a_{3} \oplus a_{5} \oplus a_{6} \oplus a_{10} \oplus a_{11} \oplus a_{14} \oplus a_{16} \oplus a_{17}$ &  \blue{1157}    	        &   1157     \\   
20 &     $a_1 \oplus a_{2} \oplus a_{3} \oplus a_{4} \oplus a_{5} \oplus a_{6} \oplus a_{7} \oplus a_{8} \oplus a_{9} \oplus a_{10} \oplus a_{11} \oplus a_{12} \oplus a_{13} \oplus a_{15} \oplus a_{16} \oplus a_{17} \oplus a_{19} \oplus a_{20}$  &  \blue{1634}     	    &   1633    \\ 
21 &     $a_1 \oplus a_{2} \oplus a_{3} \oplus a_{5} \oplus a_{8} \oplus a_{11} \oplus a_{16} \oplus a_{19} \oplus a_{20} \oplus a_{21}$  &  \blue{2311}    	        &   2311   \\  
22 &     $a_1 \oplus a_{3} \oplus a_{6} \oplus a_{12} \oplus a_{13} \oplus a_{15} \oplus a_{17} \oplus a_{18} \oplus a_{19} \oplus a_{20} \oplus a_{21} \oplus a_{22}$ &   \blue{3264}   	        &  3264     \\ 
23 &     $a_1 \oplus a_{2} \oplus a_{3} \oplus a_{4} \oplus a_{5} \oplus a_{6} \oplus a_{9} \oplus a_{10} \oplus a_{12} \oplus a_{14} \oplus a_{16} \oplus a_{17} \oplus a_{18} \oplus a_{23}$ &     \blue{4565} 	        &  4565     \\ 
24 &     $a_1 \oplus a_{2} \oplus a_{4} \oplus a_{6} \oplus a_{12} \oplus a_{13} \oplus a_{14} \oplus a_{15} \oplus a_{16} \oplus a_{17} \oplus a_{18} \oplus a_{20} \oplus a_{21} \oplus a_{24}$ & \blue{6252}   	        &  6252     \\ 
25 &    $a_1 \oplus a_{8} \oplus a_{9} \oplus a_{11} \oplus a_{12} \oplus a_{16} \oplus a_{20} \oplus a_{21} \oplus a_{24} \oplus a_{25}$     & \blue{9192}  	                & 9192         \\  
\end{tabular}
\end{center}


\section{Future directions and open problems} \label{sec:fut}

In this paper, we investigated the discrepancies of $13$ de Bruijn sequence constructions.  We proved that two constructions attain the lower bound of $\Theta(n)$ and presented one new construction that attains the upper bound of $\Theta(\frac{2^n}{\sqrt{n}})$.  It remains an interesting problem to demonstrate a generic construction for all $n$ with discrepancy that is close to that of a random stream of bits of the same length. Some additional avenues of future research include the following.

\begin{enumerate}
\item Simplify the description of the  {\bf Huang} construction~\cite{huang}.  Does it have the smallest discrepancy over all de Bruijn sequences?
\item Answer the conjectures regarding the discrepancies for the greedy {\bf Pref-same} and {\bf Pref-opposite} constructions (Conjecture~\ref{con:same} and Conjecture~\ref{con:opp}).
\item Analyze the discrepancy of {\bf PCR4} which had discrepancy closest to one we might expect from a random stream of bits.
\item Determine whether or not the maximum discrepancy of any de Bruijn sequence is ${n-1 \choose \lfloor n/2 \rfloor} + \lfloor \frac{n}{2} \rfloor$ (Conjecture~\ref{con:max}).
\item Generalize the investigation of discrepancy to de Bruijn sequences over an arbitrary alphabet size $k$.
\item Study the distribution of discrepancy over all possible de Bruijn sequences.
\end{enumerate}



\section{Acknowledgement}\label{sec:ack}
The research of Joe Sawada is supported by the \emph{Natural Sciences and Engineering Research Council of Canada} (NSERC) grant RGPIN-2018-04211.

\bibliographystyle{abbrv}
\bibliography{refs}

\newpage
\appendix
\section{Table of discrepancies}

\begin{center}
\footnotesize
\begin{tabular} {c |  cccc | c c c   }
  & \multicolumn{4}{c}{\bf ( Group 1 )}  &  \multicolumn{3}{c}{\bf ( Group 2 ) } \\
$n$ 	& {\bf Huang} & {\bf CCR2} & {\bf CCR3} & {\bf CCR1} & {\bf Pref-same} & {\bf Lex-comp} & {\bf Pref-opposite}   \\ \hline
10 &      12 &       13 &   13   &     16 &    ~  24 &        ~ 24    &    ~ 27  \\
11 &      13 &       14 &  15  &      18 &    ~  29 &        ~ 29    &    ~ 34  	\\
12 &      15 &      16 &  16  &      22 &     ~ 35 &        ~ 35    &    ~ 43  	\\
13 &      16 &      17 &  18  &      23 &    ~ 43 &        ~ 43    &    ~ 52  	\\
14 &      18 &      19 &  20  &      30 &     ~ 48 &        ~ 48    &    ~ 63  	\\
15 &      19 &      21 &  21  &      29 &    ~ 59 &        ~ 59    &    ~ 74  	\\
16 &      21 &      22 &  23  &      36 &    ~  68 &        ~ 68    &    ~ 87 \\
17 &      22 &      24 &  25  &      37 &     ~ 79 &        ~ 79    &   100	\\
18 &      24 &      26 &  26  &      43 &     ~ 88 &         ~ 88   &   115  	\\
19 &      25 &      27 &  28  &      43 &     103 &       103  &    130	\\
20 &      27 &      29 &  30  &      52 &     114 &       114   &    147	\\
21 &      28 &      31 &  31  &      50 &     127 &       127  &    164	\\
22 &      30 &      32 &  33  &      59 &     142 &       142  &    183	\\
23 &      31 &      34 &  35  &      59 &     155 &       155  &    202	\\
24 &      33 &      36 &  36  &      67 &     172 &       172  &    223	\\
25 &      35 &      37 &  38  &      66 &     187 &       187  &    244	 \\
26 &   36 &      39	 &      40	&      77	&     208		 &     208		&     267			\\
27 &    38 &      41	 &      42	&      74	&     224		 &     224		&     290			\\
28 &    40 &      43	 &      43	&      85	&     246		 &     246		&     315			\\
29 &    41 &      44	 &      45	&      84	&     264		 &     264		&     340			\\
30 &    43 &      46	 &      47	&      94	&     286		 &     286		&     367			\\
\end{tabular}
\end{center}

\medskip

\begin{center}
\footnotesize
\begin{tabular} {c |  r r r r r | r r}
  & \multicolumn{5}{c}{\bf ( Group  3 )}  &  \multicolumn{2}{c}{\bf ( Group 4 ) } \\
$n$ 	&  {\bf PCR4} &  \bblue{Random}  &  {\bf PCR3}  & {\bf PCR2}   & {\bf Prefer-1/PCR1}  & {\bf Cool-lex}  & {\bf Weight-range}  \\ \hline
10 &       29     &  \blue{50} 	&      75           	&     101   &     120   &     131	    &     131	\\
11 &      41      &   \blue{71} 	&     141    	 &     180 &     222  &     257	&     257	\\
12 &     51       &  \blue{101} 	&     248   		&     321 &     416  &     468 	&     468	\\
13 &     70       &  \blue{143}	&     468   		&     587 &     784  &     801 	&     930	\\
14 &     85       &  \blue{203} 	&     850   		&    1065 &    1488  &    1723 	&    1723	\\
15 &    110      &   \blue{288}  	&    1604   	&    1974 &    2824  &    3439 	&    3439	\\
16 &    175      &  \blue{407}  	&    2965   	&    3632 &    5376  &    6443 	&    6443	\\
17 &    246      &  \blue{575} 	&    5594   	&    6785 &   10229  &   11452 	&   12878	\\
18 &    326      &   \blue{815} 	&   10461    	&   12635 &   19484  &   24319 	&   24319	\\
19 &     462     &  \blue{1157} 	&   19765   	&   23746 &   37107  &   48629	&   48629	 \\
20 &       730   &  \blue{1634}   	&   37243  	&   44585 &   71250  &   92388 		&   92388	\\
21 &     954     &  \blue{2311} 	&   70575   	&   84270 &  138332  &  167975 	&  184766		\\
22 &    1327    &   \blue{3264} 	&  133737  	&  159281 &  268582  &  352727 	&  352727		\\
23 &    1820    &    \blue{4565}	&  254322 	&  302449 &  521553  &  705443	&  705443		 \\
24 &    2684    &    \blue{6252}	&  484172 	&  574819 & 1012795  & 1352090	& 1352090 		\\
25 &    3183    &   \blue{9192} 	&  924071 	& 1096009 & 1966813  & 2496163	& 2704168		 \\
26 &   4108	& \blue{13074}	 & 1766284 	  & 2092284 	 & 3819605 		& 5200313 		& 5200313 		\\
27 &   5604	& \blue{17933}	 & 3382851 	 & 4004050 	& 7453523 		 & 10400613 		 & 10400613		\\
28 &  7629	& \blue{22672}	 & 6488970 	 & 7672443 	 & 14544826		 & 20058314 		 & 20058314		\\
29 &  10433	& \blue{34591}	  & 12468181	  & 14730243	 & 28382864		 & 37442182 		 & 40116614		\\
30 &  13637	& \blue{57357}	  & 23991972	  & 28316271	 & 55421919		 & 77558775 		 & 77558775		\\
\end{tabular}
\end{center}

\newpage

\section{C program to construct maximum-discrepancy de Bruijn sequences}
\footnotesize
\begin{code}
#include <stdio.h>
int n,c,a[100],first;

//-------------------------------------------------------------------------------
// Generate the lexicographically smallest universal cycle (de Bruijn sequence)
// for binary strings of length "n" with minimum weight "c" when comp = 0. When
// comp=1, it complements the bits producing a maximum weight n-c universal cycle
//-------------------------------------------------------------------------------
void Gen(int t, int p, int w, int comp) {
    int i;
    
    if (t > n) {
        if (n
            if (first == 0) {
                for (i=1; i<=c; i++) printf("0");
                first = 1;
            }
            else {
                for (i=1; i <= p; i++) printf("
            }
        }
    }
    else {
        // Append 0
        a[t] = 0;
        if (a[t-p] == 0 && c-w < n-t+1)  Gen(t+1,p,w,comp);
        
        // Append 1
        a[t] = 1;
        if (a[t-p] == 1) Gen(t+1,p,w+1,comp);
        else Gen(t+1,t,w+1, comp);
    }
}
//===============================================
int main() {
    
    printf("Enter n: ");   scanf("
    
    c = n/2+1;
    for (int i=1; i<=c-1; i++) printf("1");
    a[0] = 0; first = 1;
    Gen(1,1,0,0);  
    
    c = n-c+1; first = 0;
    Gen(1,1,0,1);	
    printf("\n");
}
\end{code}

\end{document}

%% file: paper_macros.tex
\usepackage[margin=0.95in]{geometry}
\usepackage[noend]{algpseudocode}
\usepackage{algorithm}
\usepackage{times}
\usepackage{amssymb}
\usepackage{amsmath}
\usepackage{latexsym}
\usepackage{graphics}
\usepackage{url}
\usepackage{verbatim}
\usepackage{color}
\usepackage{setspace}
\usepackage{multirow}
\usepackage{multicol}
\usepackage{lineno}
\usepackage{booktabs}
\usepackage{graphicx}
\usepackage{tabularx}

\usepackage{tikz}
\usepackage{ctable}
\usepackage{pgfplots}

\usepackage{listings}
\usepackage{amssymb}
\usepackage{latexsym}
\usepackage{graphics}
\usepackage{cite}
\usepackage{booktabs}
\usepackage{amsmath}
\usepackage{ragged2e}
\usepackage{setspace}
\usepackage[T1]{fontenc}
\usepackage[utf8]{inputenc}
\usepackage[english]{babel}
\usepackage{algorithm}
\usepackage{algpseudocode}
\usepackage{pifont}
\usepackage{graphicx}
\usepackage{yhmath}
\usepackage{mathdots}
\usepackage{MnSymbol}%
\usepackage[titletoc]{appendix}

\usepackage{mathtools}
\usepackage{color}

\algnewcommand\And{\textbf{and}}
\algnewcommand\Or{\textbf{or}}
\algnewcommand\Not{\textbf{not}}
\algnewcommand\In{\textbf{in}}
\algnewcommand\Each{\textbf{each}}

\newcommand{\sub}{\mathbf{sub}}

\newtheorem{theorem}{Theorem}[section]          
\newtheorem{lemma}[theorem]{Lemma}             
\newtheorem{corollary}[theorem]{Corollary}         
\newtheorem{remark}[theorem]{Remark}           
\newtheorem{conjecture}[theorem]{Conjecture}  

\newenvironment{proof}{{\em Proof.}}{\hspace*{\fill}$\Box$\par\vspace{3mm}}

\newcommand{\squishlist}{
 \begin{list}{$\bullet$}
  { \setlength{\itemsep}{0pt}
     \setlength{\parsep}{3pt}
     \setlength{\topsep}{3pt}
     \setlength{\partopsep}{0pt}
     \setlength{\leftmargin}{2.5em}
     \setlength{\labelwidth}{1em}
     \setlength{\labelsep}{0.5em} } }

\newcommand{\squishlisttwo}{
 \begin{list}{$\triangleright$}
  { \setlength{\itemsep}{0pt}
     \setlength{\parsep}{0pt}
    \setlength{\topsep}{0pt}
    \setlength{\partopsep}{0pt}
    \setlength{\leftmargin}{2em}
    \setlength{\labelwidth}{1.5em}
    \setlength{\labelsep}{0.5em} } }

\newcommand{\squishend}{
  \end{list}  }

\usepackage{listings}

\definecolor{verbgray}{gray}{0.9}

\lstnewenvironment{code}{%
  \lstset{backgroundcolor=\color{verbgray},
 frame=single,
  language=C,
  framerule=0pt,
  basicstyle=\ttfamily,
  showstringspaces=false,
  commentstyle=\color{blue}\textit,
  keywordstyle=\color{black}\bf,
  numbers=none,
  keepspaces=true,
  columns=fullflexible}}{}

\newcommand{\bB}{\mathbf{B}}

\definecolor{shadecolor}{rgb}{.91, .91, .91}
\definecolor{bordercolor}{rgb}{.8, .8, .6}

\definecolor{ultramarine}{rgb}{0, 0.125, 0.376}

 \definecolor{arsenic}{rgb}{0.23, 0.27, 0.29}
 \definecolor{beige}{rgb}{0.96, 0.96, 0.86}

\definecolor{amber}{rgb}{1.0, 0.75, 0.0}
\definecolor{orange}{rgb}{1.0, 0.49, 0.0}
\definecolor{dandelion}{rgb}{0.94, 0.88, 0.19}

  \definecolor{indiagreen}{rgb}{0.07, 0.53, 0.03}
  \definecolor{huntergreen}{rgb}{0.21, 0.37, 0.23}

\newcommand{\blue}[1] {\textcolor{blue}{#1}}
\newcommand{\red}[1] {\textcolor{red}{#1}}

\newcommand{\bblue}[1] {{\bf \textcolor{blue}{#1}}}

\newcommand{\defo}[1] {\emph{\textcolor{blue}{#1}}}

\definecolor{shadecolor}{rgb}{.9, .9, .9}

 \usepackage{framed}

 \colorlet{framecolor}{ultramarine}
  \colorlet{exframecolor}{orange}

    {\endMakeFramed}

    \newenvironment{frshaded*}{%
    \MakeFramed {\advance\hsize-\width \FrameRestore}}%
    {\endMakeFramed}
    
    \newcounter{examplecounter}
\newenvironment{exam}{
 \begin{frshaded*}
    \refstepcounter{examplecounter}%
    \noindent
  \textbf{Example \arabic{examplecounter}}%
  \quad
}{%
\end{frshaded*}
}

{\endMakeFramed}

\newenvironment{frshaded2*}{%
    \MakeFramed {\advance\hsize-\width \FrameRestore}}%
    {\endMakeFramed}

\newenvironment{boxed2}{
 \begin{frshaded2*}
}{%
\end{frshaded2*}
}

%% file: tables-CCR.tex
\begin{center}

\begin{tikzpicture} \footnotesize
\begin{axis}[
    width = 3.0in,
    xlabel={CCR1 sequence for $n=10$},
    ylabel={\# 1s $-$ \# 0s in prefix},
    xmin=0, xmax=1024,
    ymin=-11, ymax=11,
    ytick pos=bottom,
    ymajorgrids=true,
    grid style=dashed,
]

      \addplot[
    color=blue,
    ]
    coordinates {
(0,0)(1,-1)(2,-2)(3,-3)(4,-4)(5,-5)(6,-6)(7,-7)(8,-8)(9,-9)(10,-10)(11,-9)(12,-8)(13,-7)(14,-6)(15,-5)(16,-4)(17,-3)(18,-2)(19,-1)(20,0)(21,-1)(22,-2)(23,-3)(24,-4)(25,-5)(26,-4)(27,-3)(28,-2)(29,-1)(30,-2)(31,-1)(32,0)(33,1)(34,2)(35,3)(36,2)(37,1)(38,0)(39,-1)(40,0)(41,1)(42,2)(43,1)(44,2)(45,1)(46,2)(47,3)(48,4)(49,5)(50,4)(51,3)(52,2)(53,3)(54,2)(55,3)(56,4)(57,5)(58,4)(59,3)(60,4)(61,5)(62,6)(63,5)(64,6)(65,5)(66,4)(67,3)(68,4)(69,5)(70,4)(71,3)(72,2)(73,3)(74,2)(75,3)(76,4)(77,3)(78,4)(79,3)(80,4)(81,5)(82,6)(83,5)(84,6)(85,5)(86,4)(87,5)(88,4)(89,5)(90,4)(91,3)(92,2)(93,3)(94,2)(95,3)(96,4)(97,3)(98,2)(99,1)(100,2)(101,3)(102,4)(103,3)(104,4)(105,3)(106,2)(107,3)(108,4)(109,5)(110,4)(111,3)(112,2)(113,3)(114,2)(115,3)(116,2)(117,3)(118,4)(119,3)(120,4)(121,5)(122,6)(123,5)(124,6)(125,5)(126,6)(127,5)(128,4)(129,5)(130,4)(131,3)(132,2)(133,3)(134,2)(135,3)(136,2)(137,3)(138,2)(139,1)(140,2)(141,3)(142,4)(143,3)(144,4)(145,3)(146,4)(147,3)(148,4)(149,5)(150,4)(151,3)(152,2)(153,3)(154,2)(155,3)(156,2)(157,1)(158,2)(159,1)(160,2)(161,3)(162,4)(163,3)(164,4)(165,3)(166,4)(167,5)(168,4)(169,5)(170,4)(171,3)(172,2)(173,3)(174,2)(175,3)(176,2)(177,1)(178,0)(179,-1)(180,0)(181,1)(182,2)(183,1)(184,0)(185,-1)(186,0)(187,1)(188,2)(189,3)(190,2)(191,1)(192,0)(193,1)(194,2)(195,3)(196,2)(197,1)(198,0)(199,-1)(200,0)(201,1)(202,0)(203,1)(204,2)(205,1)(206,2)(207,3)(208,4)(209,5)(210,4)(211,3)(212,4)(213,3)(214,2)(215,3)(216,2)(217,3)(218,2)(219,1)(220,2)(221,3)(222,2)(223,3)(224,4)(225,3)(226,4)(227,3)(228,4)(229,5)(230,4)(231,3)(232,4)(233,3)(234,2)(235,3)(236,2)(237,1)(238,2)(239,1)(240,2)(241,3)(242,2)(243,3)(244,4)(245,3)(246,4)(247,5)(248,4)(249,5)(250,4)(251,3)(252,4)(253,3)(254,2)(255,3)(256,2)(257,1)(258,0)(259,-1)(260,0)(261,1)(262,0)(263,1)(264,0)(265,-1)(266,0)(267,1)(268,2)(269,3)(270,2)(271,1)(272,2)(273,1)(274,2)(275,3)(276,2)(277,1)(278,2)(279,1)(280,2)(281,3)(282,2)(283,3)(284,2)(285,1)(286,2)(287,3)(288,2)(289,3)(290,2)(291,1)(292,2)(293,1)(294,2)(295,3)(296,2)(297,1)(298,0)(299,-1)(300,0)(301,1)(302,0)(303,-1)(304,0)(305,-1)(306,0)(307,1)(308,2)(309,3)(310,2)(311,1)(312,2)(313,3)(314,2)(315,3)(316,2)(317,1)(318,0)(319,-1)(320,0)(321,1)(322,0)(323,-1)(324,-2)(325,-3)(326,-2)(327,-1)(328,0)(329,1)(330,0)(331,-1)(332,0)(333,1)(334,2)(335,3)(336,2)(337,1)(338,0)(339,-1)(340,0)(341,-1)(342,0)(343,1)(344,2)(345,1)(346,2)(347,3)(348,4)(349,5)(350,4)(351,5)(352,4)(353,3)(354,2)(355,3)(356,2)(357,1)(358,0)(359,-1)(360,0)(361,-1)(362,0)(363,1)(364,0)(365,-1)(366,0)(367,1)(368,2)(369,3)(370,2)(371,3)(372,2)(373,1)(374,2)(375,3)(376,2)(377,1)(378,0)(379,-1)(380,0)(381,-1)(382,0)(383,-1)(384,0)(385,-1)(386,0)(387,1)(388,2)(389,3)(390,2)(391,3)(392,2)(393,3)(394,2)(395,3)(396,2)(397,1)(398,0)(399,-1)(400,0)(401,-1)(402,0)(403,-1)(404,-2)(405,-3)(406,-2)(407,-1)(408,0)(409,1)(410,0)(411,1)(412,0)(413,1)(414,2)(415,3)(416,2)(417,1)(418,0)(419,-1)(420,0)(421,-1)(422,-2)(423,-1)(424,0)(425,-1)(426,0)(427,1)(428,2)(429,3)(430,2)(431,3)(432,4)(433,3)(434,2)(435,3)(436,2)(437,1)(438,0)(439,-1)(440,0)(441,-1)(442,-2)(443,-1)(444,-2)(445,-3)(446,-2)(447,-1)(448,0)(449,1)(450,0)(451,1)(452,2)(453,1)(454,2)(455,3)(456,2)(457,1)(458,0)(459,-1)(460,0)(461,-1)(462,-2)(463,-3)(464,-2)(465,-3)(466,-2)(467,-1)(468,0)(469,1)(470,0)(471,1)(472,2)(473,3)(474,2)(475,3)(476,2)(477,1)(478,0)(479,-1)(480,0)(481,-1)(482,-2)(483,-3)(484,-4)(485,-5)(486,-4)(487,-3)(488,-2)(489,-3)(490,-4)(491,-3)(492,-2)(493,-1)(494,0)(495,1)(496,0)(497,-1)(498,-2)(499,-1)(500,0)(501,-1)(502,0)(503,1)(504,0)(505,-1)(506,0)(507,1)(508,2)(509,1)(510,0)(511,1)(512,0)(513,-1)(514,0)(515,1)(516,0)(517,1)(518,0)(519,1)(520,2)(521,1)(522,2)(523,3)(524,2)(525,1)(526,2)(527,1)(528,2)(529,1)(530,0)(531,1)(532,0)(533,-1)(534,0)(535,1)(536,0)(537,-1)(538,-2)(539,-1)(540,0)(541,-1)(542,0)(543,-1)(544,0)(545,-1)(546,0)(547,1)(548,2)(549,1)(550,0)(551,1)(552,0)(553,1)(554,0)(555,1)(556,2)(557,1)(558,0)(559,1)(560,2)(561,1)(562,2)(563,1)(564,2)(565,1)(566,0)(567,1)(568,2)(569,1)(570,0)(571,1)(572,0)(573,1)(574,0)(575,1)(576,0)(577,1)(578,0)(579,1)(580,2)(581,1)(582,2)(583,1)(584,2)(585,1)(586,2)(587,1)(588,2)(589,1)(590,0)(591,1)(592,0)(593,1)(594,0)(595,1)(596,0)(597,-1)(598,-2)(599,-1)(600,0)(601,-1)(602,0)(603,-1)(604,-2)(605,-3)(606,-2)(607,-1)(608,0)(609,-1)(610,-2)(611,-1)(612,-2)(613,-1)(614,0)(615,1)(616,0)(617,-1)(618,-2)(619,-1)(620,0)(621,-1)(622,-2)(623,-1)(624,0)(625,-1)(626,0)(627,1)(628,2)(629,1)(630,0)(631,1)(632,2)(633,1)(634,0)(635,1)(636,2)(637,1)(638,0)(639,1)(640,0)(641,-1)(642,-2)(643,-1)(644,0)(645,-1)(646,-2)(647,-1)(648,-2)(649,-3)(650,-2)(651,-1)(652,0)(653,-1)(654,-2)(655,-1)(656,0)(657,-1)(658,0)(659,1)(660,0)(661,-1)(662,-2)(663,-1)(664,0)(665,-1)(666,-2)(667,-3)(668,-4)(669,-5)(670,-4)(671,-3)(672,-4)(673,-3)(674,-4)(675,-3)(676,-2)(677,-1)(678,0)(679,1)(680,0)(681,-1)(682,0)(683,-1)(684,0)(685,-1)(686,-2)(687,-3)(688,-4)(689,-5)(690,-4)(691,-3)(692,-4)(693,-5)(694,-6)(695,-5)(696,-4)(697,-3)(698,-2)(699,-1)(700,-2)(701,-3)(702,-2)(703,-1)(704,0)(705,-1)(706,-2)(707,-3)(708,-4)(709,-5)(710,-4)(711,-5)(712,-4)(713,-3)(714,-4)(715,-3)(716,-2)(717,-1)(718,0)(719,1)(720,0)(721,1)(722,0)(723,-1)(724,0)(725,-1)(726,-2)(727,-3)(728,-4)(729,-5)(730,-4)(731,-5)(732,-4)(733,-5)(734,-6)(735,-5)(736,-4)(737,-3)(738,-2)(739,-1)(740,-2)(741,-1)(742,-2)(743,-1)(744,0)(745,-1)(746,-2)(747,-3)(748,-4)(749,-5)(750,-4)(751,-5)(752,-6)(753,-5)(754,-6)(755,-5)(756,-4)(757,-3)(758,-2)(759,-1)(760,-2)(761,-1)(762,0)(763,-1)(764,0)(765,-1)(766,-2)(767,-3)(768,-4)(769,-5)(770,-4)(771,-5)(772,-6)(773,-7)(774,-8)(775,-7)(776,-6)(777,-5)(778,-4)(779,-3)(780,-4)(781,-3)(782,-2)(783,-1)(784,0)(785,-1)(786,-2)(787,-3)(788,-4)(789,-5)(790,-6)(791,-5)(792,-4)(793,-3)(794,-4)(795,-3)(796,-2)(797,-1)(798,0)(799,1)(800,2)(801,1)(802,0)(803,-1)(804,0)(805,-1)(806,-2)(807,-1)(808,0)(809,-1)(810,-2)(811,-1)(812,0)(813,1)(814,0)(815,1)(816,2)(817,1)(818,0)(819,1)(820,2)(821,1)(822,0)(823,-1)(824,0)(825,-1)(826,-2)(827,-1)(828,-2)(829,-1)(830,-2)(831,-1)(832,0)(833,1)(834,0)(835,1)(836,2)(837,1)(838,2)(839,1)(840,2)(841,1)(842,0)(843,-1)(844,0)(845,-1)(846,-2)(847,-1)(848,-2)(849,-3)(850,-4)(851,-3)(852,-2)(853,-1)(854,-2)(855,-1)(856,0)(857,-1)(858,0)(859,1)(860,2)(861,1)(862,0)(863,-1)(864,0)(865,-1)(866,-2)(867,-3)(868,-2)(869,-1)(870,-2)(871,-1)(872,0)(873,1)(874,0)(875,1)(876,2)(877,3)(878,2)(879,1)(880,2)(881,1)(882,0)(883,-1)(884,0)(885,-1)(886,-2)(887,-3)(888,-2)(889,-3)(890,-4)(891,-3)(892,-2)(893,-1)(894,-2)(895,-1)(896,0)(897,1)(898,0)(899,1)(900,2)(901,1)(902,0)(903,-1)(904,0)(905,-1)(906,-2)(907,-3)(908,-4)(909,-5)(910,-6)(911,-5)(912,-4)(913,-5)(914,-6)(915,-5)(916,-4)(917,-3)(918,-2)(919,-1)(920,0)(921,-1)(922,-2)(923,-1)(924,0)(925,-1)(926,-2)(927,-3)(928,-4)(929,-5)(930,-6)(931,-5)(932,-6)(933,-5)(934,-6)(935,-5)(936,-4)(937,-3)(938,-2)(939,-1)(940,0)(941,-1)(942,0)(943,-1)(944,0)(945,-1)(946,-2)(947,-3)(948,-4)(949,-5)(950,-6)(951,-5)(952,-6)(953,-7)(954,-8)(955,-7)(956,-6)(957,-5)(958,-4)(959,-3)(960,-2)(961,-3)(962,-2)(963,-1)(964,0)(965,-1)(966,-2)(967,-3)(968,-4)(969,-5)(970,-6)(971,-7)(972,-6)(973,-5)(974,-6)(975,-5)(976,-4)(977,-3)(978,-2)(979,-1)(980,0)(981,1)(982,0)(983,-1)(984,0)(985,-1)(986,-2)(987,-3)(988,-4)(989,-5)(990,-6)(991,-7)(992,-6)(993,-7)(994,-8)(995,-7)(996,-6)(997,-5)(998,-4)(999,-3)(1000,-2)(1001,-1)(1002,-2)(1003,-1)(1004,0)(1005,-1)(1006,-2)(1007,-3)(1008,-4)(1009,-5)(1010,-6)(1011,-7)(1012,-8)(1013,-7)(1014,-8)(1015,-7)(1016,-6)(1017,-5)(1018,-4)(1019,-3)(1020,-2)(1021,-1)(1022,0)(1023,-1)(1024,0)
    };
 
\end{axis}
\end{tikzpicture}
%
 \ \ \ 
 \begin{tikzpicture} \footnotesize
\begin{axis}[
    width = 3.0in,
    xlabel={CCR2 sequence for $n=10$},
    ylabel={\# 1s $-$ \# 0s in prefix},
    xmin=0, xmax=1024,
    ymin=-11, ymax=11,
    ytick pos=bottom,
    legend pos=north east,
    ymajorgrids=true,
    grid style=dashed,
]   
          \addplot[
    color=red,
    ]
    coordinates {
(0,0)(1,1)(2,2)(3,3)(4,4)(5,5)(6,6)(7,7)(8,8)(9,9)(10,10)(11,9)(12,8)(13,7)(14,6)(15,5)(16,4)(17,3)(18,2)(19,1)(20,0)(21,1)(22,2)(23,3)(24,4)(25,5)(26,4)(27,3)(28,2)(29,1)(30,2)(31,1)(32,0)(33,-1)(34,-2)(35,-3)(36,-2)(37,-1)(38,0)(39,1)(40,0)(41,-1)(42,-2)(43,-1)(44,-2)(45,-1)(46,-2)(47,-3)(48,-4)(49,-5)(50,-4)(51,-3)(52,-2)(53,-3)(54,-2)(55,-3)(56,-4)(57,-5)(58,-4)(59,-3)(60,-4)(61,-5)(62,-6)(63,-5)(64,-6)(65,-5)(66,-4)(67,-3)(68,-4)(69,-5)(70,-4)(71,-3)(72,-2)(73,-3)(74,-2)(75,-3)(76,-4)(77,-3)(78,-4)(79,-3)(80,-4)(81,-5)(82,-6)(83,-5)(84,-6)(85,-5)(86,-4)(87,-5)(88,-4)(89,-5)(90,-4)(91,-3)(92,-2)(93,-3)(94,-2)(95,-3)(96,-4)(97,-3)(98,-2)(99,-1)(100,-2)(101,-3)(102,-4)(103,-3)(104,-4)(105,-3)(106,-2)(107,-3)(108,-4)(109,-5)(110,-4)(111,-3)(112,-2)(113,-3)(114,-2)(115,-3)(116,-2)(117,-3)(118,-4)(119,-3)(120,-4)(121,-5)(122,-6)(123,-5)(124,-6)(125,-5)(126,-6)(127,-5)(128,-4)(129,-5)(130,-4)(131,-3)(132,-2)(133,-3)(134,-2)(135,-3)(136,-2)(137,-3)(138,-2)(139,-1)(140,-2)(141,-3)(142,-4)(143,-3)(144,-4)(145,-3)(146,-4)(147,-3)(148,-4)(149,-5)(150,-4)(151,-3)(152,-2)(153,-3)(154,-2)(155,-3)(156,-2)(157,-1)(158,-2)(159,-1)(160,-2)(161,-3)(162,-4)(163,-3)(164,-4)(165,-3)(166,-4)(167,-5)(168,-4)(169,-5)(170,-4)(171,-3)(172,-2)(173,-3)(174,-2)(175,-3)(176,-2)(177,-1)(178,0)(179,1)(180,0)(181,-1)(182,-2)(183,-1)(184,0)(185,1)(186,0)(187,-1)(188,-2)(189,-3)(190,-2)(191,-1)(192,0)(193,-1)(194,-2)(195,-3)(196,-2)(197,-1)(198,0)(199,1)(200,0)(201,-1)(202,0)(203,-1)(204,-2)(205,-1)(206,-2)(207,-3)(208,-4)(209,-5)(210,-4)(211,-3)(212,-4)(213,-3)(214,-2)(215,-3)(216,-2)(217,-3)(218,-2)(219,-1)(220,-2)(221,-3)(222,-2)(223,-3)(224,-4)(225,-3)(226,-4)(227,-3)(228,-4)(229,-5)(230,-4)(231,-3)(232,-4)(233,-3)(234,-2)(235,-3)(236,-2)(237,-1)(238,-2)(239,-1)(240,-2)(241,-3)(242,-2)(243,-3)(244,-4)(245,-3)(246,-4)(247,-5)(248,-4)(249,-5)(250,-4)(251,-3)(252,-4)(253,-3)(254,-2)(255,-3)(256,-2)(257,-1)(258,0)(259,1)(260,0)(261,-1)(262,0)(263,-1)(264,0)(265,1)(266,0)(267,-1)(268,-2)(269,-3)(270,-2)(271,-1)(272,-2)(273,-1)(274,-2)(275,-3)(276,-2)(277,-1)(278,-2)(279,-1)(280,-2)(281,-3)(282,-2)(283,-3)(284,-2)(285,-1)(286,-2)(287,-3)(288,-2)(289,-3)(290,-2)(291,-1)(292,-2)(293,-1)(294,-2)(295,-3)(296,-2)(297,-1)(298,0)(299,1)(300,0)(301,-1)(302,0)(303,1)(304,0)(305,1)(306,0)(307,-1)(308,-2)(309,-3)(310,-2)(311,-1)(312,-2)(313,-3)(314,-2)(315,-3)(316,-2)(317,-1)(318,0)(319,1)(320,0)(321,-1)(322,0)(323,1)(324,2)(325,3)(326,2)(327,1)(328,0)(329,-1)(330,0)(331,1)(332,0)(333,-1)(334,-2)(335,-3)(336,-2)(337,-1)(338,0)(339,1)(340,0)(341,1)(342,0)(343,-1)(344,-2)(345,-1)(346,-2)(347,-3)(348,-4)(349,-5)(350,-4)(351,-5)(352,-4)(353,-3)(354,-2)(355,-3)(356,-2)(357,-1)(358,0)(359,1)(360,0)(361,1)(362,0)(363,-1)(364,0)(365,1)(366,0)(367,-1)(368,-2)(369,-3)(370,-2)(371,-3)(372,-2)(373,-1)(374,-2)(375,-3)(376,-2)(377,-1)(378,0)(379,1)(380,0)(381,1)(382,0)(383,1)(384,0)(385,1)(386,0)(387,-1)(388,-2)(389,-3)(390,-2)(391,-3)(392,-2)(393,-3)(394,-2)(395,-3)(396,-2)(397,-1)(398,0)(399,1)(400,0)(401,1)(402,0)(403,1)(404,2)(405,3)(406,2)(407,1)(408,0)(409,-1)(410,0)(411,-1)(412,0)(413,-1)(414,-2)(415,-3)(416,-2)(417,-1)(418,0)(419,1)(420,0)(421,1)(422,2)(423,1)(424,0)(425,1)(426,0)(427,-1)(428,-2)(429,-3)(430,-2)(431,-3)(432,-4)(433,-3)(434,-2)(435,-3)(436,-2)(437,-1)(438,0)(439,1)(440,0)(441,1)(442,2)(443,1)(444,2)(445,3)(446,2)(447,1)(448,0)(449,-1)(450,0)(451,-1)(452,-2)(453,-1)(454,-2)(455,-3)(456,-2)(457,-1)(458,0)(459,1)(460,0)(461,1)(462,2)(463,3)(464,2)(465,3)(466,2)(467,1)(468,0)(469,-1)(470,0)(471,-1)(472,-2)(473,-3)(474,-2)(475,-3)(476,-2)(477,-1)(478,0)(479,1)(480,0)(481,1)(482,2)(483,3)(484,4)(485,5)(486,4)(487,3)(488,2)(489,3)(490,4)(491,3)(492,2)(493,1)(494,0)(495,-1)(496,0)(497,1)(498,2)(499,1)(500,0)(501,1)(502,0)(503,-1)(504,0)(505,1)(506,0)(507,-1)(508,-2)(509,-1)(510,0)(511,-1)(512,0)(513,1)(514,0)(515,-1)(516,0)(517,-1)(518,0)(519,-1)(520,-2)(521,-1)(522,-2)(523,-3)(524,-2)(525,-1)(526,-2)(527,-1)(528,-2)(529,-1)(530,0)(531,-1)(532,0)(533,1)(534,0)(535,-1)(536,0)(537,1)(538,2)(539,1)(540,0)(541,1)(542,0)(543,1)(544,0)(545,1)(546,0)(547,-1)(548,-2)(549,-1)(550,0)(551,-1)(552,0)(553,-1)(554,0)(555,-1)(556,-2)(557,-1)(558,0)(559,-1)(560,-2)(561,-1)(562,-2)(563,-1)(564,-2)(565,-1)(566,0)(567,-1)(568,-2)(569,-1)(570,0)(571,-1)(572,0)(573,-1)(574,0)(575,-1)(576,0)(577,-1)(578,0)(579,-1)(580,-2)(581,-1)(582,-2)(583,-1)(584,-2)(585,-1)(586,-2)(587,-1)(588,-2)(589,-1)(590,0)(591,-1)(592,0)(593,-1)(594,0)(595,-1)(596,0)(597,1)(598,2)(599,1)(600,0)(601,1)(602,0)(603,1)(604,2)(605,3)(606,2)(607,1)(608,0)(609,1)(610,2)(611,1)(612,2)(613,1)(614,0)(615,-1)(616,0)(617,1)(618,2)(619,1)(620,0)(621,1)(622,2)(623,1)(624,0)(625,1)(626,0)(627,-1)(628,-2)(629,-1)(630,0)(631,-1)(632,-2)(633,-1)(634,0)(635,-1)(636,-2)(637,-1)(638,0)(639,-1)(640,0)(641,1)(642,2)(643,1)(644,0)(645,1)(646,2)(647,1)(648,2)(649,3)(650,2)(651,1)(652,0)(653,1)(654,2)(655,1)(656,0)(657,1)(658,0)(659,-1)(660,0)(661,1)(662,2)(663,1)(664,0)(665,1)(666,2)(667,3)(668,4)(669,5)(670,4)(671,3)(672,4)(673,3)(674,4)(675,3)(676,2)(677,1)(678,0)(679,-1)(680,0)(681,1)(682,0)(683,1)(684,0)(685,1)(686,2)(687,3)(688,4)(689,5)(690,4)(691,3)(692,4)(693,5)(694,6)(695,5)(696,4)(697,3)(698,2)(699,1)(700,2)(701,3)(702,2)(703,1)(704,0)(705,1)(706,2)(707,3)(708,4)(709,5)(710,4)(711,5)(712,4)(713,3)(714,4)(715,3)(716,2)(717,1)(718,0)(719,-1)(720,0)(721,-1)(722,0)(723,1)(724,0)(725,1)(726,2)(727,3)(728,4)(729,5)(730,4)(731,5)(732,4)(733,5)(734,6)(735,5)(736,4)(737,3)(738,2)(739,1)(740,2)(741,1)(742,2)(743,1)(744,0)(745,1)(746,2)(747,3)(748,4)(749,5)(750,4)(751,5)(752,6)(753,5)(754,6)(755,5)(756,4)(757,3)(758,2)(759,1)(760,2)(761,1)(762,0)(763,1)(764,0)(765,1)(766,2)(767,3)(768,4)(769,5)(770,4)(771,5)(772,6)(773,7)(774,8)(775,7)(776,6)(777,5)(778,4)(779,3)(780,4)(781,3)(782,2)(783,1)(784,0)(785,1)(786,2)(787,3)(788,4)(789,5)(790,6)(791,5)(792,4)(793,3)(794,4)(795,3)(796,2)(797,1)(798,0)(799,-1)(800,-2)(801,-1)(802,0)(803,1)(804,0)(805,1)(806,2)(807,1)(808,0)(809,1)(810,2)(811,1)(812,0)(813,-1)(814,0)(815,-1)(816,-2)(817,-1)(818,0)(819,-1)(820,-2)(821,-1)(822,0)(823,1)(824,0)(825,1)(826,2)(827,1)(828,2)(829,1)(830,2)(831,1)(832,0)(833,-1)(834,0)(835,-1)(836,-2)(837,-1)(838,-2)(839,-1)(840,-2)(841,-1)(842,0)(843,1)(844,0)(845,1)(846,2)(847,1)(848,2)(849,3)(850,4)(851,3)(852,2)(853,1)(854,2)(855,1)(856,0)(857,1)(858,0)(859,-1)(860,-2)(861,-1)(862,0)(863,1)(864,0)(865,1)(866,2)(867,3)(868,2)(869,1)(870,2)(871,1)(872,0)(873,-1)(874,0)(875,-1)(876,-2)(877,-3)(878,-2)(879,-1)(880,-2)(881,-1)(882,0)(883,1)(884,0)(885,1)(886,2)(887,3)(888,2)(889,3)(890,4)(891,3)(892,2)(893,1)(894,2)(895,1)(896,0)(897,-1)(898,0)(899,-1)(900,-2)(901,-1)(902,0)(903,1)(904,0)(905,1)(906,2)(907,3)(908,4)(909,5)(910,6)(911,5)(912,4)(913,5)(914,6)(915,5)(916,4)(917,3)(918,2)(919,1)(920,0)(921,1)(922,2)(923,1)(924,0)(925,1)(926,2)(927,3)(928,4)(929,5)(930,6)(931,5)(932,6)(933,5)(934,6)(935,5)(936,4)(937,3)(938,2)(939,1)(940,0)(941,1)(942,0)(943,1)(944,0)(945,1)(946,2)(947,3)(948,4)(949,5)(950,6)(951,5)(952,6)(953,7)(954,8)(955,7)(956,6)(957,5)(958,4)(959,3)(960,2)(961,3)(962,2)(963,1)(964,0)(965,1)(966,2)(967,3)(968,4)(969,5)(970,6)(971,7)(972,6)(973,5)(974,6)(975,5)(976,4)(977,3)(978,2)(979,1)(980,0)(981,-1)(982,0)(983,1)(984,0)(985,1)(986,2)(987,3)(988,4)(989,5)(990,6)(991,7)(992,6)(993,7)(994,8)(995,7)(996,6)(997,5)(998,4)(999,3)(1000,2)(1001,1)(1002,2)(1003,1)(1004,0)(1005,1)(1006,2)(1007,3)(1008,4)(1009,5)(1010,6)(1011,7)(1012,8)(1013,7)(1014,8)(1015,7)(1016,6)(1017,5)(1018,4)(1019,3)(1020,2)(1021,1)(1022,0)(1023,1)(1024,0)

    };
    
  \end{axis}
\end{tikzpicture}

\vspace{0.2in}


    \begin{tikzpicture} \footnotesize
\begin{axis}[
    width = 3.0in,
    xlabel={CCR3 sequence for $n=10$},
    ylabel={\# 1s $-$ \# 0s in prefix},
    xmin=0, xmax=1024,
    ymin=-11, ymax=11,
    ytick pos=bottom,
    ymajorgrids=true,
    grid style=dashed,
]
              \addplot[
    color=blue,
    ]
    coordinates {
(0,0)(1,1)(2,2)(3,3)(4,4)(5,5)(6,6)(7,7)(8,8)(9,9)(10,10)(11,9)(12,8)(13,7)(14,6)(15,5)(16,4)(17,3)(18,2)(19,1)(20,0)(21,1)(22,2)(23,3)(24,4)(25,5)(26,4)(27,5)(28,6)(29,7)(30,8)(31,7)(32,6)(33,5)(34,4)(35,3)(36,4)(37,3)(38,2)(39,1)(40,0)(41,1)(42,2)(43,3)(44,4)(45,3)(46,2)(47,3)(48,4)(49,5)(50,6)(51,5)(52,4)(53,3)(54,2)(55,3)(56,4)(57,3)(58,2)(59,1)(60,0)(61,1)(62,2)(63,3)(64,4)(65,5)(66,6)(67,5)(68,6)(69,7)(70,8)(71,7)(72,6)(73,5)(74,4)(75,3)(76,2)(77,3)(78,2)(79,1)(80,0)(81,1)(82,2)(83,3)(84,2)(85,3)(86,4)(87,3)(88,4)(89,5)(90,6)(91,5)(92,4)(93,3)(94,4)(95,3)(96,2)(97,3)(98,2)(99,1)(100,0)(101,1)(102,2)(103,3)(104,4)(105,3)(106,4)(107,3)(108,4)(109,5)(110,6)(111,5)(112,4)(113,3)(114,2)(115,3)(116,2)(117,3)(118,2)(119,1)(120,0)(121,1)(122,2)(123,3)(124,2)(125,1)(126,2)(127,1)(128,2)(129,3)(130,4)(131,3)(132,2)(133,1)(134,2)(135,3)(136,2)(137,3)(138,2)(139,1)(140,0)(141,1)(142,2)(143,3)(144,4)(145,5)(146,4)(147,3)(148,4)(149,5)(150,6)(151,5)(152,4)(153,3)(154,2)(155,1)(156,2)(157,3)(158,2)(159,1)(160,0)(161,1)(162,2)(163,3)(164,2)(165,3)(166,2)(167,1)(168,2)(169,3)(170,4)(171,3)(172,2)(173,1)(174,2)(175,1)(176,2)(177,3)(178,2)(179,1)(180,0)(181,1)(182,2)(183,3)(184,4)(185,3)(186,2)(187,1)(188,2)(189,3)(190,4)(191,3)(192,2)(193,1)(194,0)(195,1)(196,2)(197,3)(198,2)(199,1)(200,0)(201,1)(202,2)(203,3)(204,4)(205,5)(206,6)(207,7)(208,6)(209,7)(210,8)(211,7)(212,6)(213,5)(214,4)(215,3)(216,2)(217,1)(218,2)(219,1)(220,0)(221,1)(222,2)(223,3)(224,2)(225,3)(226,4)(227,5)(228,4)(229,5)(230,6)(231,5)(232,4)(233,3)(234,4)(235,3)(236,2)(237,1)(238,2)(239,1)(240,0)(241,1)(242,2)(243,3)(244,4)(245,3)(246,4)(247,5)(248,4)(249,5)(250,6)(251,5)(252,4)(253,3)(254,2)(255,3)(256,2)(257,1)(258,2)(259,1)(260,0)(261,1)(262,2)(263,3)(264,2)(265,1)(266,2)(267,3)(268,2)(269,3)(270,4)(271,3)(272,2)(273,1)(274,2)(275,3)(276,2)(277,1)(278,2)(279,1)(280,0)(281,1)(282,2)(283,3)(284,4)(285,5)(286,4)(287,5)(288,4)(289,5)(290,6)(291,5)(292,4)(293,3)(294,2)(295,1)(296,2)(297,1)(298,2)(299,1)(300,0)(301,1)(302,2)(303,1)(304,2)(305,3)(306,2)(307,3)(308,2)(309,3)(310,4)(311,3)(312,2)(313,3)(314,2)(315,1)(316,2)(317,1)(318,2)(319,1)(320,0)(321,1)(322,2)(323,3)(324,2)(325,3)(326,2)(327,3)(328,2)(329,3)(330,4)(331,3)(332,2)(333,1)(334,2)(335,1)(336,2)(337,1)(338,2)(339,1)(340,0)(341,1)(342,2)(343,3)(344,4)(345,3)(346,2)(347,3)(348,2)(349,3)(350,4)(351,3)(352,2)(353,1)(354,0)(355,1)(356,2)(357,1)(358,2)(359,1)(360,0)(361,1)(362,2)(363,3)(364,4)(365,5)(366,6)(367,5)(368,4)(369,5)(370,6)(371,5)(372,4)(373,3)(374,2)(375,1)(376,0)(377,1)(378,2)(379,1)(380,0)(381,1)(382,2)(383,3)(384,2)(385,3)(386,4)(387,3)(388,2)(389,3)(390,4)(391,3)(392,2)(393,1)(394,2)(395,1)(396,0)(397,1)(398,2)(399,1)(400,0)(401,1)(402,2)(403,1)(404,0)(405,1)(406,2)(407,3)(408,4)(409,3)(410,4)(411,3)(412,2)(413,3)(414,4)(415,3)(416,2)(417,1)(418,0)(419,1)(420,0)(421,1)(422,2)(423,1)(424,0)(425,1)(426,2)(427,1)(428,2)(429,1)(430,2)(431,1)(432,0)(433,1)(434,2)(435,1)(436,0)(437,1)(438,0)(439,1)(440,0)(441,1)(442,2)(443,1)(444,0)(445,1)(446,2)(447,3)(448,2)(449,1)(450,2)(451,1)(452,0)(453,1)(454,2)(455,1)(456,0)(457,-1)(458,0)(459,1)(460,0)(461,1)(462,2)(463,1)(464,0)(465,1)(466,2)(467,3)(468,4)(469,5)(470,4)(471,3)(472,2)(473,3)(474,4)(475,3)(476,2)(477,1)(478,0)(479,-1)(480,0)(481,1)(482,2)(483,1)(484,0)(485,1)(486,2)(487,3)(488,2)(489,3)(490,2)(491,1)(492,0)(493,1)(494,2)(495,1)(496,0)(497,-1)(498,0)(499,-1)(500,0)(501,1)(502,2)(503,1)(504,0)(505,1)(506,2)(507,3)(508,4)(509,5)(510,6)(511,7)(512,8)(513,7)(514,8)(515,7)(516,6)(517,5)(518,4)(519,3)(520,2)(521,1)(522,0)(523,1)(524,0)(525,1)(526,2)(527,3)(528,4)(529,3)(530,4)(531,5)(532,6)(533,5)(534,6)(535,5)(536,4)(537,3)(538,2)(539,3)(540,2)(541,1)(542,0)(543,1)(544,0)(545,1)(546,2)(547,3)(548,4)(549,5)(550,4)(551,5)(552,6)(553,5)(554,6)(555,5)(556,4)(557,3)(558,2)(559,1)(560,2)(561,1)(562,0)(563,1)(564,0)(565,1)(566,2)(567,1)(568,2)(569,3)(570,2)(571,3)(572,4)(573,3)(574,4)(575,3)(576,2)(577,3)(578,2)(579,1)(580,2)(581,1)(582,0)(583,1)(584,0)(585,1)(586,2)(587,3)(588,2)(589,3)(590,2)(591,3)(592,4)(593,3)(594,4)(595,3)(596,2)(597,1)(598,2)(599,1)(600,2)(601,1)(602,0)(603,1)(604,0)(605,1)(606,2)(607,3)(608,4)(609,3)(610,2)(611,3)(612,4)(613,3)(614,4)(615,3)(616,2)(617,1)(618,0)(619,1)(620,2)(621,1)(622,0)(623,1)(624,0)(625,1)(626,2)(627,1)(628,2)(629,1)(630,0)(631,1)(632,2)(633,1)(634,2)(635,1)(636,0)(637,1)(638,0)(639,1)(640,2)(641,1)(642,0)(643,1)(644,0)(645,1)(646,2)(647,3)(648,4)(649,5)(650,6)(651,5)(652,6)(653,5)(654,6)(655,5)(656,4)(657,3)(658,2)(659,1)(660,0)(661,1)(662,0)(663,1)(664,0)(665,1)(666,2)(667,3)(668,2)(669,3)(670,4)(671,3)(672,4)(673,3)(674,4)(675,3)(676,2)(677,1)(678,2)(679,1)(680,0)(681,1)(682,0)(683,1)(684,0)(685,1)(686,2)(687,3)(688,4)(689,3)(690,4)(691,3)(692,4)(693,3)(694,4)(695,3)(696,2)(697,1)(698,0)(699,1)(700,0)(701,1)(702,0)(703,1)(704,0)(705,1)(706,2)(707,1)(708,2)(709,1)(710,2)(711,1)(712,2)(713,1)(714,2)(715,1)(716,0)(717,1)(718,0)(719,1)(720,0)(721,1)(722,0)(723,1)(724,0)(725,1)(726,2)(727,3)(728,2)(729,1)(730,2)(731,1)(732,2)(733,1)(734,2)(735,1)(736,0)(737,-1)(738,0)(739,1)(740,0)(741,1)(742,0)(743,1)(744,0)(745,1)(746,2)(747,3)(748,4)(749,5)(750,4)(751,3)(752,4)(753,3)(754,4)(755,3)(756,2)(757,1)(758,0)(759,-1)(760,0)(761,1)(762,0)(763,1)(764,0)(765,1)(766,2)(767,1)(768,2)(769,3)(770,2)(771,1)(772,2)(773,1)(774,2)(775,1)(776,0)(777,1)(778,0)(779,-1)(780,0)(781,1)(782,0)(783,1)(784,0)(785,1)(786,2)(787,3)(788,2)(789,3)(790,2)(791,1)(792,2)(793,1)(794,2)(795,1)(796,0)(797,-1)(798,0)(799,-1)(800,0)(801,1)(802,0)(803,1)(804,0)(805,1)(806,2)(807,3)(808,4)(809,3)(810,2)(811,1)(812,2)(813,1)(814,2)(815,1)(816,0)(817,-1)(818,-2)(819,-1)(820,0)(821,1)(822,0)(823,1)(824,0)(825,1)(826,2)(827,3)(828,4)(829,5)(830,6)(831,7)(832,6)(833,5)(834,6)(835,5)(836,4)(837,3)(838,2)(839,1)(840,0)(841,-1)(842,0)(843,1)(844,0)(845,1)(846,2)(847,3)(848,2)(849,3)(850,4)(851,5)(852,4)(853,3)(854,4)(855,3)(856,2)(857,1)(858,2)(859,1)(860,0)(861,-1)(862,0)(863,1)(864,0)(865,1)(866,2)(867,3)(868,4)(869,3)(870,4)(871,5)(872,4)(873,3)(874,4)(875,3)(876,2)(877,1)(878,0)(879,1)(880,0)(881,-1)(882,0)(883,1)(884,0)(885,1)(886,2)(887,3)(888,2)(889,1)(890,2)(891,3)(892,2)(893,1)(894,2)(895,1)(896,0)(897,-1)(898,0)(899,1)(900,0)(901,-1)(902,0)(903,1)(904,0)(905,1)(906,2)(907,3)(908,4)(909,5)(910,4)(911,5)(912,4)(913,3)(914,4)(915,3)(916,2)(917,1)(918,0)(919,-1)(920,0)(921,-1)(922,0)(923,1)(924,0)(925,1)(926,2)(927,3)(928,2)(929,3)(930,2)(931,3)(932,2)(933,1)(934,2)(935,1)(936,0)(937,-1)(938,0)(939,-1)(940,0)(941,-1)(942,0)(943,1)(944,0)(945,1)(946,2)(947,3)(948,4)(949,3)(950,2)(951,3)(952,2)(953,1)(954,2)(955,1)(956,0)(957,-1)(958,-2)(959,-1)(960,0)(961,-1)(962,0)(963,1)(964,0)(965,1)(966,2)(967,3)(968,4)(969,5)(970,6)(971,5)(972,4)(973,3)(974,4)(975,3)(976,2)(977,1)(978,0)(979,-1)(980,-2)(981,-1)(982,0)(983,1)(984,0)(985,1)(986,2)(987,3)(988,4)(989,3)(990,4)(991,3)(992,2)(993,1)(994,2)(995,1)(996,0)(997,-1)(998,-2)(999,-1)(1000,-2)(1001,-1)(1002,0)(1003,1)(1004,0)(1005,1)(1006,2)(1007,3)(1008,4)(1009,5)(1010,4)(1011,3)(1012,2)(1013,1)(1014,2)(1015,1)(1016,0)(1017,-1)(1018,-2)(1019,-3)(1020,-2)(1021,-1)(1022,0)(1023,1)(1024,0)

    };
\end{axis}
\end{tikzpicture}
%
 %
 \ \ \     
\begin{tikzpicture} \footnotesize
\begin{axis}[
    width = 3.0in,
    xlabel={Huang sequence for $n=10$},
    ylabel={\# 1s $-$ \# 0s in prefix},
    xmin=0, xmax=1024,
    ymin=-11, ymax=11,
    ytick pos=bottom,
    ymajorgrids=true,
    grid style=dashed,
]

                  \addplot[
    color=red,
    ]
    coordinates {
(0,0)(1,1)(2,2)(3,3)(4,4)(5,5)(6,6)(7,7)(8,8)(9,9)(10,10)(11,9)(12,8)(13,7)(14,6)(15,5)(16,4)(17,3)(18,2)(19,1)(20,0)(21,1)(22,0)(23,1)(24,0)(25,1)(26,0)(27,1)(28,0)(29,1)(30,2)(31,1)(32,2)(33,1)(34,2)(35,1)(36,2)(37,1)(38,2)(39,1)(40,0)(41,1)(42,0)(43,1)(44,0)(45,1)(46,0)(47,1)(48,2)(49,3)(50,4)(51,3)(52,4)(53,3)(54,4)(55,3)(56,4)(57,3)(58,2)(59,1)(60,0)(61,1)(62,0)(63,1)(64,0)(65,1)(66,2)(67,1)(68,2)(69,3)(70,4)(71,3)(72,4)(73,3)(74,4)(75,3)(76,2)(77,3)(78,2)(79,1)(80,0)(81,1)(82,0)(83,1)(84,0)(85,1)(86,2)(87,3)(88,4)(89,5)(90,6)(91,5)(92,6)(93,5)(94,6)(95,5)(96,4)(97,3)(98,2)(99,1)(100,0)(101,1)(102,0)(103,1)(104,2)(105,1)(106,2)(107,1)(108,2)(109,3)(110,4)(111,3)(112,4)(113,3)(114,2)(115,3)(116,2)(117,3)(118,2)(119,1)(120,0)(121,1)(122,0)(123,1)(124,2)(125,1)(126,2)(127,3)(128,4)(129,5)(130,6)(131,5)(132,6)(133,5)(134,4)(135,5)(136,4)(137,3)(138,2)(139,1)(140,0)(141,1)(142,0)(143,1)(144,2)(145,3)(146,2)(147,3)(148,4)(149,5)(150,6)(151,5)(152,6)(153,5)(154,4)(155,3)(156,4)(157,3)(158,2)(159,1)(160,0)(161,1)(162,0)(163,1)(164,2)(165,3)(166,4)(167,5)(168,6)(169,7)(170,8)(171,7)(172,8)(173,7)(174,6)(175,5)(176,4)(177,3)(178,2)(179,1)(180,0)(181,1)(182,2)(183,1)(184,0)(185,-1)(186,-2)(187,-1)(188,-2)(189,-1)(190,0)(191,-1)(192,-2)(193,-1)(194,-2)(195,-1)(196,0)(197,-1)(198,0)(199,-1)(200,-2)(201,-1)(202,0)(203,-1)(204,0)(205,-1)(206,-2)(207,-1)(208,-2)(209,-1)(210,0)(211,-1)(212,-2)(213,-1)(214,0)(215,1)(216,2)(217,1)(218,2)(219,1)(220,0)(221,1)(222,2)(223,1)(224,0)(225,-1)(226,-2)(227,-1)(228,0)(229,1)(230,2)(231,1)(232,0)(233,1)(234,0)(235,1)(236,2)(237,1)(238,0)(239,-1)(240,-2)(241,-1)(242,0)(243,-1)(244,0)(245,-1)(246,-2)(247,-1)(248,0)(249,1)(250,2)(251,1)(252,0)(253,1)(254,2)(255,3)(256,4)(257,3)(258,2)(259,1)(260,0)(261,1)(262,2)(263,1)(264,0)(265,-1)(266,0)(267,-1)(268,0)(269,1)(270,2)(271,1)(272,0)(273,1)(274,2)(275,3)(276,2)(277,3)(278,2)(279,1)(280,0)(281,1)(282,2)(283,1)(284,0)(285,-1)(286,0)(287,1)(288,0)(289,1)(290,2)(291,1)(292,0)(293,1)(294,2)(295,3)(296,2)(297,1)(298,2)(299,1)(300,0)(301,1)(302,2)(303,1)(304,0)(305,-1)(306,0)(307,1)(308,2)(309,3)(310,4)(311,3)(312,2)(313,3)(314,4)(315,5)(316,4)(317,3)(318,2)(319,1)(320,0)(321,1)(322,2)(323,1)(324,0)(325,1)(326,0)(327,1)(328,0)(329,1)(330,2)(331,1)(332,0)(333,1)(334,2)(335,1)(336,2)(337,1)(338,2)(339,1)(340,0)(341,1)(342,2)(343,1)(344,0)(345,1)(346,0)(347,1)(348,2)(349,3)(350,4)(351,3)(352,2)(353,3)(354,4)(355,3)(356,4)(357,3)(358,2)(359,1)(360,0)(361,1)(362,2)(363,1)(364,0)(365,1)(366,2)(367,1)(368,0)(369,1)(370,2)(371,1)(372,2)(373,3)(374,4)(375,3)(376,2)(377,3)(378,4)(379,3)(380,2)(381,3)(382,2)(383,1)(384,0)(385,1)(386,2)(387,1)(388,0)(389,1)(390,2)(391,3)(392,2)(393,3)(394,4)(395,3)(396,2)(397,3)(398,4)(399,3)(400,2)(401,1)(402,2)(403,1)(404,0)(405,1)(406,2)(407,1)(408,0)(409,1)(410,2)(411,3)(412,4)(413,5)(414,6)(415,5)(416,4)(417,5)(418,6)(419,5)(420,4)(421,3)(422,2)(423,1)(424,0)(425,1)(426,2)(427,1)(428,2)(429,1)(430,2)(431,1)(432,2)(433,3)(434,4)(435,3)(436,2)(437,3)(438,2)(439,3)(440,2)(441,3)(442,2)(443,1)(444,0)(445,1)(446,2)(447,1)(448,2)(449,1)(450,2)(451,3)(452,4)(453,5)(454,6)(455,5)(456,4)(457,5)(458,4)(459,5)(460,4)(461,3)(462,2)(463,1)(464,0)(465,1)(466,2)(467,1)(468,2)(469,3)(470,2)(471,3)(472,4)(473,5)(474,6)(475,5)(476,4)(477,5)(478,4)(479,3)(480,4)(481,3)(482,2)(483,1)(484,0)(485,1)(486,2)(487,1)(488,2)(489,3)(490,4)(491,5)(492,6)(493,7)(494,8)(495,7)(496,6)(497,7)(498,6)(499,5)(500,4)(501,3)(502,2)(503,1)(504,0)(505,1)(506,2)(507,3)(508,2)(509,1)(510,0)(511,1)(512,0)(513,1)(514,2)(515,1)(516,0)(517,-1)(518,0)(519,1)(520,2)(521,1)(522,2)(523,1)(524,0)(525,1)(526,2)(527,3)(528,2)(529,1)(530,0)(531,1)(532,2)(533,3)(534,4)(535,3)(536,2)(537,1)(538,2)(539,3)(540,4)(541,3)(542,2)(543,1)(544,0)(545,1)(546,2)(547,3)(548,2)(549,1)(550,2)(551,1)(552,2)(553,3)(554,4)(555,3)(556,2)(557,1)(558,2)(559,3)(560,2)(561,3)(562,2)(563,1)(564,0)(565,1)(566,2)(567,3)(568,2)(569,1)(570,2)(571,3)(572,2)(573,3)(574,4)(575,3)(576,2)(577,1)(578,2)(579,3)(580,2)(581,1)(582,2)(583,1)(584,0)(585,1)(586,2)(587,3)(588,2)(589,1)(590,2)(591,3)(592,4)(593,5)(594,6)(595,5)(596,4)(597,3)(598,4)(599,5)(600,4)(601,3)(602,2)(603,1)(604,0)(605,1)(606,2)(607,3)(608,2)(609,3)(610,2)(611,3)(612,2)(613,3)(614,4)(615,3)(616,2)(617,1)(618,2)(619,1)(620,2)(621,1)(622,2)(623,1)(624,0)(625,1)(626,2)(627,3)(628,2)(629,3)(630,2)(631,3)(632,4)(633,5)(634,6)(635,5)(636,4)(637,3)(638,4)(639,3)(640,4)(641,3)(642,2)(643,1)(644,0)(645,1)(646,2)(647,3)(648,2)(649,3)(650,4)(651,3)(652,4)(653,5)(654,6)(655,5)(656,4)(657,3)(658,4)(659,3)(660,2)(661,3)(662,2)(663,1)(664,0)(665,1)(666,2)(667,3)(668,2)(669,3)(670,4)(671,5)(672,6)(673,7)(674,8)(675,7)(676,6)(677,5)(678,6)(679,5)(680,4)(681,3)(682,2)(683,1)(684,0)(685,1)(686,2)(687,3)(688,4)(689,3)(690,4)(691,3)(692,4)(693,5)(694,6)(695,5)(696,4)(697,3)(698,2)(699,3)(700,2)(701,3)(702,2)(703,1)(704,0)(705,1)(706,2)(707,3)(708,4)(709,3)(710,4)(711,5)(712,6)(713,7)(714,8)(715,7)(716,6)(717,5)(718,4)(719,5)(720,4)(721,3)(722,2)(723,1)(724,0)(725,1)(726,2)(727,3)(728,4)(729,5)(730,4)(731,5)(732,4)(733,5)(734,6)(735,5)(736,4)(737,3)(738,2)(739,1)(740,2)(741,1)(742,2)(743,1)(744,0)(745,1)(746,0)(747,1)(748,2)(749,3)(750,2)(751,3)(752,2)(753,3)(754,4)(755,3)(756,4)(757,3)(758,2)(759,1)(760,2)(761,1)(762,2)(763,1)(764,0)(765,1)(766,2)(767,1)(768,2)(769,3)(770,2)(771,3)(772,2)(773,3)(774,4)(775,3)(776,2)(777,3)(778,2)(779,1)(780,2)(781,1)(782,2)(783,1)(784,0)(785,1)(786,2)(787,3)(788,4)(789,5)(790,4)(791,5)(792,6)(793,7)(794,8)(795,7)(796,6)(797,5)(798,4)(799,3)(800,4)(801,3)(802,2)(803,1)(804,0)(805,1)(806,2)(807,3)(808,4)(809,5)(810,6)(811,5)(812,6)(813,7)(814,8)(815,7)(816,6)(817,5)(818,4)(819,3)(820,2)(821,3)(822,2)(823,1)(824,0)(825,1)(826,0)(827,1)(828,2)(829,3)(830,4)(831,3)(832,4)(833,5)(834,6)(835,5)(836,6)(837,5)(838,4)(839,3)(840,2)(841,3)(842,2)(843,1)(844,0)(845,1)(846,2)(847,1)(848,2)(849,3)(850,4)(851,3)(852,4)(853,5)(854,6)(855,5)(856,4)(857,5)(858,4)(859,3)(860,2)(861,3)(862,2)(863,1)(864,0)(865,1)(866,2)(867,3)(868,4)(869,5)(870,6)(871,7)(872,6)(873,7)(874,8)(875,7)(876,6)(877,5)(878,4)(879,3)(880,2)(881,1)(882,2)(883,1)(884,0)(885,1)(886,0)(887,1)(888,0)(889,1)(890,2)(891,3)(892,2)(893,3)(894,4)(895,3)(896,4)(897,3)(898,4)(899,3)(900,2)(901,1)(902,2)(903,1)(904,0)(905,1)(906,0)(907,1)(908,2)(909,3)(910,4)(911,5)(912,4)(913,5)(914,6)(915,5)(916,6)(917,5)(918,4)(919,3)(920,2)(921,1)(922,2)(923,1)(924,0)(925,1)(926,2)(927,1)(928,2)(929,1)(930,2)(931,3)(932,2)(933,3)(934,4)(935,3)(936,2)(937,3)(938,2)(939,3)(940,2)(941,1)(942,2)(943,1)(944,0)(945,1)(946,2)(947,1)(948,2)(949,3)(950,4)(951,5)(952,4)(953,5)(954,6)(955,5)(956,4)(957,5)(958,4)(959,3)(960,2)(961,1)(962,2)(963,1)(964,0)(965,1)(966,2)(967,3)(968,2)(969,3)(970,4)(971,5)(972,4)(973,5)(974,6)(975,5)(976,4)(977,3)(978,4)(979,3)(980,2)(981,1)(982,2)(983,1)(984,0)(985,1)(986,2)(987,3)(988,4)(989,3)(990,4)(991,5)(992,4)(993,5)(994,6)(995,5)(996,4)(997,3)(998,2)(999,3)(1000,2)(1001,1)(1002,2)(1003,1)(1004,0)(1005,1)(1006,0)(1007,1)(1008,2)(1009,1)(1010,2)(1011,3)(1012,2)(1013,3)(1014,4)(1015,3)(1016,4)(1017,3)(1018,2)(1019,3)(1020,2)(1021,1)(1022,2)(1023,1)(1024,0)
    };
 
\end{axis}
\end{tikzpicture}

\end{center}

%% file: tables-same.tex
\begin{center}

\begin{tikzpicture} \footnotesize
\begin{axis}[
    width = 3.0in,
    xlabel={Pref-same sequence for $n=10$},
    ylabel={\# 1s $-$ \# 0s in prefix},
    xmin=0, xmax=1024,
    ymin=-12, ymax=22,
    ytick pos=bottom,
    legend pos=north east,
    ymajorgrids=true,
    grid style=dashed,
]

          \addplot[
    color=blue, 
    ]
    coordinates {
(0,0)(1,1)(2,2)(3,3)(4,4)(5,5)(6,6)(7,7)(8,8)(9,9)(10,10)(11,9)(12,8)(13,7)(14,6)(15,5)(16,4)(17,3)(18,2)(19,1)(20,0)(21,1)(22,2)(23,3)(24,4)(25,5)(26,6)(27,7)(28,8)(29,7)(30,8)(31,9)(32,10)(33,11)(34,12)(35,13)(36,14)(37,13)(38,12)(39,13)(40,14)(41,15)(42,16)(43,17)(44,18)(45,19)(46,18)(47,19)(48,18)(49,17)(50,16)(51,15)(52,14)(53,13)(54,12)(55,11)(56,12)(57,11)(58,10)(59,9)(60,8)(61,7)(62,6)(63,5)(64,6)(65,7)(66,6)(67,5)(68,4)(69,3)(70,2)(71,1)(72,0)(73,1)(74,0)(75,1)(76,2)(77,3)(78,4)(79,5)(80,6)(81,5)(82,4)(83,3)(84,4)(85,5)(86,6)(87,7)(88,8)(89,9)(90,8)(91,7)(92,8)(93,7)(94,6)(95,5)(96,4)(97,3)(98,2)(99,3)(100,4)(101,5)(102,4)(103,3)(104,2)(105,1)(106,0)(107,-1)(108,0)(109,1)(110,0)(111,1)(112,2)(113,3)(114,4)(115,5)(116,6)(117,5)(118,6)(119,7)(120,6)(121,5)(122,4)(123,3)(124,2)(125,1)(126,2)(127,1)(128,0)(129,1)(130,2)(131,3)(132,4)(133,5)(134,6)(135,5)(136,6)(137,5)(138,6)(139,7)(140,8)(141,9)(142,10)(143,9)(144,8)(145,7)(146,6)(147,7)(148,8)(149,9)(150,10)(151,11)(152,10)(153,9)(154,8)(155,9)(156,8)(157,7)(158,6)(159,5)(160,4)(161,5)(162,6)(163,7)(164,8)(165,7)(166,6)(167,5)(168,4)(169,3)(170,4)(171,5)(172,6)(173,5)(174,6)(175,7)(176,8)(177,9)(178,10)(179,9)(180,8)(181,9)(182,10)(183,9)(184,8)(185,7)(186,6)(187,5)(188,6)(189,7)(190,6)(191,5)(192,6)(193,7)(194,8)(195,9)(196,10)(197,9)(198,8)(199,9)(200,8)(201,9)(202,10)(203,11)(204,12)(205,13)(206,12)(207,13)(208,14)(209,15)(210,14)(211,13)(212,12)(213,11)(214,10)(215,11)(216,10)(217,9)(218,8)(219,9)(220,10)(221,11)(222,12)(223,13)(224,12)(225,13)(226,14)(227,13)(228,14)(229,15)(230,16)(231,17)(232,18)(233,17)(234,18)(235,17)(236,16)(237,17)(238,18)(239,19)(240,20)(241,21)(242,20)(243,21)(244,20)(245,21)(246,20)(247,19)(248,18)(249,17)(250,16)(251,15)(252,16)(253,15)(254,16)(255,15)(256,14)(257,13)(258,12)(259,11)(260,12)(261,13)(262,12)(263,13)(264,12)(265,11)(266,10)(267,9)(268,8)(269,9)(270,8)(271,7)(272,8)(273,7)(274,6)(275,5)(276,4)(277,3)(278,4)(279,3)(280,4)(281,5)(282,4)(283,3)(284,2)(285,1)(286,0)(287,1)(288,0)(289,1)(290,0)(291,1)(292,2)(293,3)(294,4)(295,3)(296,2)(297,1)(298,0)(299,1)(300,0)(301,-1)(302,-2)(303,-3)(304,-2)(305,-1)(306,0)(307,1)(308,0)(309,1)(310,2)(311,3)(312,4)(313,3)(314,2)(315,1)(316,2)(317,3)(318,2)(319,1)(320,0)(321,-1)(322,0)(323,1)(324,2)(325,1)(326,0)(327,1)(328,2)(329,3)(330,4)(331,3)(332,2)(333,1)(334,2)(335,1)(336,2)(337,3)(338,4)(339,5)(340,4)(341,3)(342,4)(343,5)(344,6)(345,5)(346,4)(347,3)(348,2)(349,3)(350,4)(351,3)(352,2)(353,1)(354,2)(355,3)(356,4)(357,5)(358,4)(359,3)(360,4)(361,5)(362,4)(363,5)(364,6)(365,7)(366,8)(367,7)(368,6)(369,7)(370,6)(371,5)(372,6)(373,7)(374,8)(375,9)(376,8)(377,7)(378,8)(379,7)(380,8)(381,7)(382,6)(383,5)(384,4)(385,5)(386,6)(387,7)(388,6)(389,7)(390,6)(391,5)(392,4)(393,3)(394,4)(395,5)(396,4)(397,3)(398,4)(399,3)(400,2)(401,1)(402,0)(403,1)(404,2)(405,1)(406,2)(407,3)(408,2)(409,1)(410,0)(411,-1)(412,0)(413,1)(414,0)(415,1)(416,0)(417,1)(418,2)(419,3)(420,4)(421,3)(422,4)(423,5)(424,6)(425,5)(426,6)(427,7)(428,8)(429,9)(430,8)(431,9)(432,10)(433,9)(434,8)(435,9)(436,10)(437,11)(438,12)(439,11)(440,12)(441,13)(442,12)(443,13)(444,12)(445,11)(446,10)(447,9)(448,10)(449,9)(450,8)(451,7)(452,8)(453,7)(454,6)(455,5)(456,4)(457,5)(458,4)(459,3)(460,4)(461,5)(462,4)(463,3)(464,2)(465,1)(466,2)(467,1)(468,0)(469,1)(470,0)(471,1)(472,2)(473,3)(474,4)(475,3)(476,4)(477,3)(478,2)(479,1)(480,2)(481,3)(482,4)(483,5)(484,4)(485,5)(486,4)(487,3)(488,4)(489,3)(490,2)(491,1)(492,0)(493,1)(494,0)(495,1)(496,2)(497,3)(498,2)(499,1)(500,0)(501,-1)(502,0)(503,-1)(504,0)(505,1)(506,0)(507,1)(508,2)(509,3)(510,4)(511,3)(512,4)(513,3)(514,4)(515,5)(516,4)(517,3)(518,2)(519,1)(520,2)(521,1)(522,2)(523,1)(524,0)(525,1)(526,2)(527,3)(528,4)(529,3)(530,4)(531,3)(532,4)(533,3)(534,4)(535,5)(536,6)(537,5)(538,4)(539,3)(540,4)(541,5)(542,6)(543,5)(544,4)(545,3)(546,4)(547,5)(548,4)(549,5)(550,6)(551,7)(552,6)(553,5)(554,4)(555,5)(556,4)(557,3)(558,4)(559,5)(560,6)(561,5)(562,4)(563,3)(564,4)(565,3)(566,4)(567,3)(568,2)(569,1)(570,2)(571,3)(572,4)(573,3)(574,2)(575,3)(576,2)(577,1)(578,0)(579,1)(580,2)(581,3)(582,2)(583,3)(584,4)(585,3)(586,2)(587,1)(588,2)(589,3)(590,4)(591,3)(592,4)(593,3)(594,4)(595,5)(596,6)(597,5)(598,4)(599,5)(600,6)(601,7)(602,6)(603,7)(604,8)(605,9)(606,8)(607,7)(608,8)(609,9)(610,8)(611,7)(612,8)(613,9)(614,10)(615,9)(616,8)(617,9)(618,10)(619,9)(620,10)(621,9)(622,8)(623,7)(624,8)(625,9)(626,8)(627,7)(628,6)(629,7)(630,6)(631,5)(632,4)(633,5)(634,6)(635,5)(636,4)(637,5)(638,6)(639,5)(640,4)(641,3)(642,4)(643,5)(644,4)(645,3)(646,4)(647,3)(648,4)(649,5)(650,6)(651,5)(652,4)(653,5)(654,4)(655,3)(656,4)(657,3)(658,2)(659,1)(660,2)(661,3)(662,2)(663,3)(664,4)(665,3)(666,4)(667,5)(668,6)(669,5)(670,4)(671,5)(672,4)(673,5)(674,6)(675,5)(676,4)(677,3)(678,4)(679,5)(680,4)(681,5)(682,4)(683,3)(684,4)(685,5)(686,6)(687,5)(688,4)(689,5)(690,4)(691,5)(692,4)(693,5)(694,6)(695,7)(696,6)(697,7)(698,8)(699,9)(700,8)(701,9)(702,8)(703,7)(704,6)(705,7)(706,6)(707,5)(708,4)(709,5)(710,4)(711,5)(712,6)(713,7)(714,6)(715,7)(716,8)(717,7)(718,6)(719,7)(720,6)(721,5)(722,4)(723,5)(724,4)(725,3)(726,4)(727,5)(728,4)(729,5)(730,6)(731,7)(732,6)(733,7)(734,8)(735,7)(736,8)(737,9)(738,8)(739,7)(740,6)(741,7)(742,6)(743,5)(744,6)(745,5)(746,4)(747,5)(748,6)(749,7)(750,6)(751,7)(752,8)(753,7)(754,8)(755,7)(756,8)(757,9)(758,10)(759,9)(760,10)(761,9)(762,8)(763,9)(764,10)(765,9)(766,8)(767,7)(768,8)(769,7)(770,8)(771,9)(772,8)(773,7)(774,8)(775,9)(776,10)(777,9)(778,10)(779,9)(780,8)(781,9)(782,8)(783,9)(784,10)(785,11)(786,10)(787,11)(788,10)(789,11)(790,12)(791,11)(792,12)(793,13)(794,14)(795,13)(796,14)(797,13)(798,14)(799,13)(800,12)(801,13)(802,14)(803,15)(804,14)(805,15)(806,14)(807,15)(808,14)(809,15)(810,14)(811,13)(812,12)(813,11)(814,12)(815,11)(816,12)(817,11)(818,12)(819,11)(820,10)(821,9)(822,10)(823,11)(824,10)(825,11)(826,10)(827,11)(828,10)(829,9)(830,8)(831,9)(832,8)(833,7)(834,8)(835,7)(836,8)(837,7)(838,6)(839,5)(840,6)(841,5)(842,6)(843,7)(844,6)(845,7)(846,6)(847,5)(848,4)(849,5)(850,4)(851,5)(852,4)(853,3)(854,4)(855,3)(856,2)(857,1)(858,2)(859,1)(860,2)(861,1)(862,2)(863,3)(864,2)(865,1)(866,0)(867,1)(868,0)(869,1)(870,0)(871,1)(872,0)(873,1)(874,2)(875,1)(876,0)(877,1)(878,2)(879,1)(880,0)(881,1)(882,0)(883,-1)(884,0)(885,1)(886,0)(887,-1)(888,0)(889,1)(890,0)(891,1)(892,2)(893,1)(894,0)(895,1)(896,2)(897,1)(898,2)(899,1)(900,2)(901,3)(902,2)(903,1)(904,2)(905,1)(906,0)(907,1)(908,0)(909,1)(910,2)(911,1)(912,0)(913,1)(914,0)(915,1)(916,2)(917,1)(918,2)(919,3)(920,2)(921,1)(922,2)(923,1)(924,2)(925,1)(926,0)(927,1)(928,2)(929,1)(930,0)(931,1)(932,0)(933,1)(934,0)(935,1)(936,0)(937,-1)(938,0)(939,1)(940,0)(941,1)(942,2)(943,1)(944,2)(945,1)(946,0)(947,1)(948,2)(949,1)(950,2)(951,1)(952,0)(953,1)(954,0)(955,-1)(956,0)(957,1)(958,0)(959,1)(960,0)(961,1)(962,0)(963,1)(964,2)(965,1)(966,2)(967,3)(968,2)(969,3)(970,4)(971,3)(972,4)(973,3)(974,4)(975,3)(976,2)(977,3)(978,2)(979,1)(980,2)(981,1)(982,0)(983,1)(984,0)(985,1)(986,0)(987,1)(988,2)(989,1)(990,2)(991,1)(992,0)(993,1)(994,0)(995,1)(996,0)(997,-1)(998,0)(999,-1)(1000,0)(1001,1)(1002,0)(1003,1)(1004,0)(1005,1)(1006,2)(1007,1)(1008,2)(1009,1)(1010,2)(1011,1)(1012,2)(1013,1)(1014,2)(1015,1)(1016,0)(1017,1)(1018,0)(1019,1)(1020,0)(1021,1)(1022,0)(1023,1)(1024,0)
    };
 %
 \end{axis}
\end{tikzpicture}
%
\ \ \ 
\begin{tikzpicture} \footnotesize
\begin{axis}[
    width = 3.0in,
    xlabel={Pref-opp sequence for $n=10$},
    ylabel={\# 1s $-$ \# 0s in prefix},
    xmin=0, xmax=1024,
    ymin=-12, ymax=22,
    ytick pos=bottom,
    legend pos=north east,
    ymajorgrids=true,
    grid style=dashed,
]

      \addplot[
    color=red,
    ]
    coordinates {
(0,0)(1,1)(2,0)(3,1)(4,0)(5,1)(6,0)(7,1)(8,0)(9,1)(10,0)(11,1)(12,2)(13,1)(14,2)(15,1)(16,2)(17,1)(18,2)(19,1)(20,0)(21,1)(22,0)(23,1)(24,0)(25,1)(26,0)(27,1)(28,0)(29,-1)(30,-2)(31,-1)(32,-2)(33,-1)(34,-2)(35,-1)(36,-2)(37,-1)(38,0)(39,1)(40,0)(41,1)(42,0)(43,1)(44,0)(45,1)(46,2)(47,1)(48,0)(49,1)(50,0)(51,1)(52,0)(53,1)(54,0)(55,-1)(56,0)(57,1)(58,0)(59,1)(60,0)(61,1)(62,0)(63,1)(64,2)(65,3)(66,4)(67,3)(68,4)(69,3)(70,4)(71,3)(72,2)(73,3)(74,2)(75,1)(76,2)(77,1)(78,2)(79,1)(80,2)(81,3)(82,2)(83,3)(84,4)(85,3)(86,4)(87,3)(88,4)(89,3)(90,2)(91,3)(92,4)(93,5)(94,4)(95,5)(96,4)(97,5)(98,4)(99,3)(100,2)(101,3)(102,4)(103,3)(104,4)(105,3)(106,4)(107,3)(108,2)(109,1)(110,0)(111,1)(112,0)(113,1)(114,0)(115,1)(116,0)(117,-1)(118,-2)(119,-3)(120,-4)(121,-3)(122,-4)(123,-3)(124,-4)(125,-3)(126,-2)(127,-3)(128,-4)(129,-5)(130,-4)(131,-5)(132,-4)(133,-5)(134,-4)(135,-3)(136,-2)(137,-3)(138,-4)(139,-3)(140,-4)(141,-3)(142,-4)(143,-3)(144,-2)(145,-1)(146,0)(147,1)(148,0)(149,1)(150,0)(151,1)(152,2)(153,1)(154,2)(155,1)(156,0)(157,1)(158,0)(159,1)(160,0)(161,-1)(162,0)(163,-1)(164,0)(165,1)(166,0)(167,1)(168,0)(169,1)(170,2)(171,1)(172,2)(173,3)(174,4)(175,3)(176,4)(177,3)(178,4)(179,5)(180,4)(181,3)(182,4)(183,5)(184,4)(185,5)(186,4)(187,5)(188,6)(189,5)(190,4)(191,3)(192,2)(193,3)(194,2)(195,3)(196,2)(197,1)(198,2)(199,1)(200,0)(201,-1)(202,0)(203,-1)(204,0)(205,-1)(206,-2)(207,-1)(208,0)(209,-1)(210,-2)(211,-1)(212,-2)(213,-1)(214,-2)(215,-3)(216,-2)(217,-1)(218,0)(219,1)(220,0)(221,1)(222,0)(223,1)(224,2)(225,3)(226,2)(227,3)(228,4)(229,3)(230,4)(231,3)(232,4)(233,5)(234,6)(235,5)(236,4)(237,3)(238,4)(239,3)(240,4)(241,3)(242,2)(243,1)(244,2)(245,1)(246,0)(247,1)(248,0)(249,1)(250,0)(251,-1)(252,-2)(253,-1)(254,0)(255,1)(256,0)(257,1)(258,0)(259,1)(260,2)(261,3)(262,4)(263,3)(264,2)(265,3)(266,2)(267,3)(268,2)(269,1)(270,0)(271,-1)(272,0)(273,1)(274,0)(275,1)(276,0)(277,1)(278,2)(279,3)(280,4)(281,5)(282,6)(283,5)(284,6)(285,5)(286,4)(287,5)(288,4)(289,5)(290,6)(291,7)(292,6)(293,7)(294,6)(295,5)(296,6)(297,5)(298,4)(299,5)(300,6)(301,5)(302,6)(303,5)(304,4)(305,5)(306,4)(307,3)(308,2)(309,1)(310,2)(311,1)(312,2)(313,3)(314,2)(315,3)(316,2)(317,1)(318,0)(319,1)(320,0)(321,1)(322,2)(323,1)(324,2)(325,3)(326,2)(327,1)(328,2)(329,1)(330,2)(331,3)(332,2)(333,3)(334,4)(335,5)(336,6)(337,5)(338,6)(339,5)(340,4)(341,5)(342,6)(343,5)(344,6)(345,7)(346,6)(347,7)(348,6)(349,5)(350,6)(351,7)(352,6)(353,5)(354,4)(355,5)(356,4)(357,5)(358,6)(359,5)(360,4)(361,5)(362,4)(363,3)(364,4)(365,3)(366,4)(367,5)(368,4)(369,3)(370,4)(371,5)(372,6)(373,5)(374,6)(375,5)(376,4)(377,5)(378,6)(379,7)(380,6)(381,5)(382,6)(383,5)(384,6)(385,7)(386,6)(387,5)(388,4)(389,5)(390,6)(391,5)(392,6)(393,5)(394,4)(395,5)(396,6)(397,7)(398,8)(399,9)(400,8)(401,9)(402,8)(403,7)(404,6)(405,7)(406,6)(407,5)(408,4)(409,5)(410,4)(411,5)(412,6)(413,7)(414,6)(415,7)(416,8)(417,9)(418,8)(419,9)(420,8)(421,7)(422,6)(423,7)(424,8)(425,7)(426,6)(427,7)(428,6)(429,7)(430,8)(431,9)(432,8)(433,7)(434,8)(435,9)(436,8)(437,9)(438,8)(439,7)(440,6)(441,7)(442,8)(443,9)(444,10)(445,9)(446,10)(447,9)(448,8)(449,7)(450,6)(451,7)(452,6)(453,5)(454,6)(455,5)(456,6)(457,7)(458,8)(459,9)(460,8)(461,9)(462,10)(463,9)(464,10)(465,9)(466,8)(467,7)(468,6)(469,7)(470,8)(471,9)(472,8)(473,9)(474,8)(475,7)(476,6)(477,5)(478,4)(479,5)(480,6)(481,5)(482,6)(483,5)(484,4)(485,3)(486,2)(487,1)(488,0)(489,1)(490,0)(491,1)(492,0)(493,-1)(494,-2)(495,-3)(496,-4)(497,-5)(498,-6)(499,-5)(500,-6)(501,-5)(502,-4)(503,-5)(504,-6)(505,-7)(506,-8)(507,-9)(508,-8)(509,-9)(510,-8)(511,-7)(512,-6)(513,-7)(514,-8)(515,-9)(516,-10)(517,-9)(518,-10)(519,-9)(520,-8)(521,-7)(522,-6)(523,-7)(524,-8)(525,-9)(526,-8)(527,-9)(528,-8)(529,-7)(530,-6)(531,-5)(532,-4)(533,-5)(534,-6)(535,-5)(536,-6)(537,-5)(538,-4)(539,-3)(540,-2)(541,-1)(542,0)(543,1)(544,0)(545,1)(546,2)(547,1)(548,2)(549,3)(550,2)(551,3)(552,4)(553,3)(554,2)(555,1)(556,2)(557,1)(558,0)(559,1)(560,0)(561,-1)(562,0)(563,-1)(564,-2)(565,-1)(566,0)(567,1)(568,0)(569,1)(570,2)(571,1)(572,2)(573,3)(574,4)(575,3)(576,2)(577,3)(578,2)(579,1)(580,2)(581,1)(582,0)(583,-1)(584,0)(585,1)(586,0)(587,1)(588,2)(589,1)(590,2)(591,3)(592,4)(593,5)(594,6)(595,5)(596,6)(597,7)(598,6)(599,5)(600,6)(601,5)(602,4)(603,3)(604,4)(605,3)(606,2)(607,3)(608,4)(609,3)(610,4)(611,5)(612,6)(613,5)(614,6)(615,7)(616,6)(617,5)(618,6)(619,7)(620,6)(621,5)(622,6)(623,5)(624,4)(625,5)(626,6)(627,5)(628,4)(629,5)(630,6)(631,5)(632,6)(633,7)(634,6)(635,5)(636,6)(637,7)(638,8)(639,9)(640,8)(641,9)(642,10)(643,9)(644,8)(645,7)(646,8)(647,9)(648,10)(649,9)(650,10)(651,11)(652,10)(653,9)(654,8)(655,7)(656,8)(657,9)(658,8)(659,9)(660,10)(661,9)(662,8)(663,7)(664,6)(665,5)(666,4)(667,5)(668,4)(669,3)(670,4)(671,3)(672,2)(673,1)(674,0)(675,-1)(676,0)(677,-1)(678,-2)(679,-1)(680,0)(681,-1)(682,-2)(683,-3)(684,-4)(685,-3)(686,-4)(687,-5)(688,-4)(689,-3)(690,-2)(691,-3)(692,-4)(693,-5)(694,-4)(695,-5)(696,-6)(697,-5)(698,-4)(699,-3)(700,-2)(701,-3)(702,-4)(703,-3)(704,-4)(705,-5)(706,-4)(707,-3)(708,-2)(709,-1)(710,0)(711,1)(712,0)(713,1)(714,2)(715,3)(716,2)(717,3)(718,4)(719,5)(720,6)(721,5)(722,6)(723,7)(724,8)(725,7)(726,6)(727,7)(728,8)(729,9)(730,8)(731,9)(732,10)(733,11)(734,10)(735,9)(736,8)(737,9)(738,10)(739,9)(740,10)(741,11)(742,12)(743,11)(744,10)(745,9)(746,8)(747,7)(748,8)(749,7)(750,6)(751,5)(752,6)(753,5)(754,4)(755,3)(756,2)(757,3)(758,2)(759,1)(760,0)(761,1)(762,2)(763,1)(764,0)(765,-1)(766,0)(767,-1)(768,-2)(769,-3)(770,-2)(771,-1)(772,0)(773,-1)(774,-2)(775,-1)(776,-2)(777,-3)(778,-4)(779,-3)(780,-2)(781,-1)(782,0)(783,1)(784,0)(785,1)(786,2)(787,3)(788,4)(789,3)(790,2)(791,3)(792,4)(793,3)(794,4)(795,5)(796,6)(797,7)(798,6)(799,5)(800,4)(801,3)(802,4)(803,3)(804,2)(805,1)(806,0)(807,1)(808,2)(809,1)(810,0)(811,1)(812,0)(813,-1)(814,-2)(815,-3)(816,-2)(817,-1)(818,0)(819,1)(820,0)(821,1)(822,2)(823,3)(824,4)(825,5)(826,4)(827,3)(828,2)(829,3)(830,2)(831,1)(832,0)(833,-1)(834,-2)(835,-1)(836,0)(837,1)(838,0)(839,1)(840,2)(841,3)(842,4)(843,5)(844,6)(845,5)(846,4)(847,5)(848,4)(849,3)(850,2)(851,1)(852,0)(853,-1)(854,0)(855,1)(856,0)(857,1)(858,2)(859,3)(860,4)(861,5)(862,6)(863,7)(864,8)(865,7)(866,6)(867,7)(868,8)(869,7)(870,6)(871,7)(872,8)(873,9)(874,8)(875,7)(876,8)(877,9)(878,8)(879,7)(880,6)(881,7)(882,8)(883,7)(884,6)(885,7)(886,8)(887,7)(888,6)(889,5)(890,4)(891,3)(892,4)(893,5)(894,4)(895,3)(896,4)(897,5)(898,6)(899,7)(900,8)(901,7)(902,6)(903,7)(904,8)(905,9)(906,8)(907,7)(908,6)(909,5)(910,6)(911,7)(912,6)(913,5)(914,4)(915,5)(916,6)(917,7)(918,8)(919,7)(920,6)(921,7)(922,8)(923,9)(924,10)(925,9)(926,8)(927,7)(928,8)(929,9)(930,8)(931,7)(932,6)(933,5)(934,6)(935,7)(936,8)(937,7)(938,6)(939,7)(940,8)(941,9)(942,10)(943,11)(944,12)(945,13)(946,12)(947,11)(948,10)(949,11)(950,12)(951,13)(952,12)(953,11)(954,10)(955,11)(956,12)(957,13)(958,14)(959,15)(960,16)(961,15)(962,14)(963,13)(964,12)(965,13)(966,14)(967,15)(968,16)(969,17)(970,16)(971,15)(972,14)(973,13)(974,12)(975,13)(976,14)(977,15)(978,16)(979,15)(980,14)(981,13)(982,12)(983,11)(984,10)(985,11)(986,12)(987,13)(988,12)(989,11)(990,10)(991,9)(992,8)(993,7)(994,6)(995,7)(996,8)(997,7)(998,6)(999,5)(1000,4)(1001,3)(1002,2)(1003,1)(1004,0)(1005,1)(1006,0)(1007,-1)(1008,-2)(1009,-3)(1010,-4)(1011,-5)(1012,-6)(1013,-7)(1014,-8)(1015,-9)(1016,-8)(1017,-7)(1018,-6)(1019,-5)(1020,-4)(1021,-3)(1022,-2)(1023,-1)(1024,0)
    };
    
\end{axis}
\end{tikzpicture}

\end{center}

%% file: tables-PCR.tex
\begin{center}

\begin{tikzpicture} \footnotesize
\begin{axis}[
    width = 3.0in,
    xlabel={PCR1 sequence for $n=10$},
    ylabel={\# 1s $-$ \# 0s in prefix},
    xmin=0, xmax=1024,
    ymin=-25, ymax=125,
    ytick={0,25, 50,75,100},
    ytick pos=bottom,
    legend pos=north east,
    ymajorgrids=true,
    grid style=dashed,
]

      \addplot[
    color=blue,
    ]
    coordinates {
(0,0)(1,1)(2,2)(3,3)(4,4)(5,5)(6,6)(7,7)(8,8)(9,9)(10,10)(11,9)(12,10)(13,11)(14,12)(15,13)(16,14)(17,15)(18,16)(19,17)(20,16)(21,15)(22,16)(23,17)(24,18)(25,19)(26,20)(27,21)(28,22)(29,21)(30,22)(31,21)(32,22)(33,23)(34,24)(35,25)(36,26)(37,27)(38,28)(39,27)(40,26)(41,25)(42,26)(43,27)(44,28)(45,29)(46,30)(47,31)(48,30)(49,31)(50,32)(51,31)(52,32)(53,33)(54,34)(55,35)(56,36)(57,37)(58,36)(59,37)(60,36)(61,35)(62,36)(63,37)(64,38)(65,39)(66,40)(67,41)(68,40)(69,39)(70,40)(71,39)(72,40)(73,41)(74,42)(75,43)(76,44)(77,45)(78,44)(79,43)(80,42)(81,41)(82,42)(83,43)(84,44)(85,45)(86,46)(87,45)(88,46)(89,47)(90,48)(91,47)(92,48)(93,49)(94,50)(95,51)(96,52)(97,51)(98,52)(99,53)(100,52)(101,51)(102,52)(103,53)(104,54)(105,55)(106,56)(107,55)(108,56)(109,55)(110,56)(111,55)(112,56)(113,57)(114,58)(115,59)(116,60)(117,59)(118,60)(119,59)(120,58)(121,57)(122,58)(123,59)(124,60)(125,61)(126,62)(127,61)(128,60)(129,61)(130,62)(131,61)(132,62)(133,63)(134,64)(135,65)(136,66)(137,65)(138,64)(139,65)(140,64)(141,63)(142,64)(143,65)(144,66)(145,67)(146,68)(147,67)(148,66)(149,65)(150,66)(151,65)(152,66)(153,67)(154,68)(155,69)(156,70)(157,69)(158,68)(159,67)(160,66)(161,65)(162,66)(163,67)(164,68)(165,69)(166,68)(167,69)(168,70)(169,71)(170,72)(171,71)(172,72)(173,73)(174,74)(175,73)(176,72)(177,73)(178,74)(179,75)(180,76)(181,75)(182,76)(183,77)(184,76)(185,77)(186,76)(187,77)(188,78)(189,79)(190,80)(191,79)(192,80)(193,81)(194,80)(195,79)(196,78)(197,79)(198,80)(199,81)(200,82)(201,81)(202,82)(203,81)(204,82)(205,83)(206,82)(207,83)(208,84)(209,85)(210,86)(211,85)(212,86)(213,85)(214,86)(215,85)(216,84)(217,85)(218,86)(219,87)(220,88)(221,87)(222,88)(223,87)(224,86)(225,87)(226,86)(227,87)(228,88)(229,89)(230,90)(231,89)(232,90)(233,89)(234,88)(235,87)(236,86)(237,87)(238,88)(239,89)(240,90)(241,89)(242,88)(243,89)(244,90)(245,91)(246,90)(247,91)(248,92)(249,93)(250,94)(251,93)(252,92)(253,93)(254,94)(255,93)(256,92)(257,93)(258,94)(259,95)(260,96)(261,95)(262,94)(263,95)(264,94)(265,95)(266,94)(267,95)(268,96)(269,97)(270,98)(271,97)(272,96)(273,97)(274,96)(275,95)(276,94)(277,95)(278,96)(279,97)(280,98)(281,97)(282,96)(283,95)(284,96)(285,97)(286,96)(287,97)(288,98)(289,99)(290,100)(291,99)(292,98)(293,97)(294,98)(295,97)(296,96)(297,97)(298,98)(299,99)(300,100)(301,99)(302,98)(303,97)(304,96)(305,97)(306,96)(307,97)(308,98)(309,99)(310,100)(311,99)(312,98)(313,97)(314,96)(315,95)(316,94)(317,95)(318,96)(319,97)(320,96)(321,97)(322,98)(323,99)(324,98)(325,99)(326,98)(327,99)(328,100)(329,101)(330,100)(331,101)(332,102)(333,103)(334,102)(335,101)(336,100)(337,101)(338,102)(339,103)(340,102)(341,103)(342,104)(343,103)(344,104)(345,105)(346,104)(347,105)(348,106)(349,107)(350,106)(351,107)(352,108)(353,107)(354,108)(355,107)(356,106)(357,107)(358,108)(359,109)(360,108)(361,109)(362,110)(363,109)(364,108)(365,109)(366,108)(367,109)(368,110)(369,111)(370,110)(371,111)(372,112)(373,111)(374,110)(375,109)(376,108)(377,109)(378,110)(379,111)(380,110)(381,111)(382,110)(383,111)(384,112)(385,111)(386,110)(387,111)(388,112)(389,113)(390,112)(391,113)(392,112)(393,113)(394,112)(395,113)(396,112)(397,113)(398,114)(399,115)(400,114)(401,115)(402,114)(403,115)(404,114)(405,113)(406,112)(407,113)(408,114)(409,115)(410,114)(411,115)(412,114)(413,113)(414,114)(415,115)(416,114)(417,115)(418,116)(419,117)(420,116)(421,117)(422,116)(423,115)(424,116)(425,115)(426,114)(427,115)(428,116)(429,117)(430,116)(431,117)(432,116)(433,115)(434,114)(435,115)(436,114)(437,115)(438,116)(439,117)(440,116)(441,117)(442,116)(443,115)(444,114)(445,113)(446,112)(447,113)(448,114)(449,115)(450,114)(451,113)(452,114)(453,115)(454,116)(455,115)(456,114)(457,115)(458,116)(459,115)(460,116)(461,115)(462,116)(463,117)(464,118)(465,117)(466,116)(467,117)(468,118)(469,117)(470,116)(471,115)(472,116)(473,117)(474,118)(475,117)(476,116)(477,117)(478,116)(479,117)(480,118)(481,117)(482,118)(483,119)(484,120)(485,119)(486,118)(487,119)(488,118)(489,119)(490,118)(491,117)(492,118)(493,119)(494,120)(495,119)(496,118)(497,119)(498,118)(499,117)(500,118)(501,117)(502,118)(503,119)(504,120)(505,119)(506,118)(507,119)(508,118)(509,117)(510,116)(511,115)(512,116)(513,117)(514,118)(515,117)(516,116)(517,115)(518,116)(519,117)(520,116)(521,115)(522,116)(523,117)(524,118)(525,117)(526,116)(527,115)(528,116)(529,115)(530,116)(531,115)(532,116)(533,117)(534,118)(535,117)(536,116)(537,115)(538,116)(539,115)(540,114)(541,113)(542,114)(543,115)(544,116)(545,115)(546,114)(547,113)(548,112)(549,113)(550,114)(551,113)(552,114)(553,115)(554,116)(555,115)(556,114)(557,113)(558,112)(559,113)(560,112)(561,111)(562,112)(563,113)(564,114)(565,113)(566,112)(567,111)(568,110)(569,109)(570,110)(571,109)(572,110)(573,111)(574,112)(575,111)(576,110)(577,109)(578,108)(579,107)(580,106)(581,105)(582,106)(583,107)(584,106)(585,107)(586,108)(587,107)(588,108)(589,109)(590,108)(591,107)(592,108)(593,109)(594,108)(595,109)(596,110)(597,109)(598,110)(599,109)(600,110)(601,109)(602,110)(603,111)(604,110)(605,111)(606,112)(607,111)(608,112)(609,111)(610,110)(611,109)(612,110)(613,111)(614,110)(615,111)(616,112)(617,111)(618,110)(619,111)(620,110)(621,109)(622,110)(623,111)(624,110)(625,111)(626,112)(627,111)(628,110)(629,109)(630,110)(631,109)(632,110)(633,111)(634,110)(635,111)(636,112)(637,111)(638,110)(639,109)(640,108)(641,107)(642,108)(643,109)(644,108)(645,109)(646,108)(647,109)(648,110)(649,109)(650,110)(651,109)(652,110)(653,111)(654,110)(655,109)(656,108)(657,109)(658,110)(659,109)(660,110)(661,109)(662,110)(663,109)(664,110)(665,109)(666,108)(667,109)(668,110)(669,109)(670,110)(671,109)(672,110)(673,109)(674,108)(675,109)(676,108)(677,109)(678,110)(679,109)(680,110)(681,109)(682,110)(683,109)(684,108)(685,107)(686,106)(687,107)(688,108)(689,107)(690,108)(691,107)(692,106)(693,107)(694,108)(695,107)(696,106)(697,107)(698,108)(699,107)(700,108)(701,107)(702,106)(703,107)(704,106)(705,107)(706,106)(707,107)(708,108)(709,107)(710,108)(711,107)(712,106)(713,107)(714,106)(715,105)(716,104)(717,105)(718,106)(719,105)(720,106)(721,105)(722,104)(723,103)(724,104)(725,103)(726,102)(727,103)(728,104)(729,103)(730,104)(731,103)(732,102)(733,101)(734,100)(735,101)(736,100)(737,101)(738,102)(739,101)(740,102)(741,101)(742,100)(743,99)(744,98)(745,97)(746,96)(747,97)(748,98)(749,97)(750,96)(751,97)(752,98)(753,97)(754,96)(755,97)(756,96)(757,97)(758,98)(759,97)(760,96)(761,97)(762,98)(763,97)(764,96)(765,95)(766,94)(767,95)(768,96)(769,95)(770,94)(771,95)(772,94)(773,95)(774,94)(775,95)(776,94)(777,95)(778,96)(779,95)(780,94)(781,95)(782,94)(783,95)(784,94)(785,93)(786,92)(787,93)(788,94)(789,93)(790,92)(791,93)(792,92)(793,91)(794,92)(795,91)(796,90)(797,91)(798,92)(799,91)(800,90)(801,91)(802,90)(803,89)(804,88)(805,89)(806,88)(807,89)(808,90)(809,89)(810,88)(811,89)(812,88)(813,87)(814,86)(815,85)(816,84)(817,85)(818,86)(819,85)(820,84)(821,83)(822,84)(823,85)(824,84)(825,83)(826,82)(827,83)(828,82)(829,83)(830,82)(831,81)(832,82)(833,83)(834,82)(835,81)(836,80)(837,81)(838,80)(839,79)(840,80)(841,79)(842,80)(843,81)(844,80)(845,79)(846,78)(847,79)(848,78)(849,77)(850,76)(851,75)(852,76)(853,77)(854,76)(855,75)(856,74)(857,73)(858,74)(859,73)(860,74)(861,73)(862,74)(863,75)(864,74)(865,73)(866,72)(867,71)(868,72)(869,71)(870,70)(871,69)(872,70)(873,71)(874,70)(875,69)(876,68)(877,67)(878,66)(879,67)(880,66)(881,65)(882,66)(883,67)(884,66)(885,65)(886,64)(887,63)(888,62)(889,61)(890,62)(891,61)(892,62)(893,63)(894,62)(895,61)(896,60)(897,59)(898,58)(899,57)(900,56)(901,55)(902,56)(903,55)(904,56)(905,55)(906,56)(907,55)(908,56)(909,55)(910,56)(911,55)(912,54)(913,53)(914,54)(915,53)(916,54)(917,53)(918,54)(919,53)(920,52)(921,53)(922,52)(923,51)(924,52)(925,51)(926,52)(927,51)(928,52)(929,51)(930,50)(931,49)(932,48)(933,47)(934,48)(935,47)(936,48)(937,47)(938,46)(939,47)(940,46)(941,47)(942,46)(943,45)(944,46)(945,45)(946,44)(947,43)(948,42)(949,43)(950,42)(951,43)(952,42)(953,41)(954,40)(955,41)(956,40)(957,39)(958,38)(959,39)(960,38)(961,39)(962,38)(963,37)(964,36)(965,35)(966,36)(967,35)(968,34)(969,35)(970,34)(971,35)(972,34)(973,33)(974,32)(975,31)(976,30)(977,29)(978,28)(979,29)(980,28)(981,27)(982,28)(983,27)(984,26)(985,27)(986,26)(987,25)(988,24)(989,25)(990,24)(991,23)(992,24)(993,23)(994,22)(995,21)(996,20)(997,19)(998,18)(999,19)(1000,18)(1001,17)(1002,16)(1003,17)(1004,16)(1005,15)(1006,14)(1007,13)(1008,12)(1009,13)(1010,12)(1011,11)(1012,10)(1013,9)(1014,10)(1015,9)(1016,8)(1017,7)(1018,6)(1019,5)(1020,4)(1021,3)(1022,2)(1023,1)(1024,0)
    };
       \end{axis}
\end{tikzpicture}
\ \ \ 
    \begin{tikzpicture} \footnotesize
\begin{axis}[
    width = 3.0in,
    xlabel={PCR2 sequence for $n=10$},
    ylabel={\# 1s $-$ \# 0s in prefix},
    xmin=0, xmax=1024,
    ymin=-25, ymax=125,
    ytick={0,25, 50,75,100},
    ytick pos=bottom,
    legend pos=north east,
    ymajorgrids=true,
    grid style=dashed,
]
          \addplot[
    color=red,
    ]
    coordinates {
(0,0)(1,1)(2,2)(3,3)(4,4)(5,5)(6,6)(7,7)(8,8)(9,9)(10,10)(11,9)(12,10)(13,11)(14,12)(15,13)(16,12)(17,13)(18,14)(19,15)(20,16)(21,17)(22,16)(23,17)(24,18)(25,19)(26,18)(27,19)(28,20)(29,21)(30,22)(31,21)(32,20)(33,21)(34,22)(35,23)(36,22)(37,23)(38,24)(39,25)(40,26)(41,27)(42,28)(43,27)(44,28)(45,29)(46,28)(47,29)(48,30)(49,31)(50,30)(51,31)(52,32)(53,31)(54,32)(55,33)(56,32)(57,33)(58,34)(59,35)(60,36)(61,35)(62,36)(63,35)(64,36)(65,37)(66,36)(67,37)(68,38)(69,39)(70,38)(71,37)(72,38)(73,37)(74,38)(75,39)(76,38)(77,39)(78,40)(79,41)(80,42)(81,43)(82,42)(83,41)(84,42)(85,43)(86,42)(87,43)(88,44)(89,45)(90,44)(91,45)(92,44)(93,43)(94,44)(95,45)(96,44)(97,45)(98,46)(99,47)(100,48)(101,47)(102,46)(103,45)(104,46)(105,47)(106,46)(107,47)(108,48)(109,49)(110,48)(111,47)(112,46)(113,45)(114,46)(115,47)(116,46)(117,47)(118,48)(119,49)(120,50)(121,51)(122,52)(123,53)(124,52)(125,53)(126,52)(127,53)(128,54)(129,55)(130,54)(131,55)(132,56)(133,57)(134,56)(135,57)(136,56)(137,57)(138,58)(139,59)(140,60)(141,59)(142,60)(143,61)(144,60)(145,61)(146,60)(147,61)(148,62)(149,61)(150,62)(151,61)(152,62)(153,63)(154,64)(155,63)(156,62)(157,63)(158,64)(159,63)(160,64)(161,63)(162,64)(163,65)(164,66)(165,67)(166,68)(167,67)(168,68)(169,67)(170,68)(171,67)(172,68)(173,69)(174,68)(175,69)(176,70)(177,69)(178,70)(179,69)(180,70)(181,69)(182,70)(183,71)(184,72)(185,71)(186,72)(187,71)(188,72)(189,71)(190,72)(191,71)(192,72)(193,71)(194,72)(195,73)(196,72)(197,71)(198,72)(199,71)(200,72)(201,71)(202,72)(203,71)(204,72)(205,73)(206,74)(207,75)(208,74)(209,73)(210,74)(211,73)(212,74)(213,73)(214,74)(215,75)(216,74)(217,75)(218,74)(219,73)(220,74)(221,73)(222,74)(223,73)(224,74)(225,75)(226,76)(227,75)(228,74)(229,73)(230,74)(231,73)(232,74)(233,73)(234,74)(235,75)(236,74)(237,73)(238,72)(239,71)(240,72)(241,71)(242,72)(243,71)(244,72)(245,73)(246,74)(247,75)(248,76)(249,77)(250,76)(251,75)(252,76)(253,75)(254,76)(255,77)(256,78)(257,77)(258,78)(259,79)(260,78)(261,77)(262,78)(263,77)(264,78)(265,79)(266,78)(267,77)(268,78)(269,79)(270,78)(271,77)(272,78)(273,77)(274,78)(275,79)(276,80)(277,81)(278,80)(279,81)(280,80)(281,79)(282,80)(283,79)(284,80)(285,81)(286,80)(287,81)(288,80)(289,81)(290,80)(291,79)(292,80)(293,79)(294,80)(295,81)(296,82)(297,81)(298,80)(299,81)(300,80)(301,79)(302,80)(303,79)(304,80)(305,81)(306,80)(307,79)(308,78)(309,79)(310,78)(311,77)(312,78)(313,77)(314,78)(315,79)(316,80)(317,81)(318,82)(319,81)(320,80)(321,79)(322,80)(323,79)(324,80)(325,81)(326,80)(327,81)(328,82)(329,81)(330,80)(331,79)(332,80)(333,79)(334,80)(335,81)(336,82)(337,81)(338,82)(339,81)(340,80)(341,79)(342,80)(343,79)(344,80)(345,81)(346,80)(347,79)(348,80)(349,79)(350,78)(351,77)(352,78)(353,77)(354,78)(355,79)(356,80)(357,81)(358,80)(359,79)(360,78)(361,77)(362,78)(363,77)(364,78)(365,79)(366,78)(367,79)(368,78)(369,77)(370,76)(371,75)(372,76)(373,75)(374,76)(375,77)(376,78)(377,77)(378,76)(379,75)(380,74)(381,73)(382,74)(383,73)(384,74)(385,75)(386,74)(387,73)(388,72)(389,71)(390,70)(391,69)(392,70)(393,69)(394,70)(395,71)(396,72)(397,73)(398,74)(399,75)(400,76)(401,77)(402,76)(403,75)(404,76)(405,77)(406,78)(407,79)(408,78)(409,79)(410,80)(411,81)(412,80)(413,79)(414,80)(415,81)(416,82)(417,81)(418,80)(419,81)(420,82)(421,83)(422,84)(423,85)(424,84)(425,85)(426,86)(427,85)(428,84)(429,85)(430,86)(431,85)(432,86)(433,87)(434,86)(435,87)(436,88)(437,87)(438,86)(439,87)(440,88)(441,89)(442,88)(443,89)(444,88)(445,89)(446,90)(447,89)(448,88)(449,89)(450,90)(451,91)(452,92)(453,91)(454,90)(455,91)(456,92)(457,91)(458,90)(459,91)(460,92)(461,91)(462,92)(463,91)(464,90)(465,91)(466,92)(467,91)(468,90)(469,91)(470,92)(471,93)(472,92)(473,91)(474,90)(475,91)(476,92)(477,91)(478,90)(479,91)(480,92)(481,93)(482,94)(483,95)(484,96)(485,95)(486,96)(487,95)(488,94)(489,95)(490,96)(491,97)(492,96)(493,97)(494,98)(495,97)(496,98)(497,97)(498,96)(499,97)(500,98)(501,99)(502,100)(503,99)(504,100)(505,99)(506,100)(507,99)(508,98)(509,99)(510,100)(511,99)(512,100)(513,99)(514,100)(515,99)(516,100)(517,99)(518,98)(519,99)(520,100)(521,101)(522,100)(523,99)(524,100)(525,99)(526,100)(527,99)(528,98)(529,99)(530,98)(531,99)(532,98)(533,97)(534,98)(535,99)(536,98)(537,97)(538,96)(539,97)(540,96)(541,97)(542,96)(543,95)(544,96)(545,97)(546,98)(547,99)(548,100)(549,99)(550,98)(551,99)(552,98)(553,97)(554,98)(555,99)(556,98)(557,99)(558,100)(559,99)(560,98)(561,99)(562,98)(563,97)(564,98)(565,99)(566,100)(567,99)(568,100)(569,99)(570,98)(571,99)(572,98)(573,97)(574,98)(575,97)(576,98)(577,97)(578,98)(579,97)(580,96)(581,97)(582,96)(583,95)(584,96)(585,97)(586,96)(587,95)(588,96)(589,95)(590,94)(591,95)(592,94)(593,93)(594,94)(595,95)(596,96)(597,97)(598,96)(599,95)(600,94)(601,95)(602,94)(603,93)(604,94)(605,95)(606,94)(607,95)(608,94)(609,93)(610,92)(611,93)(612,92)(613,91)(614,92)(615,93)(616,94)(617,93)(618,92)(619,91)(620,90)(621,91)(622,90)(623,89)(624,90)(625,89)(626,90)(627,89)(628,88)(629,87)(630,86)(631,87)(632,86)(633,85)(634,86)(635,87)(636,86)(637,85)(638,84)(639,83)(640,82)(641,83)(642,82)(643,81)(644,82)(645,83)(646,84)(647,85)(648,86)(649,87)(650,88)(651,87)(652,86)(653,85)(654,86)(655,87)(656,88)(657,87)(658,88)(659,89)(660,90)(661,89)(662,88)(663,87)(664,88)(665,89)(666,90)(667,91)(668,90)(669,91)(670,92)(671,91)(672,90)(673,89)(674,90)(675,91)(676,90)(677,91)(678,90)(679,91)(680,92)(681,91)(682,90)(683,89)(684,90)(685,91)(686,92)(687,91)(688,90)(689,91)(690,92)(691,91)(692,90)(693,89)(694,90)(695,91)(696,90)(697,89)(698,88)(699,89)(700,90)(701,91)(702,92)(703,93)(704,92)(705,93)(706,92)(707,91)(708,90)(709,91)(710,92)(711,91)(712,92)(713,93)(714,92)(715,93)(716,92)(717,91)(718,90)(719,91)(720,92)(721,93)(722,92)(723,93)(724,92)(725,93)(726,92)(727,91)(728,90)(729,91)(730,90)(731,91)(732,90)(733,91)(734,90)(735,91)(736,90)(737,89)(738,88)(739,89)(740,90)(741,89)(742,88)(743,89)(744,88)(745,89)(746,88)(747,87)(748,86)(749,87)(750,88)(751,89)(752,90)(753,89)(754,88)(755,89)(756,88)(757,87)(758,86)(759,87)(760,88)(761,87)(762,88)(763,87)(764,86)(765,87)(766,86)(767,85)(768,84)(769,85)(770,84)(771,83)(772,84)(773,83)(774,82)(775,83)(776,82)(777,81)(778,80)(779,81)(780,82)(781,83)(782,82)(783,81)(784,80)(785,81)(786,80)(787,79)(788,78)(789,79)(790,78)(791,79)(792,78)(793,77)(794,76)(795,77)(796,76)(797,75)(798,74)(799,75)(800,76)(801,75)(802,74)(803,73)(804,72)(805,73)(806,72)(807,71)(808,70)(809,71)(810,72)(811,73)(812,74)(813,75)(814,76)(815,75)(816,74)(817,73)(818,72)(819,73)(820,74)(821,75)(822,74)(823,75)(824,76)(825,75)(826,74)(827,73)(828,72)(829,73)(830,74)(831,73)(832,72)(833,73)(834,74)(835,73)(836,72)(837,71)(838,70)(839,71)(840,72)(841,73)(842,74)(843,73)(844,74)(845,73)(846,72)(847,71)(848,70)(849,71)(850,72)(851,71)(852,72)(853,71)(854,72)(855,71)(856,70)(857,69)(858,68)(859,69)(860,70)(861,71)(862,70)(863,69)(864,70)(865,69)(866,68)(867,67)(868,66)(869,67)(870,66)(871,67)(872,66)(873,65)(874,66)(875,65)(876,64)(877,63)(878,62)(879,63)(880,64)(881,63)(882,62)(883,61)(884,62)(885,61)(886,60)(887,59)(888,58)(889,59)(890,58)(891,57)(892,56)(893,55)(894,56)(895,57)(896,58)(897,59)(898,60)(899,59)(900,58)(901,57)(902,56)(903,55)(904,56)(905,57)(906,56)(907,57)(908,58)(909,57)(910,56)(911,55)(912,54)(913,53)(914,54)(915,55)(916,56)(917,55)(918,56)(919,55)(920,54)(921,53)(922,52)(923,51)(924,52)(925,51)(926,52)(927,51)(928,52)(929,51)(930,50)(931,49)(932,48)(933,47)(934,48)(935,49)(936,48)(937,47)(938,48)(939,47)(940,46)(941,45)(942,44)(943,43)(944,44)(945,43)(946,42)(947,41)(948,42)(949,41)(950,40)(951,39)(952,38)(953,37)(954,38)(955,39)(956,40)(957,41)(958,40)(959,39)(960,38)(961,37)(962,36)(963,35)(964,36)(965,37)(966,36)(967,37)(968,36)(969,35)(970,34)(971,33)(972,32)(973,31)(974,32)(975,31)(976,30)(977,31)(978,30)(979,29)(980,28)(981,27)(982,26)(983,25)(984,26)(985,27)(986,28)(987,27)(988,26)(989,25)(990,24)(991,23)(992,22)(993,21)(994,22)(995,21)(996,22)(997,21)(998,20)(999,19)(1000,18)(1001,17)(1002,16)(1003,15)(1004,16)(1005,17)(1006,16)(1007,15)(1008,14)(1009,13)(1010,12)(1011,11)(1012,10)(1013,9)(1014,10)(1015,9)(1016,8)(1017,7)(1018,6)(1019,5)(1020,4)(1021,3)(1022,2)(1023,1)(1024,0)
    };
      \end{axis}
\end{tikzpicture}

\vspace{0.2in}

    \begin{tikzpicture} \footnotesize
\begin{axis}[
    width = 3.0in,
    xlabel={PCR3 sequence for $n=10$},
    ylabel={\# 1s $-$ \# 0s in prefix},
    xmin=0, xmax=1024,
    ymin=-25, ymax=125,
    ytick={0,25, 50,75,100},
    legend pos=north east,
    ymajorgrids=true,
    grid style=dashed,
]

              \addplot[
    color=blue,
    ]
    coordinates {
(0,0)(1,-1)(2,-2)(3,-3)(4,-4)(5,-5)(6,-6)(7,-7)(8,-8)(9,-9)(10,-10)(11,-9)(12,-8)(13,-7)(14,-6)(15,-5)(16,-4)(17,-3)(18,-2)(19,-1)(20,0)(21,-1)(22,0)(23,1)(24,2)(25,3)(26,4)(27,5)(28,6)(29,7)(30,6)(31,5)(32,6)(33,7)(34,8)(35,9)(36,10)(37,11)(38,12)(39,11)(40,10)(41,9)(42,10)(43,11)(44,12)(45,13)(46,14)(47,15)(48,14)(49,15)(50,14)(51,13)(52,14)(53,15)(54,16)(55,17)(56,18)(57,19)(58,18)(59,17)(60,16)(61,15)(62,16)(63,17)(64,18)(65,19)(66,20)(67,19)(68,20)(69,21)(70,20)(71,19)(72,20)(73,21)(74,22)(75,23)(76,24)(77,23)(78,24)(79,23)(80,22)(81,21)(82,22)(83,23)(84,24)(85,25)(86,26)(87,25)(88,24)(89,23)(90,22)(91,21)(92,22)(93,23)(94,24)(95,25)(96,24)(97,25)(98,26)(99,27)(100,28)(101,27)(102,28)(103,29)(104,30)(105,29)(106,28)(107,29)(108,30)(109,31)(110,32)(111,31)(112,32)(113,33)(114,32)(115,31)(116,30)(117,31)(118,32)(119,33)(120,34)(121,33)(122,34)(123,33)(124,34)(125,33)(126,32)(127,33)(128,34)(129,35)(130,36)(131,35)(132,36)(133,35)(134,34)(135,33)(136,32)(137,33)(138,34)(139,35)(140,36)(141,35)(142,34)(143,35)(144,34)(145,33)(146,32)(147,33)(148,34)(149,35)(150,36)(151,35)(152,34)(153,33)(154,32)(155,31)(156,30)(157,31)(158,32)(159,33)(160,32)(161,33)(162,34)(163,35)(164,36)(165,37)(166,36)(167,37)(168,38)(169,39)(170,38)(171,39)(172,40)(173,41)(174,42)(175,41)(176,40)(177,41)(178,42)(179,43)(180,42)(181,43)(182,44)(183,45)(184,44)(185,43)(186,42)(187,43)(188,44)(189,45)(190,44)(191,45)(192,46)(193,45)(194,46)(195,45)(196,44)(197,45)(198,46)(199,47)(200,46)(201,47)(202,48)(203,47)(204,46)(205,45)(206,44)(207,45)(208,46)(209,47)(210,46)(211,47)(212,46)(213,47)(214,48)(215,47)(216,46)(217,47)(218,48)(219,49)(220,48)(221,49)(222,48)(223,49)(224,48)(225,47)(226,46)(227,47)(228,48)(229,49)(230,48)(231,49)(232,48)(233,47)(234,46)(235,45)(236,44)(237,45)(238,46)(239,47)(240,46)(241,45)(242,46)(243,47)(244,48)(245,47)(246,46)(247,47)(248,48)(249,47)(250,46)(251,45)(252,46)(253,47)(254,48)(255,47)(256,46)(257,47)(258,46)(259,45)(260,44)(261,43)(262,44)(263,45)(264,46)(265,45)(266,44)(267,43)(268,42)(269,41)(270,40)(271,39)(272,40)(273,41)(274,40)(275,41)(276,42)(277,43)(278,44)(279,45)(280,46)(281,45)(282,46)(283,47)(284,46)(285,47)(286,48)(287,49)(288,50)(289,51)(290,50)(291,49)(292,50)(293,51)(294,50)(295,51)(296,52)(297,53)(298,54)(299,53)(300,52)(301,51)(302,52)(303,53)(304,52)(305,53)(306,54)(307,55)(308,54)(309,55)(310,54)(311,53)(312,54)(313,55)(314,54)(315,55)(316,56)(317,57)(318,56)(319,55)(320,54)(321,53)(322,54)(323,55)(324,54)(325,55)(326,56)(327,55)(328,56)(329,57)(330,58)(331,57)(332,58)(333,59)(334,58)(335,59)(336,60)(337,59)(338,60)(339,61)(340,60)(341,59)(342,60)(343,61)(344,60)(345,61)(346,62)(347,61)(348,62)(349,61)(350,60)(351,59)(352,60)(353,61)(354,60)(355,61)(356,62)(357,61)(358,60)(359,59)(360,58)(361,57)(362,58)(363,59)(364,58)(365,59)(366,58)(367,59)(368,60)(369,61)(370,60)(371,59)(372,60)(373,61)(374,60)(375,61)(376,60)(377,61)(378,62)(379,61)(380,60)(381,59)(382,60)(383,61)(384,60)(385,61)(386,60)(387,61)(388,60)(389,61)(390,60)(391,59)(392,60)(393,61)(394,60)(395,61)(396,60)(397,61)(398,60)(399,59)(400,58)(401,57)(402,58)(403,59)(404,58)(405,59)(406,58)(407,57)(408,58)(409,57)(410,56)(411,55)(412,56)(413,57)(414,56)(415,57)(416,56)(417,55)(418,54)(419,53)(420,52)(421,51)(422,52)(423,53)(424,52)(425,51)(426,52)(427,53)(428,54)(429,55)(430,54)(431,53)(432,54)(433,55)(434,54)(435,53)(436,54)(437,55)(438,56)(439,55)(440,54)(441,53)(442,54)(443,55)(444,54)(445,53)(446,54)(447,55)(448,54)(449,55)(450,54)(451,53)(452,54)(453,55)(454,54)(455,53)(456,54)(457,55)(458,54)(459,53)(460,52)(461,51)(462,52)(463,53)(464,52)(465,51)(466,52)(467,51)(468,52)(469,51)(470,50)(471,49)(472,50)(473,51)(474,50)(475,49)(476,50)(477,49)(478,48)(479,47)(480,46)(481,45)(482,46)(483,47)(484,46)(485,45)(486,44)(487,45)(488,46)(489,45)(490,44)(491,43)(492,44)(493,43)(494,42)(495,41)(496,40)(497,41)(498,42)(499,41)(500,40)(501,39)(502,38)(503,37)(504,36)(505,35)(506,34)(507,35)(508,34)(509,35)(510,36)(511,37)(512,38)(513,39)(514,40)(515,41)(516,40)(517,41)(518,40)(519,41)(520,42)(521,43)(522,44)(523,45)(524,46)(525,45)(526,44)(527,45)(528,44)(529,45)(530,46)(531,47)(532,48)(533,49)(534,48)(535,47)(536,46)(537,47)(538,46)(539,47)(540,48)(541,49)(542,50)(543,49)(544,50)(545,51)(546,50)(547,51)(548,50)(549,51)(550,52)(551,53)(552,54)(553,53)(554,54)(555,53)(556,52)(557,53)(558,52)(559,53)(560,54)(561,55)(562,56)(563,55)(564,54)(565,53)(566,52)(567,53)(568,52)(569,53)(570,54)(571,55)(572,54)(573,55)(574,56)(575,57)(576,56)(577,57)(578,56)(579,57)(580,58)(581,59)(582,58)(583,59)(584,60)(585,59)(586,58)(587,59)(588,58)(589,59)(590,60)(591,61)(592,60)(593,61)(594,60)(595,59)(596,58)(597,59)(598,58)(599,59)(600,60)(601,61)(602,60)(603,59)(604,58)(605,57)(606,56)(607,57)(608,56)(609,57)(610,58)(611,57)(612,58)(613,59)(614,60)(615,61)(616,60)(617,61)(618,60)(619,61)(620,62)(621,61)(622,62)(623,63)(624,64)(625,63)(626,62)(627,63)(628,62)(629,63)(630,64)(631,63)(632,64)(633,65)(634,64)(635,63)(636,62)(637,63)(638,62)(639,63)(640,64)(641,63)(642,64)(643,63)(644,64)(645,65)(646,64)(647,65)(648,64)(649,65)(650,64)(651,63)(652,64)(653,63)(654,64)(655,65)(656,64)(657,65)(658,64)(659,63)(660,62)(661,61)(662,62)(663,61)(664,62)(665,63)(666,62)(667,61)(668,62)(669,63)(670,62)(671,61)(672,62)(673,61)(674,62)(675,63)(676,62)(677,61)(678,62)(679,61)(680,60)(681,59)(682,60)(683,59)(684,60)(685,61)(686,60)(687,59)(688,58)(689,57)(690,56)(691,55)(692,56)(693,55)(694,56)(695,55)(696,56)(697,57)(698,58)(699,59)(700,60)(701,59)(702,60)(703,59)(704,60)(705,59)(706,60)(707,61)(708,62)(709,63)(710,62)(711,61)(712,62)(713,61)(714,62)(715,61)(716,62)(717,63)(718,64)(719,63)(720,62)(721,61)(722,62)(723,61)(724,62)(725,61)(726,62)(727,63)(728,62)(729,63)(730,64)(731,63)(732,64)(733,63)(734,64)(735,63)(736,64)(737,65)(738,64)(739,65)(740,64)(741,63)(742,64)(743,63)(744,64)(745,63)(746,64)(747,65)(748,64)(749,63)(750,62)(751,61)(752,62)(753,61)(754,62)(755,61)(756,62)(757,61)(758,62)(759,63)(760,64)(761,63)(762,64)(763,63)(764,64)(765,63)(766,64)(767,63)(768,64)(769,65)(770,64)(771,63)(772,64)(773,63)(774,64)(775,63)(776,64)(777,63)(778,64)(779,63)(780,64)(781,63)(782,62)(783,61)(784,62)(785,61)(786,62)(787,61)(788,62)(789,61)(790,60)(791,59)(792,58)(793,57)(794,58)(795,57)(796,58)(797,57)(798,56)(799,57)(800,58)(801,59)(802,58)(803,57)(804,58)(805,57)(806,58)(807,57)(808,56)(809,57)(810,58)(811,57)(812,56)(813,55)(814,56)(815,55)(816,56)(817,55)(818,54)(819,55)(820,54)(821,55)(822,54)(823,53)(824,54)(825,53)(826,52)(827,51)(828,50)(829,51)(830,50)(831,51)(832,50)(833,49)(834,48)(835,47)(836,46)(837,45)(838,44)(839,45)(840,44)(841,43)(842,44)(843,45)(844,46)(845,47)(846,48)(847,47)(848,46)(849,47)(850,46)(851,45)(852,46)(853,47)(854,48)(855,49)(856,48)(857,47)(858,46)(859,47)(860,46)(861,45)(862,46)(863,47)(864,48)(865,47)(866,48)(867,47)(868,46)(869,47)(870,46)(871,45)(872,46)(873,47)(874,48)(875,47)(876,46)(877,45)(878,44)(879,45)(880,44)(881,43)(882,44)(883,45)(884,44)(885,45)(886,46)(887,45)(888,44)(889,45)(890,44)(891,43)(892,44)(893,45)(894,44)(895,45)(896,44)(897,43)(898,42)(899,43)(900,42)(901,41)(902,42)(903,43)(904,42)(905,41)(906,40)(907,39)(908,38)(909,39)(910,38)(911,37)(912,38)(913,37)(914,38)(915,39)(916,40)(917,39)(918,38)(919,39)(920,38)(921,37)(922,38)(923,37)(924,38)(925,39)(926,38)(927,37)(928,36)(929,37)(930,36)(931,35)(932,36)(933,35)(934,36)(935,35)(936,36)(937,35)(938,34)(939,35)(940,34)(941,33)(942,34)(943,33)(944,34)(945,33)(946,32)(947,31)(948,30)(949,31)(950,30)(951,29)(952,30)(953,29)(954,28)(955,29)(956,30)(957,29)(958,28)(959,29)(960,28)(961,27)(962,28)(963,27)(964,26)(965,27)(966,26)(967,25)(968,24)(969,25)(970,24)(971,23)(972,24)(973,23)(974,22)(975,21)(976,20)(977,19)(978,18)(979,19)(980,18)(981,17)(982,16)(983,17)(984,18)(985,19)(986,18)(987,17)(988,16)(989,17)(990,16)(991,15)(992,14)(993,15)(994,16)(995,15)(996,14)(997,13)(998,12)(999,13)(1000,12)(1001,11)(1002,10)(1003,11)(1004,10)(1005,11)(1006,10)(1007,9)(1008,8)(1009,9)(1010,8)(1011,7)(1012,6)(1013,7)(1014,6)(1015,5)(1016,4)(1017,3)(1018,2)(1019,3)(1020,2)(1021,1)(1022,0)(1023,-1)(1024,0)
    };
   \end{axis}
\end{tikzpicture}
\ \ \ 
 \begin{tikzpicture} \footnotesize
\begin{axis}[
    width = 3.0in,
    xlabel={PCR4 sequence for $n=10$},
    ylabel={\# 1s $-$ \# 0s in prefix},
    xmin=0, xmax=1024,
    ymin=-25, ymax=125,
    ytick={0,25, 50,75,100},
    legend pos=north east,
    ymajorgrids=true,
    grid style=dashed,
]

                  \addplot[
    color=red,
    ]
    coordinates {
(0,0)(1,-1)(2,-2)(3,-3)(4,-4)(5,-5)(6,-6)(7,-7)(8,-8)(9,-9)(10,-10)(11,-9)(12,-8)(13,-7)(14,-6)(15,-5)(16,-4)(17,-3)(18,-2)(19,-1)(20,0)(21,-1)(22,0)(23,1)(24,2)(25,3)(26,2)(27,3)(28,4)(29,3)(30,4)(31,3)(32,4)(33,5)(34,4)(35,5)(36,4)(37,5)(38,4)(39,3)(40,4)(41,3)(42,4)(43,3)(44,2)(45,3)(46,2)(47,3)(48,4)(49,3)(50,4)(51,3)(52,4)(53,5)(54,6)(55,5)(56,4)(57,5)(58,6)(59,5)(60,6)(61,5)(62,6)(63,7)(64,6)(65,5)(66,4)(67,5)(68,6)(69,5)(70,6)(71,5)(72,6)(73,5)(74,4)(75,3)(76,2)(77,3)(78,4)(79,3)(80,4)(81,3)(82,4)(83,5)(84,6)(85,7)(86,6)(87,7)(88,6)(89,5)(90,6)(91,5)(92,6)(93,7)(94,8)(95,7)(96,6)(97,7)(98,6)(99,5)(100,6)(101,5)(102,6)(103,7)(104,6)(105,5)(106,4)(107,5)(108,4)(109,3)(110,4)(111,3)(112,4)(113,3)(114,2)(115,1)(116,0)(117,1)(118,0)(119,-1)(120,0)(121,-1)(122,0)(123,1)(124,2)(125,3)(126,2)(127,3)(128,4)(129,5)(130,4)(131,3)(132,4)(133,5)(134,6)(135,5)(136,4)(137,5)(138,6)(139,5)(140,4)(141,3)(142,4)(143,5)(144,4)(145,3)(146,2)(147,3)(148,2)(149,1)(150,0)(151,-1)(152,0)(153,-1)(154,-2)(155,-3)(156,-4)(157,-3)(158,-2)(159,-3)(160,-4)(161,-5)(162,-4)(163,-3)(164,-2)(165,-3)(166,-4)(167,-3)(168,-4)(169,-5)(170,-6)(171,-7)(172,-6)(173,-5)(174,-4)(175,-5)(176,-6)(177,-5)(178,-4)(179,-3)(180,-2)(181,-3)(182,-2)(183,-1)(184,-2)(185,-3)(186,-4)(187,-3)(188,-2)(189,-1)(190,0)(191,-1)(192,0)(193,-1)(194,-2)(195,-3)(196,-4)(197,-3)(198,-2)(199,-1)(200,0)(201,-1)(202,0)(203,1)(204,2)(205,3)(206,4)(207,3)(208,4)(209,5)(210,4)(211,3)(212,4)(213,5)(214,6)(215,7)(216,6)(217,5)(218,6)(219,5)(220,4)(221,3)(222,4)(223,5)(224,6)(225,5)(226,4)(227,3)(228,4)(229,3)(230,2)(231,1)(232,2)(233,3)(234,2)(235,1)(236,0)(237,-1)(238,0)(239,-1)(240,-2)(241,-3)(242,-2)(243,-3)(244,-4)(245,-5)(246,-6)(247,-7)(248,-6)(249,-7)(250,-8)(251,-9)(252,-8)(253,-7)(254,-6)(255,-5)(256,-6)(257,-7)(258,-6)(259,-5)(260,-6)(261,-7)(262,-6)(263,-5)(264,-4)(265,-5)(266,-6)(267,-7)(268,-6)(269,-5)(270,-6)(271,-7)(272,-6)(273,-5)(274,-6)(275,-7)(276,-8)(277,-9)(278,-8)(279,-7)(280,-8)(281,-9)(282,-8)(283,-9)(284,-10)(285,-11)(286,-12)(287,-13)(288,-12)(289,-11)(290,-12)(291,-13)(292,-12)(293,-11)(294,-10)(295,-9)(296,-8)(297,-9)(298,-8)(299,-9)(300,-10)(301,-11)(302,-10)(303,-9)(304,-8)(305,-7)(306,-6)(307,-7)(308,-6)(309,-5)(310,-4)(311,-5)(312,-4)(313,-3)(314,-2)(315,-3)(316,-2)(317,-3)(318,-2)(319,-1)(320,-2)(321,-3)(322,-2)(323,-1)(324,-2)(325,-3)(326,-2)(327,-3)(328,-2)(329,-3)(330,-4)(331,-5)(332,-4)(333,-3)(334,-4)(335,-5)(336,-4)(337,-5)(338,-4)(339,-3)(340,-4)(341,-5)(342,-4)(343,-5)(344,-6)(345,-7)(346,-6)(347,-7)(348,-6)(349,-7)(350,-8)(351,-9)(352,-8)(353,-9)(354,-10)(355,-11)(356,-10)(357,-11)(358,-10)(359,-9)(360,-10)(361,-11)(362,-10)(363,-9)(364,-8)(365,-9)(366,-8)(367,-9)(368,-8)(369,-9)(370,-10)(371,-11)(372,-10)(373,-9)(374,-8)(375,-9)(376,-8)(377,-9)(378,-8)(379,-7)(380,-6)(381,-7)(382,-6)(383,-5)(384,-6)(385,-7)(386,-6)(387,-7)(388,-6)(389,-5)(390,-4)(391,-5)(392,-4)(393,-5)(394,-6)(395,-7)(396,-6)(397,-7)(398,-6)(399,-5)(400,-4)(401,-5)(402,-4)(403,-3)(404,-2)(405,-1)(406,-2)(407,-3)(408,-2)(409,-1)(410,0)(411,-1)(412,0)(413,1)(414,2)(415,1)(416,0)(417,-1)(418,0)(419,1)(420,2)(421,1)(422,2)(423,3)(424,2)(425,1)(426,0)(427,-1)(428,0)(429,1)(430,2)(431,1)(432,2)(433,1)(434,0)(435,-1)(436,-2)(437,-3)(438,-2)(439,-1)(440,0)(441,-1)(442,0)(443,1)(444,2)(445,3)(446,4)(447,5)(448,4)(449,5)(450,4)(451,3)(452,4)(453,5)(454,6)(455,7)(456,8)(457,7)(458,6)(459,7)(460,6)(461,5)(462,6)(463,7)(464,8)(465,9)(466,8)(467,7)(468,6)(469,7)(470,6)(471,5)(472,6)(473,7)(474,8)(475,7)(476,6)(477,5)(478,4)(479,5)(480,4)(481,3)(482,4)(483,5)(484,4)(485,3)(486,2)(487,1)(488,0)(489,1)(490,0)(491,-1)(492,0)(493,-1)(494,-2)(495,-3)(496,-4)(497,-5)(498,-6)(499,-5)(500,-6)(501,-7)(502,-6)(503,-5)(504,-4)(505,-3)(506,-2)(507,-1)(508,-2)(509,-1)(510,0)(511,-1)(512,0)(513,1)(514,2)(515,1)(516,2)(517,1)(518,0)(519,1)(520,2)(521,1)(522,2)(523,3)(524,2)(525,1)(526,2)(527,1)(528,0)(529,1)(530,2)(531,1)(532,2)(533,1)(534,0)(535,-1)(536,0)(537,-1)(538,-2)(539,-1)(540,0)(541,-1)(542,0)(543,1)(544,2)(545,1)(546,2)(547,3)(548,2)(549,3)(550,2)(551,1)(552,2)(553,3)(554,4)(555,3)(556,4)(557,3)(558,2)(559,3)(560,2)(561,1)(562,2)(563,3)(564,2)(565,1)(566,2)(567,1)(568,0)(569,1)(570,0)(571,-1)(572,0)(573,-1)(574,-2)(575,-3)(576,-2)(577,-3)(578,-4)(579,-3)(580,-4)(581,-5)(582,-4)(583,-3)(584,-2)(585,-3)(586,-2)(587,-1)(588,-2)(589,-1)(590,0)(591,-1)(592,0)(593,1)(594,0)(595,-1)(596,0)(597,1)(598,0)(599,1)(600,0)(601,-1)(602,0)(603,-1)(604,-2)(605,-3)(606,-2)(607,-1)(608,-2)(609,-1)(610,-2)(611,-3)(612,-2)(613,-1)(614,-2)(615,-3)(616,-2)(617,-1)(618,-2)(619,-1)(620,0)(621,-1)(622,0)(623,-1)(624,-2)(625,-3)(626,-2)(627,-1)(628,-2)(629,-1)(630,0)(631,-1)(632,0)(633,1)(634,2)(635,3)(636,2)(637,3)(638,2)(639,3)(640,2)(641,1)(642,2)(643,3)(644,4)(645,3)(646,2)(647,3)(648,2)(649,3)(650,2)(651,1)(652,2)(653,3)(654,2)(655,1)(656,0)(657,1)(658,0)(659,1)(660,0)(661,-1)(662,0)(663,-1)(664,-2)(665,-3)(666,-4)(667,-3)(668,-4)(669,-3)(670,-4)(671,-5)(672,-4)(673,-3)(674,-2)(675,-1)(676,-2)(677,-1)(678,-2)(679,-1)(680,0)(681,-1)(682,0)(683,1)(684,0)(685,1)(686,0)(687,1)(688,0)(689,1)(690,0)(691,-1)(692,0)(693,1)(694,0)(695,1)(696,0)(697,1)(698,0)(699,1)(700,2)(701,1)(702,2)(703,1)(704,0)(705,1)(706,0)(707,1)(708,0)(709,1)(710,0)(711,-1)(712,0)(713,-1)(714,-2)(715,-1)(716,-2)(717,-1)(718,-2)(719,-1)(720,0)(721,-1)(722,0)(723,1)(724,2)(725,1)(726,0)(727,1)(728,0)(729,1)(730,2)(731,1)(732,2)(733,3)(734,2)(735,1)(736,0)(737,1)(738,0)(739,1)(740,2)(741,1)(742,2)(743,1)(744,0)(745,-1)(746,-2)(747,-1)(748,-2)(749,-1)(750,0)(751,-1)(752,0)(753,1)(754,2)(755,3)(756,4)(757,3)(758,2)(759,3)(760,4)(761,3)(762,4)(763,5)(764,6)(765,7)(766,6)(767,5)(768,4)(769,5)(770,6)(771,5)(772,6)(773,7)(774,8)(775,7)(776,6)(777,5)(778,4)(779,5)(780,6)(781,5)(782,6)(783,7)(784,6)(785,5)(786,4)(787,3)(788,2)(789,3)(790,4)(791,3)(792,4)(793,3)(794,2)(795,1)(796,0)(797,-1)(798,-2)(799,-1)(800,0)(801,-1)(802,0)(803,1)(804,2)(805,3)(806,4)(807,5)(808,6)(809,5)(810,6)(811,5)(812,6)(813,7)(814,8)(815,9)(816,10)(817,9)(818,10)(819,9)(820,10)(821,9)(822,10)(823,11)(824,12)(825,11)(826,12)(827,11)(828,12)(829,11)(830,12)(831,11)(832,12)(833,11)(834,12)(835,13)(836,12)(837,11)(838,12)(839,11)(840,12)(841,11)(842,12)(843,11)(844,12)(845,11)(846,10)(847,9)(848,10)(849,9)(850,10)(851,9)(852,10)(853,9)(854,10)(855,11)(856,12)(857,13)(858,12)(859,11)(860,12)(861,11)(862,12)(863,11)(864,12)(865,13)(866,14)(867,13)(868,12)(869,11)(870,12)(871,11)(872,12)(873,11)(874,12)(875,13)(876,12)(877,11)(878,10)(879,9)(880,10)(881,9)(882,10)(883,9)(884,10)(885,9)(886,8)(887,7)(888,6)(889,5)(890,6)(891,5)(892,6)(893,5)(894,6)(895,7)(896,8)(897,9)(898,10)(899,11)(900,10)(901,9)(902,10)(903,9)(904,10)(905,11)(906,12)(907,13)(908,14)(909,13)(910,12)(911,11)(912,12)(913,11)(914,12)(915,13)(916,14)(917,15)(918,14)(919,13)(920,12)(921,11)(922,12)(923,11)(924,12)(925,13)(926,14)(927,13)(928,12)(929,11)(930,10)(931,9)(932,10)(933,9)(934,10)(935,11)(936,10)(937,9)(938,8)(939,7)(940,6)(941,5)(942,6)(943,5)(944,6)(945,5)(946,4)(947,3)(948,2)(949,1)(950,0)(951,-1)(952,0)(953,-1)(954,0)(955,1)(956,2)(957,3)(958,4)(959,5)(960,6)(961,7)(962,6)(963,5)(964,6)(965,7)(966,8)(967,9)(968,10)(969,11)(970,12)(971,11)(972,10)(973,9)(974,10)(975,11)(976,12)(977,13)(978,14)(979,15)(980,14)(981,13)(982,12)(983,11)(984,12)(985,13)(986,14)(987,15)(988,16)(989,15)(990,14)(991,13)(992,12)(993,11)(994,12)(995,13)(996,14)(997,15)(998,14)(999,13)(1000,12)(1001,11)(1002,10)(1003,9)(1004,10)(1005,11)(1006,12)(1007,11)(1008,10)(1009,9)(1010,8)(1011,7)(1012,6)(1013,5)(1014,6)(1015,7)(1016,6)(1017,5)(1018,4)(1019,3)(1020,2)(1021,1)(1022,0)(1023,-1)(1024,0)
    };
 
\end{axis}
\end{tikzpicture}

\end{center}

%% file: tables-range.tex
\begin{center}

\begin{tikzpicture} \footnotesize
\begin{axis}[
    width = 3.0in,
    xlabel={Cool-lex sequence for $n=10$},
    ylabel={\# 1s $-$ \# 0s in prefix},
    xmin=0, xmax=1024,
    ymin=-25, ymax=140,
    xtick={0,200,400,600,800,1000},
    ytick={0,25, 50,75, 100,125},
    legend pos=north east,
    ymajorgrids=true,
    grid style=dashed,
]

      \addplot[
    color=blue,
    ]
    coordinates {
(0,0)(1,-1)(2,0)(3,1)(4,2)(5,3)(6,2)(7,3)(8,4)(9,5)(10,6)(11,7)(12,6)(13,7)(14,8)(15,9)(16,8)(17,9)(18,10)(19,11)(20,12)(21,13)(22,14)(23,13)(24,14)(25,15)(26,14)(27,15)(28,16)(29,17)(30,18)(31,19)(32,20)(33,21)(34,20)(35,21)(36,20)(37,21)(38,22)(39,23)(40,24)(41,25)(42,26)(43,27)(44,28)(45,27)(46,26)(47,25)(48,26)(49,27)(50,28)(51,27)(52,28)(53,29)(54,30)(55,31)(56,30)(57,29)(58,28)(59,29)(60,30)(61,29)(62,30)(63,31)(64,32)(65,33)(66,34)(67,33)(68,32)(69,31)(70,32)(71,31)(72,32)(73,33)(74,34)(75,35)(76,36)(77,37)(78,36)(79,35)(80,36)(81,35)(82,36)(83,37)(84,38)(85,37)(86,38)(87,39)(88,40)(89,39)(90,38)(91,39)(92,38)(93,39)(94,40)(95,39)(96,40)(97,41)(98,42)(99,43)(100,42)(101,41)(102,42)(103,41)(104,42)(105,41)(106,42)(107,43)(108,44)(109,45)(110,46)(111,45)(112,44)(113,45)(114,46)(115,45)(116,46)(117,47)(118,46)(119,47)(120,48)(121,49)(122,48)(123,47)(124,48)(125,49)(126,48)(127,49)(128,48)(129,49)(130,50)(131,51)(132,52)(133,51)(134,50)(135,51)(136,52)(137,53)(138,52)(139,53)(140,52)(141,53)(142,54)(143,55)(144,54)(145,53)(146,54)(147,55)(148,56)(149,55)(150,54)(151,55)(152,56)(153,57)(154,58)(155,57)(156,56)(157,57)(158,58)(159,57)(160,56)(161,57)(162,58)(163,59)(164,60)(165,61)(166,60)(167,59)(168,60)(169,59)(170,58)(171,59)(172,60)(173,61)(174,62)(175,63)(176,64)(177,63)(178,64)(179,63)(180,62)(181,63)(182,64)(183,65)(184,64)(185,65)(186,66)(187,67)(188,66)(189,67)(190,66)(191,65)(192,66)(193,67)(194,66)(195,67)(196,68)(197,69)(198,70)(199,69)(200,70)(201,69)(202,68)(203,69)(204,68)(205,69)(206,70)(207,71)(208,72)(209,73)(210,72)(211,73)(212,72)(213,73)(214,72)(215,73)(216,74)(217,73)(218,74)(219,75)(220,76)(221,75)(222,76)(223,75)(224,76)(225,75)(226,76)(227,75)(228,76)(229,77)(230,78)(231,79)(232,78)(233,79)(234,78)(235,79)(236,80)(237,79)(238,80)(239,81)(240,80)(241,81)(242,82)(243,81)(244,82)(245,81)(246,82)(247,83)(248,82)(249,83)(250,82)(251,83)(252,84)(253,85)(254,84)(255,85)(256,84)(257,85)(258,86)(259,85)(260,84)(261,85)(262,86)(263,87)(264,88)(265,87)(266,88)(267,87)(268,88)(269,87)(270,86)(271,87)(272,88)(273,89)(274,90)(275,91)(276,90)(277,91)(278,92)(279,91)(280,90)(281,91)(282,92)(283,91)(284,92)(285,93)(286,94)(287,93)(288,94)(289,95)(290,94)(291,93)(292,94)(293,93)(294,94)(295,95)(296,96)(297,97)(298,96)(299,97)(300,98)(301,97)(302,98)(303,97)(304,98)(305,97)(306,98)(307,99)(308,100)(309,99)(310,100)(311,101)(312,100)(313,101)(314,102)(315,101)(316,100)(317,101)(318,102)(319,103)(320,102)(321,103)(322,104)(323,103)(324,104)(325,103)(326,102)(327,103)(328,104)(329,105)(330,106)(331,105)(332,106)(333,107)(334,108)(335,107)(336,106)(337,105)(338,106)(339,107)(340,108)(341,109)(342,108)(343,109)(344,110)(345,109)(346,108)(347,107)(348,108)(349,109)(350,110)(351,111)(352,112)(353,111)(354,112)(355,111)(356,110)(357,109)(358,110)(359,111)(360,112)(361,113)(362,114)(363,115)(364,114)(365,113)(366,112)(367,111)(368,110)(369,111)(370,112)(371,111)(372,112)(373,113)(374,114)(375,113)(376,112)(377,111)(378,110)(379,109)(380,110)(381,109)(382,110)(383,111)(384,112)(385,113)(386,112)(387,111)(388,110)(389,109)(390,110)(391,109)(392,110)(393,111)(394,110)(395,111)(396,112)(397,111)(398,110)(399,109)(400,108)(401,109)(402,108)(403,109)(404,108)(405,109)(406,110)(407,111)(408,110)(409,109)(410,108)(411,107)(412,108)(413,109)(414,108)(415,109)(416,108)(417,109)(418,110)(419,109)(420,108)(421,107)(422,106)(423,107)(424,108)(425,107)(426,106)(427,107)(428,108)(429,109)(430,108)(431,107)(432,106)(433,105)(434,106)(435,105)(436,104)(437,105)(438,106)(439,107)(440,108)(441,107)(442,106)(443,105)(444,106)(445,105)(446,104)(447,105)(448,106)(449,105)(450,106)(451,107)(452,106)(453,105)(454,104)(455,105)(456,104)(457,103)(458,104)(459,103)(460,104)(461,105)(462,106)(463,105)(464,104)(465,103)(466,104)(467,103)(468,104)(469,103)(470,104)(471,103)(472,104)(473,105)(474,104)(475,103)(476,102)(477,103)(478,102)(479,103)(480,102)(481,101)(482,102)(483,103)(484,104)(485,103)(486,102)(487,101)(488,102)(489,103)(490,102)(491,101)(492,102)(493,101)(494,102)(495,103)(496,102)(497,101)(498,100)(499,101)(500,102)(501,101)(502,102)(503,101)(504,100)(505,101)(506,102)(507,101)(508,100)(509,99)(510,100)(511,101)(512,100)(513,99)(514,98)(515,99)(516,100)(517,101)(518,100)(519,99)(520,98)(521,99)(522,98)(523,97)(524,96)(525,97)(526,98)(527,99)(528,100)(529,99)(530,98)(531,99)(532,98)(533,97)(534,96)(535,97)(536,98)(537,97)(538,98)(539,99)(540,98)(541,97)(542,98)(543,97)(544,96)(545,95)(546,96)(547,95)(548,96)(549,97)(550,98)(551,97)(552,96)(553,97)(554,96)(555,95)(556,96)(557,95)(558,96)(559,95)(560,96)(561,97)(562,96)(563,95)(564,96)(565,95)(566,94)(567,95)(568,96)(569,95)(570,94)(571,95)(572,96)(573,95)(574,94)(575,95)(576,94)(577,93)(578,94)(579,93)(580,92)(581,93)(582,94)(583,95)(584,94)(585,93)(586,94)(587,93)(588,94)(589,93)(590,92)(591,93)(592,92)(593,93)(594,94)(595,93)(596,92)(597,93)(598,92)(599,93)(600,92)(601,93)(602,92)(603,93)(604,92)(605,93)(606,92)(607,91)(608,92)(609,91)(610,92)(611,91)(612,92)(613,91)(614,90)(615,91)(616,92)(617,91)(618,90)(619,91)(620,90)(621,91)(622,90)(623,89)(624,88)(625,89)(626,90)(627,91)(628,90)(629,89)(630,90)(631,91)(632,90)(633,89)(634,88)(635,89)(636,88)(637,89)(638,90)(639,89)(640,88)(641,89)(642,90)(643,89)(644,88)(645,87)(646,86)(647,87)(648,88)(649,89)(650,88)(651,87)(652,88)(653,87)(654,86)(655,85)(656,84)(657,85)(658,86)(659,87)(660,88)(661,87)(662,88)(663,87)(664,86)(665,85)(666,84)(667,85)(668,86)(669,85)(670,86)(671,87)(672,86)(673,87)(674,86)(675,85)(676,84)(677,83)(678,84)(679,83)(680,84)(681,85)(682,86)(683,85)(684,86)(685,85)(686,84)(687,83)(688,84)(689,83)(690,84)(691,83)(692,84)(693,85)(694,84)(695,85)(696,84)(697,83)(698,82)(699,83)(700,84)(701,83)(702,82)(703,83)(704,84)(705,83)(706,84)(707,83)(708,82)(709,81)(710,82)(711,81)(712,80)(713,81)(714,82)(715,83)(716,82)(717,83)(718,82)(719,81)(720,82)(721,81)(722,80)(723,81)(724,80)(725,81)(726,82)(727,81)(728,82)(729,81)(730,80)(731,81)(732,80)(733,81)(734,80)(735,79)(736,80)(737,81)(738,80)(739,81)(740,80)(741,79)(742,80)(743,79)(744,78)(745,77)(746,78)(747,79)(748,80)(749,79)(750,80)(751,79)(752,80)(753,79)(754,78)(755,77)(756,78)(757,77)(758,78)(759,79)(760,78)(761,79)(762,78)(763,79)(764,78)(765,77)(766,78)(767,77)(768,76)(769,77)(770,78)(771,77)(772,78)(773,77)(774,78)(775,77)(776,78)(777,77)(778,76)(779,75)(780,76)(781,77)(782,76)(783,77)(784,76)(785,77)(786,76)(787,75)(788,74)(789,73)(790,74)(791,75)(792,76)(793,75)(794,76)(795,77)(796,76)(797,75)(798,74)(799,73)(800,72)(801,73)(802,74)(803,75)(804,74)(805,75)(806,74)(807,73)(808,72)(809,71)(810,70)(811,71)(812,72)(813,73)(814,74)(815,73)(816,72)(817,71)(818,70)(819,69)(820,68)(821,67)(822,68)(823,67)(824,68)(825,69)(826,68)(827,67)(828,66)(829,65)(830,64)(831,63)(832,64)(833,63)(834,64)(835,63)(836,64)(837,63)(838,62)(839,61)(840,60)(841,59)(842,58)(843,59)(844,58)(845,57)(846,58)(847,59)(848,58)(849,57)(850,56)(851,55)(852,54)(853,55)(854,54)(855,53)(856,54)(857,53)(858,54)(859,53)(860,52)(861,51)(862,50)(863,49)(864,50)(865,49)(866,48)(867,47)(868,48)(869,49)(870,48)(871,47)(872,46)(873,45)(874,46)(875,45)(876,44)(877,43)(878,44)(879,43)(880,44)(881,43)(882,42)(883,41)(884,40)(885,41)(886,40)(887,39)(888,40)(889,39)(890,38)(891,39)(892,38)(893,37)(894,36)(895,35)(896,36)(897,35)(898,34)(899,33)(900,32)(901,33)(902,34)(903,33)(904,32)(905,31)(906,32)(907,31)(908,30)(909,29)(910,28)(911,29)(912,28)(913,29)(914,28)(915,27)(916,26)(917,27)(918,26)(919,25)(920,24)(921,25)(922,24)(923,23)(924,24)(925,23)(926,22)(927,21)(928,22)(929,21)(930,20)(931,19)(932,18)(933,17)(934,18)(935,19)(936,18)(937,17)(938,18)(939,17)(940,16)(941,15)(942,14)(943,13)(944,14)(945,13)(946,14)(947,13)(948,12)(949,13)(950,12)(951,11)(952,10)(953,9)(954,8)(955,7)(956,8)(957,9)(958,8)(959,9)(960,8)(961,7)(962,6)(963,5)(964,4)(965,3)(966,2)(967,3)(968,4)(969,3)(970,2)(971,1)(972,0)(973,-1)(974,-2)(975,-3)(976,-4)(977,-5)(978,-6)(979,-5)(980,-6)(981,-7)(982,-8)(983,-9)(984,-10)(985,-11)(986,-12)(987,-13)(988,-12)(989,-11)(990,-10)(991,-11)(992,-12)(993,-13)(994,-14)(995,-15)(996,-16)(997,-15)(998,-14)(999,-13)(1000,-12)(1001,-11)(1002,-12)(1003,-13)(1004,-14)(1005,-15)(1006,-14)(1007,-13)(1008,-12)(1009,-11)(1010,-10)(1011,-9)(1012,-8)(1013,-9)(1014,-10)(1015,-9)(1016,-8)(1017,-7)(1018,-6)(1019,-5)(1020,-4)(1021,-3)(1022,-2)(1023,-1)(1024,0)
    };
    \end{axis}
\end{tikzpicture}
%
   \ \ \  
    \begin{tikzpicture} \footnotesize
\begin{axis}[
    width = 3.0in,
    xlabel={Weight-range sequence for $n=10$},
    ylabel={\# 1s $-$ \# 0s in prefix},
    xmin=0, xmax=1024,
    ymin=-25, ymax=140,
   xtick={0,200,400,600,800,1000},
    ytick={0,25, 50,75, 100,125},
    legend pos=north east,
    ymajorgrids=true,
    grid style=dashed,
]
          \addplot[
    color=red,
    ]
    coordinates {
(0,0)(1,1)(2,2)(3,3)(4,4)(5,5)(6,4)(7,3)(8,2)(9,1)(10,2)(11,3)(12,4)(13,5)(14,6)(15,7)(16,6)(17,5)(18,4)(19,5)(20,4)(21,5)(22,6)(23,7)(24,8)(25,9)(26,8)(27,7)(28,6)(29,7)(30,8)(31,7)(32,8)(33,9)(34,10)(35,11)(36,10)(37,9)(38,8)(39,9)(40,10)(41,11)(42,10)(43,11)(44,12)(45,13)(46,12)(47,11)(48,10)(49,11)(50,12)(51,13)(52,14)(53,13)(54,14)(55,15)(56,14)(57,13)(58,12)(59,13)(60,14)(61,15)(62,16)(63,17)(64,16)(65,17)(66,16)(67,15)(68,14)(69,15)(70,16)(71,17)(72,18)(73,19)(74,20)(75,21)(76,20)(77,19)(78,20)(79,19)(80,18)(81,19)(82,20)(83,21)(84,22)(85,23)(86,22)(87,21)(88,22)(89,21)(90,22)(91,21)(92,22)(93,23)(94,24)(95,25)(96,24)(97,23)(98,24)(99,23)(100,24)(101,25)(102,24)(103,25)(104,26)(105,27)(106,26)(107,25)(108,26)(109,25)(110,26)(111,27)(112,28)(113,27)(114,28)(115,29)(116,28)(117,27)(118,28)(119,27)(120,28)(121,29)(122,30)(123,31)(124,30)(125,31)(126,30)(127,29)(128,30)(129,29)(130,30)(131,31)(132,32)(133,33)(134,34)(135,35)(136,34)(137,33)(138,34)(139,35)(140,34)(141,33)(142,34)(143,35)(144,36)(145,37)(146,36)(147,35)(148,36)(149,37)(150,36)(151,37)(152,36)(153,37)(154,38)(155,39)(156,38)(157,37)(158,38)(159,39)(160,38)(161,39)(162,40)(163,39)(164,40)(165,41)(166,40)(167,39)(168,40)(169,41)(170,40)(171,41)(172,42)(173,43)(174,42)(175,43)(176,42)(177,41)(178,42)(179,43)(180,42)(181,43)(182,44)(183,45)(184,46)(185,47)(186,46)(187,45)(188,46)(189,47)(190,48)(191,47)(192,46)(193,47)(194,48)(195,49)(196,48)(197,49)(198,48)(199,49)(200,50)(201,49)(202,48)(203,49)(204,50)(205,51)(206,50)(207,51)(208,52)(209,51)(210,52)(211,51)(212,50)(213,51)(214,52)(215,53)(216,52)(217,53)(218,54)(219,55)(220,56)(221,55)(222,54)(223,55)(224,56)(225,57)(226,58)(227,57)(228,58)(229,57)(230,58)(231,57)(232,56)(233,57)(234,58)(235,59)(236,60)(237,59)(238,60)(239,61)(240,62)(241,61)(242,60)(243,61)(244,62)(245,63)(246,64)(247,65)(248,64)(249,65)(250,66)(251,65)(252,64)(253,65)(254,66)(255,67)(256,68)(257,69)(258,70)(259,69)(260,70)(261,69)(262,68)(263,69)(264,70)(265,71)(266,72)(267,73)(268,74)(269,75)(270,76)(271,75)(272,76)(273,75)(274,76)(275,75)(276,76)(277,75)(278,76)(279,77)(280,78)(281,77)(282,78)(283,77)(284,78)(285,77)(286,78)(287,79)(288,78)(289,79)(290,80)(291,79)(292,80)(293,79)(294,80)(295,79)(296,80)(297,81)(298,82)(299,83)(300,84)(301,83)(302,84)(303,83)(304,84)(305,85)(306,84)(307,85)(308,84)(309,85)(310,86)(311,85)(312,86)(313,87)(314,88)(315,89)(316,88)(317,89)(318,88)(319,89)(320,90)(321,91)(322,90)(323,91)(324,92)(325,93)(326,92)(327,93)(328,92)(329,93)(330,94)(331,95)(332,96)(333,95)(334,96)(335,97)(336,96)(337,97)(338,96)(339,97)(340,98)(341,99)(342,100)(343,101)(344,102)(345,103)(346,102)(347,103)(348,104)(349,103)(350,104)(351,105)(352,104)(353,105)(354,106)(355,107)(356,106)(357,107)(358,108)(359,107)(360,108)(361,109)(362,110)(363,111)(364,112)(365,113)(366,112)(367,113)(368,114)(369,115)(370,114)(371,115)(372,116)(373,117)(374,118)(375,119)(376,118)(377,119)(378,120)(379,121)(380,122)(381,121)(382,122)(383,123)(384,124)(385,125)(386,126)(387,127)(388,128)(389,129)(390,130)(391,131)(392,130)(393,129)(394,128)(395,127)(396,126)(397,127)(398,128)(399,129)(400,130)(401,129)(402,130)(403,129)(404,128)(405,127)(406,126)(407,127)(408,128)(409,129)(410,130)(411,129)(412,128)(413,129)(414,128)(415,127)(416,126)(417,127)(418,128)(419,129)(420,130)(421,129)(422,128)(423,127)(424,128)(425,127)(426,126)(427,127)(428,128)(429,129)(430,130)(431,129)(432,128)(433,127)(434,126)(435,127)(436,126)(437,127)(438,128)(439,129)(440,130)(441,129)(442,128)(443,127)(444,126)(445,125)(446,124)(447,125)(448,126)(449,127)(450,126)(451,127)(452,128)(453,127)(454,126)(455,125)(456,124)(457,125)(458,126)(459,127)(460,126)(461,127)(462,126)(463,127)(464,126)(465,125)(466,124)(467,125)(468,126)(469,127)(470,126)(471,127)(472,126)(473,125)(474,126)(475,125)(476,124)(477,125)(478,126)(479,127)(480,126)(481,127)(482,126)(483,125)(484,124)(485,125)(486,124)(487,125)(488,126)(489,127)(490,126)(491,127)(492,126)(493,125)(494,124)(495,123)(496,122)(497,123)(498,124)(499,125)(500,124)(501,123)(502,124)(503,125)(504,124)(505,123)(506,122)(507,123)(508,124)(509,125)(510,124)(511,123)(512,124)(513,123)(514,124)(515,123)(516,122)(517,123)(518,124)(519,125)(520,124)(521,123)(522,124)(523,123)(524,122)(525,123)(526,122)(527,123)(528,124)(529,125)(530,124)(531,123)(532,124)(533,123)(534,122)(535,121)(536,120)(537,121)(538,122)(539,123)(540,122)(541,121)(542,120)(543,121)(544,122)(545,121)(546,120)(547,121)(548,122)(549,123)(550,122)(551,121)(552,120)(553,121)(554,120)(555,121)(556,120)(557,121)(558,122)(559,123)(560,122)(561,121)(562,120)(563,121)(564,120)(565,119)(566,118)(567,119)(568,120)(569,121)(570,120)(571,119)(572,118)(573,117)(574,118)(575,119)(576,118)(577,119)(578,120)(579,121)(580,120)(581,119)(582,118)(583,117)(584,118)(585,117)(586,116)(587,117)(588,118)(589,119)(590,118)(591,117)(592,116)(593,115)(594,114)(595,115)(596,114)(597,115)(598,116)(599,117)(600,116)(601,115)(602,114)(603,113)(604,112)(605,111)(606,110)(607,111)(608,112)(609,111)(610,112)(611,113)(612,112)(613,113)(614,112)(615,111)(616,110)(617,111)(618,112)(619,111)(620,112)(621,113)(622,112)(623,111)(624,112)(625,111)(626,110)(627,111)(628,112)(629,111)(630,112)(631,113)(632,112)(633,111)(634,110)(635,111)(636,110)(637,111)(638,112)(639,111)(640,112)(641,113)(642,112)(643,111)(644,110)(645,109)(646,108)(647,109)(648,110)(649,109)(650,110)(651,109)(652,110)(653,111)(654,110)(655,109)(656,108)(657,109)(658,110)(659,109)(660,110)(661,109)(662,110)(663,109)(664,110)(665,109)(666,108)(667,109)(668,110)(669,109)(670,110)(671,109)(672,110)(673,109)(674,108)(675,109)(676,108)(677,109)(678,110)(679,109)(680,110)(681,109)(682,110)(683,109)(684,108)(685,107)(686,106)(687,107)(688,108)(689,107)(690,108)(691,107)(692,106)(693,107)(694,108)(695,107)(696,106)(697,107)(698,108)(699,107)(700,108)(701,107)(702,106)(703,107)(704,106)(705,107)(706,106)(707,107)(708,108)(709,107)(710,108)(711,107)(712,106)(713,107)(714,106)(715,105)(716,104)(717,105)(718,106)(719,105)(720,106)(721,105)(722,104)(723,103)(724,104)(725,103)(726,102)(727,103)(728,104)(729,103)(730,104)(731,103)(732,102)(733,101)(734,100)(735,101)(736,100)(737,101)(738,102)(739,101)(740,102)(741,101)(742,100)(743,99)(744,98)(745,97)(746,96)(747,97)(748,98)(749,97)(750,96)(751,97)(752,98)(753,97)(754,96)(755,97)(756,96)(757,97)(758,98)(759,97)(760,96)(761,97)(762,98)(763,97)(764,96)(765,95)(766,94)(767,95)(768,96)(769,95)(770,94)(771,95)(772,94)(773,95)(774,94)(775,95)(776,94)(777,95)(778,96)(779,95)(780,94)(781,95)(782,94)(783,95)(784,94)(785,93)(786,92)(787,93)(788,94)(789,93)(790,92)(791,93)(792,92)(793,91)(794,92)(795,91)(796,90)(797,91)(798,92)(799,91)(800,90)(801,91)(802,90)(803,89)(804,88)(805,89)(806,88)(807,89)(808,90)(809,89)(810,88)(811,89)(812,88)(813,87)(814,86)(815,85)(816,84)(817,85)(818,86)(819,85)(820,84)(821,83)(822,84)(823,85)(824,84)(825,83)(826,82)(827,83)(828,82)(829,83)(830,82)(831,81)(832,82)(833,83)(834,82)(835,81)(836,80)(837,81)(838,80)(839,79)(840,80)(841,79)(842,80)(843,81)(844,80)(845,79)(846,78)(847,79)(848,78)(849,77)(850,76)(851,75)(852,76)(853,77)(854,76)(855,75)(856,74)(857,73)(858,74)(859,73)(860,74)(861,73)(862,74)(863,75)(864,74)(865,73)(866,72)(867,71)(868,72)(869,71)(870,70)(871,69)(872,70)(873,71)(874,70)(875,69)(876,68)(877,67)(878,66)(879,67)(880,66)(881,65)(882,66)(883,67)(884,66)(885,65)(886,64)(887,63)(888,62)(889,61)(890,62)(891,61)(892,62)(893,63)(894,62)(895,61)(896,60)(897,59)(898,58)(899,57)(900,56)(901,55)(902,56)(903,55)(904,56)(905,55)(906,56)(907,55)(908,56)(909,55)(910,56)(911,55)(912,54)(913,53)(914,54)(915,53)(916,54)(917,53)(918,54)(919,53)(920,52)(921,53)(922,52)(923,51)(924,52)(925,51)(926,52)(927,51)(928,52)(929,51)(930,50)(931,49)(932,48)(933,47)(934,48)(935,47)(936,48)(937,47)(938,46)(939,47)(940,46)(941,47)(942,46)(943,45)(944,46)(945,45)(946,44)(947,43)(948,42)(949,43)(950,42)(951,43)(952,42)(953,41)(954,40)(955,41)(956,40)(957,39)(958,38)(959,39)(960,38)(961,39)(962,38)(963,37)(964,36)(965,35)(966,36)(967,35)(968,34)(969,35)(970,34)(971,35)(972,34)(973,33)(974,32)(975,31)(976,30)(977,29)(978,28)(979,29)(980,28)(981,27)(982,28)(983,27)(984,26)(985,27)(986,26)(987,25)(988,24)(989,25)(990,24)(991,23)(992,24)(993,23)(994,22)(995,21)(996,20)(997,19)(998,18)(999,19)(1000,18)(1001,17)(1002,16)(1003,17)(1004,16)(1005,15)(1006,14)(1007,13)(1008,12)(1009,13)(1010,12)(1011,11)(1012,10)(1013,9)(1014,10)(1015,9)(1016,8)(1017,7)(1018,6)(1019,5)(1020,4)(1021,3)(1022,2)(1023,1)(1024,0)
    };
 
\end{axis}
\end{tikzpicture}

\end{center}